\documentclass[11pt,a4paper,reqno]{amsart}
\usepackage{amsmath}
\usepackage[english]{babel}
\usepackage{latexsym}
\usepackage{amssymb}
\usepackage{amscd}
\usepackage{amsgen,amstext,amsbsy,amsopn,amsfonts}
\usepackage{math rsfs}
\usepackage{bm,bbm}
\usepackage{amsthm,epsfig,graphicx,graphics}
\usepackage[latin1]{inputenc}
\usepackage{xspace}
\usepackage{amsxtra}
\usepackage{color}
\usepackage{dsfont}
\usepackage{enumerate}
\usepackage{tikz}
\usepackage{hyperref}
\usetikzlibrary{angles}
\usetikzlibrary{quotes}
\usetikzlibrary{patterns}
\usepackage[font=footnotesize]{caption}
\usepackage[capitalise,nameinlink]{cleveref}
\crefname{section}{\textsection}{\textsection}
\crefname{subsection}{\textsection}{\textsection}
\crefname{subsubsection}{\textsection}{\textsection}
\crefname{paragraph}{\textparagraph}{\textparagraph}

\makeatletter
\renewcommand*{\@textcolor}[3]{%
  \protect\leavevmode
  \begingroup
    \color#1{#2}#3%
  \endgroup
}
\makeatother

\usepackage{hyperref}

\usepackage{geometry}
\DeclareMathAlphabet{\mathpzc}{OT1}{pzc}{m}{it}

\makeatletter
\renewcommand*{\@textcolor}[3]{%
  \protect\leavevmode
  \begingroup
    \color#1{#2}#3%
  \endgroup
}
\makeatother

\numberwithin{equation}{section}
\newcommand{\bdm}{\begin{displaymath}}
\newcommand{\edm}{\end{displaymath}}
\newcommand{\bay}{\begin{array}{c}}
\newcommand{\eay}{\end{array}}
\newcommand{\ben}{\begin{enumerate}}
\newcommand{\een}{\end{enumerate}}
\newcommand{\beq}{\begin{equation}}
\newcommand{\eeq}{\end{equation}}
\newcommand{\beqn}{\begin{eqnarray}}
\newcommand{\eeqn}{\end{eqnarray}}
\newcommand{\bml}[1]{\begin{multline} #1 \end{multline}}
\newcommand{\bmln}[1]{\begin{multline*} #1 \end{multline*}}

\newcommand{\lf}{\left}
\newcommand{\ri}{\right}

\newcommand{\xv}{\mathbf{x}}

\newcommand{\rv}{\mathbf{r}}
\newcommand{\rrv}{\mathbf{R}}

\newcommand{\av}{\mathbf{a}}
\newcommand{\nv}{\mathbf{n}}
\newcommand{\fv}{\mathbf{F}}
\newcommand{\ev}{\mathbf{e}}
\newcommand{\gamv}{\bm{\gamma}}

\newcommand{\deps}{\delta_{\eps}}

\newcommand{\diff}{\mathrm{d}}
\newcommand{\eps}{\varepsilon}

\newcommand{\dist}{\mathrm{dist}}

\newcommand{\as}{\alpha_{\star}}

\newcommand{\nuv}{\bm{\nu}}
\newcommand{\tav}{\bm{\tau}}
\newcommand{\phiv}{\bm{\Phi}}
\renewcommand{\ss}{\mathsf{s}}
\renewcommand{\tt}{\mathsf{t}}

\newcommand{\glf}{\mathcal{E}^{\mathrm{GL}}}
\newcommand{\gle}{E^{\mathrm{GL}}}
\newcommand{\glm}{\psi^{\mathrm{GL}}}

\newcommand{\gldom}{\mathscr{D}^{\mathrm{GL}}}
\newcommand{\aav}{\mathbf{A}}

\newcommand{\aavm}{\mathbf{A}^{\mathrm{GL}}}
\newcommand{\aeps}{a_{\eps}}
\newcommand{\aave}{\av_{\eps}}

\newcommand{\psid}{\psi_{\mathrm{D}}}

\newcommand{\psin}{\psi_{\mathrm{N}}}

\newcommand{\hex}{h_{\mathrm{ex}}}
\newcommand{\theo}{\Theta_0}

\newcommand{\glfk}{\mathcal{G}_{\kappa}^{\mathrm{GL}}}

\newcommand{\glfe}{\mathcal{E}_{\eps}^{\mathrm{GL}}}
\newcommand{\gep}{\mathcal{G}_{\eps}}

\newcommand{\glee}{E_{\eps}^{\mathrm{GL}}}

\newcommand{\curv}{\mathfrak{K}(\ss)}
\newcommand{\dom}{\mathscr{D}}
\newcommand{\domd}{\mathscr{D}_{\mathrm{D}}}
\newcommand{\domk}{\mathscr{D}_{\mathrm{D},\varkappa}}
\newcommand{\doms}{\mathscr{D}_{\star}}
\newcommand{\domdx}{\mathscr{D}_{\mathrm{D},\varkappa}}
\newcommand{\domn}{\mathscr{D}_{\mathrm{N}}}

\newcommand{\eones}{E^{\mathrm{1D}}_{\star}}

\newcommand{\fol}{f_{0}}
\newcommand{\Fol}{F_{0}}

\newcommand{\fs}{f_{\star}}

\newcommand{\onedom}{\mathscr{D}^{\mathrm{1D}}}
\newcommand{\onedomk}{\mathscr{D}_{\ell}^{\mathrm{1D}}}

\newcommand{\curl}{\mathrm{curl}}

\newcommand{\ann}{\mathcal{A}}
\newcommand{\anne}{\mathcal{A}_{\eps}}

\renewcommand{\acute}{\mathcal{A}_{\mathrm{cut},\eps}}
\newcommand{\acutet}{\widetilde{\mathcal{A}}_{\mathrm{cut},\eps}}
\newcommand{\acut}{\mathcal{A}_{\mathrm{cut}}}

\newcommand{\trial}{\psi_{\mathrm{trial}}}

\newcommand{\trialjt}{\widetilde\psi_{\beta_j,\varkappa_j}}

\newcommand{\cell}{\mathscr{C}}
\newcommand{\celln}{\mathscr{C}_n}

\newcommand{\G}{\mathcal{G}}

\newcommand{\disp}{\displaystyle}
\newcommand{\tx}{\textstyle}

\newcommand{\phitrial}{\phi_{\mathrm{trial}}}

\newcommand{\Z}{\mathbb{Z}}
\newcommand{\R}{\mathbb{R}}
\newcommand{\N}{\mathbb{N}}

\newcommand{\E}{\mathcal{E}}

\newcommand{\B}{\mathcal{B}}

\newcommand{\OO}{\mathcal{O}}

\newcommand{\al}{\alpha}

\newcommand{\Om}{\Omega}

\newcommand{\Gt}{\widetilde{\G}}

\newcommand{\Kt}{\widetilde{K}}
\newcommand{\Et}{\widetilde{E}}
\newcommand{\Ett}{\widetilde{\E}}

\newcommand{\one}{\mathds{1}}

\newcommand{\supp}{\mathrm{supp}}

\newcommand{\Hcc}{H_{\mathrm{c}2}}
\newcommand{\Hccc}{H_{\mathrm{c}3}}
\newcommand{\Hstar}{H_{\mathrm{corner}}}

\newcommand{\logi}{|\log \eps| ^{\infty}}

\newtheorem{teo}{Theorem}[section]
\newtheorem{lem}{Lemma}[section]
\newtheorem{pro}{Proposition}[section]

\newtheorem{asum}{Assumption}
\newtheorem{conj}{Conjecture}

\theoremstyle{remark}
\newtheorem{remark}{Remark}[section]

\newcommand{\annbk}{\bar{I}_{k,\ell}}
\newcommand{\annol}{I_{\bar{\ell}}}

\newcommand{\btik}{\bar{t}_{k,\ell}}

\newcommand{\fk}{f_{k}}

\newcommand{\fone}{\E^{\mathrm{1D}}}
\newcommand{\fonekal}{\E ^{\rm 1D}_{k,\alpha}}

\newcommand{\eone}{E^{\mathrm{1D}}}

\newcommand{\eonek}{E ^{\rm 1D}_{k}}

\newcommand{\eoneo}{E ^{\rm 1D}_{0}}

\newcommand{\alk}{\alpha_k}

\newcommand{\alO}{\alpha_0}

\newcommand{\potkal}{V_{k,\alpha}}

\newcommand{\jv}{\mathbf{j}}

\renewcommand{\leq}{\leqslant}
\renewcommand{\geq}{\geqslant}
\renewcommand{\Im}{\mathrm{Im}}
\renewcommand{\Re}{\mathrm{Re}}

\newcommand{\corner}{\Gamma_{\beta}(L,\ell)}
\newcommand{\cornern}{\Gamma_{\beta}(L_n,\ell_n)}
\newcommand{\cornert}{\widetilde{\Gamma}_{\beta}(L,\ell)}
\newcommand{\cornerj}{\Gamma_{\beta_j}(\leps,\elle)}

\newcommand{\cornerr}{\Gamma_{j, \mathrm{rect}}}
\newcommand{\cornere}{\Gamma_{\beta_j}(\leps,\elle)}

\newcommand{\corn}{\Gamma_{\beta}}

\newcommand{\bdo}{\partial \Gamma_{\mathrm{out}}}
\newcommand{\bdi}{\partial \Gamma_{\mathrm{in}}}
\newcommand{\bdbd}{\partial \Gamma_{\mathrm{bd}}}

\newcommand{\osmooth}{\partial \Omega_{\mathrm{smooth}}}
\newcommand{\rect}{R(L,\ell)}

\newcommand{\iell}{I_{\ell}}
\newcommand{\iellb}{I_{\bar{\ell}}}
\newcommand{\elle}{\ell_{\eps}}
\newcommand{\leps}{L_{\eps}}

\newcommand{\corneru}{\Gamma_{j,\eps}}

\newcommand{\smooth}{\mathcal{I}_{\mathrm{smooth}}}

\newcommand{\psit}{\widetilde{\psi}}
\newcommand{\psilf}{\psi_{\beta}}
\newcommand{\psil}{\psi_n}

\newcommand{\ed}{E_{\mathrm{D}}}
\newcommand{\edk}{E_{\mathrm{D},\varkappa}}

\newcommand{\en}{E_{\mathrm{N}}}
\newcommand{\edn}{E_{\mathrm{D}/\mathrm{N}}}

\newcommand{\ecorr}{E_{\mathrm{corr}}}
\newcommand{\ecorn}{E_{\mathrm{corner},\beta}}
\newcommand{\ecornl}{E_{\mathrm{corner},\beta}(L,\ell)}
\newcommand{\ecornln}{E_{\mathrm{corner},\beta}(L_n,\ell_n)}
\newcommand{\ecornt}{\widetilde{E}_{\mathrm{corner},\beta}}

\newcommand{\glecorn}{E_{\beta}}
\newcommand{\glecornt}{\widetilde{E}_{\beta}}

\newcommand{\wG}{\widetilde{\G}}

\newcommand{\exl}{\OO(\ell^{-\infty})}
\newcommand{\exln}{\OO(\ell_n^{-\infty})}

\newcommand{\ee}{\mathfrak{e}}

\begin{document}

\title{Effects of Corners in Surface Superconductivity}

\author[M. Correggi]{Michele Correggi}
\address{Scuola Normale Superiore,  Piazza dei Cavalieri, 7, 56126 Pisa, Italy.}
\email{michele.correggi@gmail.com}

\author[E.L. Giacomelli]{Emanuela L. Giacomelli}
\address{Department of Mathematics, Universit\"{a}t T\"{u}bingen, Auf der Morgenstelle, 10, 72076, T\"{u}bingen, Germany.}
\email{emanuela.giacomelli.elg@gmail.com}

\date{\today}

\begin{abstract}
We study the Ginzburg-Landau functional describing an extreme type-II superconductor wire with cross section with finitely many corners at the boundary. We derive the ground state energy asymptotics up to $ o(1) $ errors in the surface superconductivity regime, i.e., between the second and third critical fields. We show that, compared to the case of smooth domains, each corner provides an additional contribution of order $ \OO(1) $ depending on the corner opening angle. The corner energy is in turn obtained from an implicit model problem in an infinite wedge-like domain with fixed magnetic field. We also prove that such an auxiliary problem is well-posed and its ground state energy bounded and, finally, state a conjecture about its explicit dependence on the opening angle of the sector.
\end{abstract}

\maketitle

\tableofcontents

\section{Introduction}\label{sec:intro}

\noindent
The phenomenon of conventional superconductivity (see, e.g., \cite{Ti} for a review of the physics of superconductors) is nowadays very well understood at the microscopic level thanks to the Bardeen-Cooper-Schrieffer (BCS) theory \cite{BCS}: a collective behavior of the current carriers in the material is responsible for a sudden drop of the resistivity below a certain critical temperature. It is however astonishing how a phenomenological model as the Ginzburg-Landau (GL) theory \cite{GL} is capable of predicting most of the key equilibrium features of the phenomenon, in particular concerning the response of the superconducting material to an external field. When it was introduced in the `50s, indeed, the GL model was motivated only from purely phenomenological considerations. Only later it was shown that the GL theory emerges as an effective macroscopic model from the BCS theory suitably close to the critical temperature \cite{Gor,FHSS,FL}. 

The interplay between superconductivity and strong magnetic fields is known to generate a very rich variety of physical phenomena since the pioneering works of Abrikosov \cite{Ab} and St. James and De Gennes \cite{SJdG} in the late `50s/early `60s, who predicted the occurrence of the famous vortex lattice and of surface superconductivity, respectively, working only in the framework of the GL theory. In extreme synthesis, the response of a type-II superconducting material to the external magnetic field can vary from a perfect repulsion of the field ({\it Meissner effect}), for small fields, to a complete loss of superconductivity, for sufficiently strong ones. In between, several different phases of the material can be observed, ranging from various kinds of vortex states to configurations where the superconduction gets restricted to boundary regions. 
Each of these phase transitions can be associated with a {\it critical magnetic field} marking the threshold for the transition: the three major critical fields are 
	\begin{itemize}
		\item the {\it first critical field}, which separate the Meissner behavior, i.e., {when the} magnetic field inside the material is zero and superconductivity is unaffected, from states where the penetration of the field has occurred at least at isolated points ({\it vortices}), where superconductivity is lost;
		\item the {\it second critical field}, above which the superconducting behavior gets confined at the surface of the sample ({\it surface superconductivity});
		\item the {\it third critical field}, which marks the complete loss of superconductivity.
	\end{itemize}

Let us now introduce in more detail the GL theory: the free energy of the material is given by a nonlinear functional, which in the case of a superconducting {infinite} wire of cross section $ \Omega \subset \R^2 $ reads {in suitable units}
\begin{equation}
	\label{eq: glf}
	\glfe [\psi, \textbf{A}]= \displaystyle\int_{\Omega} \diff\textbf{r}\; \bigg\{ \bigg| \left ( \nabla + i \frac{\textbf{A}}{\varepsilon ^ 2}\right)\psi \bigg|^2 -\frac{1}{2b\varepsilon ^2}(2|\psi|^2-|\psi|^4)	 \bigg\}+\frac{1}{\varepsilon ^4}\displaystyle\int_{\R^2} \diff\textbf{r}\; |\mbox{curl}\textbf{A} - 1|^2,
\end{equation}
{where $ \eps, b > 0  $ are two parameters depending on the {\it London penetration depth} and the intensity of the applied magnetic field, which is assumed to be parallel to the wire.} The function $ \psi $, a.k.a. {\it wave function} or {\it order parameter}, is complex{, while} $ \aav $ is the {\it induced magnetic potential}, whose $ \curl $ yields the intensity of the magnetic field outside and inside the sample {(measured in units $ \eps^{-2} $)}. The physical meaning of the order parameter is twofold: $ |\psi|^2 $ yields the relative density of Cooper pairs and, at the same time, the phase of $ \psi $ contains the information about the stationary {\it current} flowing in the superconductor, i.e., 
\beq
	\label{eq: current}
	\jv[\psi] : = \tx\frac{i}{2} \lf( \psi \nabla \psi^* - \psi^* \nabla \psi \ri) = \Im \lf( \psi^* \nabla \psi \ri).
\eeq
Hence, one typically speaks of a {\it normal state}, if $ \psi = 0 $ and $ \aav $ is such that $ \curl \aav = 1 $, while the {\it perfect superconducting state} is identified by $ |\psi| = 1 $, $ \aav = 0 $. Whenever $ |\psi| $ is non-vanishing everywhere but not identically $ 1 $, the superconductor is said to be in a {\it mixed state}. {Any equilibrium state of the sample minimizes the free energy \eqref{eq: glf} and thus we set}
\beq
	\label{eq: glee}
	\glee := \displaystyle\min_{(\psi,\textbf{A})\in \gldom} \glfe[\psi,\textbf{A}],
\eeq
and denote by $ (\glm, \aavm) $ any minimizing configuration, where
\beq
	\label{eq: gldom}
	\gldom = \lf\{ \lf( \psi, \aav\ri) \in H^1(\Omega) \times H^1_{\mathrm{loc}}(\R^2; \R^2) \: \big| \: \curl \aav - 1 \in L^2(\R^2) \ri\}.
\eeq
We provide some details about the above minimization and the properties of any minimizing configuration $ (\glm, \aavm) $ in \cref{sec: minimization}. We also use the following convention: if we need to specify the dependence on the domain $ \Omega $, we write $ \glfe[\psi, \mathbf{A}; \Omega]  $ for the functional and $ \glee(\Omega) $ for the corresponding ground state energy.

{In the rest of the paper we are going to study the minimization \eqref{eq: glee} in the asymptotic regime}
\beq
	\varepsilon \ll 1,
\eeq
corresponding to an {\it extreme type-II superconductor}. Under this idealization, one can identify the mathematical counterparts of the critical values of the external magnetic field described above in terms of properties of the minimizing configuration $ (\glm, \aavm) $ and it is also possible to precisely identify the behavior of such thresholds {(see, e.g., \cite{SS} for an extensive discussion of the first phase transition). In particular, assuming that $ \Omega $ is a simply connected domain with {\it smooth boundary} $ \partial \Omega $, the {\it second critical field} associated with the transition from {\it bulk} to {\it surface superconductivity} is identified with $ b = 1 $ \cite[Chpt. 10.6]{FH1} and thus with a field of intensity}
			\beq
				{\Hcc = \frac{1}{\eps^2},}
			\eeq
		{based on sharp estimates ({\it Agmon estimates}) of the decay of $ \glm $ in the distance from the boundary (see \cref{app:Agmon}); the {\it third critical field} marking the transition to the normal state on the other hand corresponds to $ b = \theo^{-1} > 1 $, where  $ \theo \simeq 0.59 $ is a universal constant, i.e., more precisely \cite[Chpt. 13]{FH1},}
		\beq
			{\Hccc = \frac{1}{\theo \eps^2} + \OO(1).}
		\eeq

\subsection{Setting: Domains with Corners}

In this paper we are exclusively concerned with the behavior of the superconductor for very strong magnetic fields {\it above the second critical one}, i.e., we always assume that $ \hex > \Hcc $, or, more concretely,
\beq
	b > 1.
\eeq
The main novelty of this paper compared to other works on the GL functional above the second critical field is that we assume that $ \Omega $ is a bounded domain with {a} {\it Lipschitz boundary}, i.e., we allow for the presence of {\it corners} on $ \partial \Omega $ (see \cref{fig: domain}). Indeed, apart from few physics papers (see \cite{BDFM,FDM,SP}), the GL theory on domains with corners has already been studied only in \cite{BNF,HK,Ja,Pa2}, with the focus on the third critical field though, and in \cite{CG}, whose results are improved in this work.

	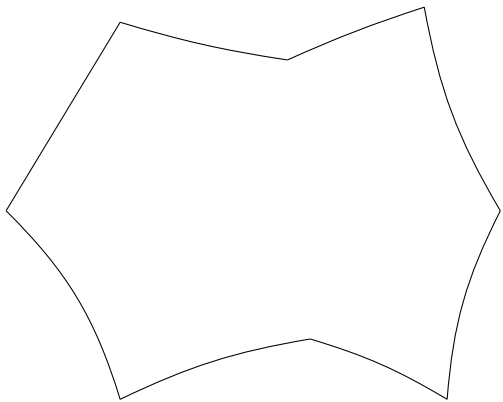
\begin{figure}[!ht]
		\begin{center}
		\begin{tikzpicture}
			\draw (0.5,3) to (2,5.5);
			\draw (2,5.5) to[bend right= 4] (4.2,5);
			\draw (6,5.7) to[bend right = 3] (4.2,5);
			\draw (6,5.7) to[bend right= 10] (7,3);
			\draw (7,3) to[bend right= 11] (6.3,0.5);
			\draw (4.5,1.3) to[bend right = 8] (2,0.5);
			\draw (4.5, 1.3) to[bend left = 7] (6.3,0.5);
			\draw (2,0.5) to[bend right= 14] (0.5,3);
		\end{tikzpicture}
		\caption{A domain $ \Omega $ with Lipschitz boundary and finitely many corners on $ \partial \Omega $.}\label{fig: domain}
		\end{center}{}
	\end{figure}

The main reason why it is interesting to study the behavior of the GL functional in domains with corners for large magnetic fields is that for smaller fields one expects that the presence of corners does not affect the salient features of superconductivity. Indeed, the occurrence of vortices but also their uniform distribution and arrangement in regular lattices, which occur for magnetic fields below $ \Hcc $, are bulk phenomena and, as such, little influenced by the boundary regularity. On the opposite, the surface superconductivity regime, where the density of Cooper pairs is non-vanishing only at and close to the boundary, might clearly depend on the presence of singularities along $ \partial \Omega $. It is then important to know if and to what extent corners can modify the boundary behavior, even more so, considered that in physics experiments it is hardly possible to distinguish between a sample with smooth boundary and another which has corners there (see, e.g., \cite[Fig. 1]{NSG}).

	We now specify in more detail our assumptions on the domain $ \Omega $. First, we require that it is simply connected and its boundary $ \partial \Omega $ is Lipschitz (see, e.g., \cite[Def. 1.4.5.1]{Gr}) {and, more concretely,} it is a curvilinear polygon of class $ C^{\infty} $, given by smooth pieces glued together at finitely many points, where however the curvature remains finite (no cusps). These assumptions are the same made, e.g., in \cite{BNF,CG,HK}.

	\begin{asum}[Piecewise smooth boundary]
		\label{asum: boundary 1}
		\mbox{}	\\
		Let $\Om$ be a bounded open subset of $\R^2$. We assume that $ \partial \Omega $ is a smooth curvilinear polygon, i.e., for every $ \rv \in \partial \Omega $ there exists a neighborhood $ U $ of $ \rv $ and a map $ \phiv: U \to \R^2 $, such that
		\begin{enumerate}[(1)]
			\item $ \phiv $ is injective;
			\item $ \phiv $ together with $ \phiv^{-1} $ (defined from $ \phiv (U) $) are smooth;
			\item the region $\Om \cap U $ coincides with either $ \lf\{ \rv \in \Om \: | \: \lf( \phiv(\rv) \ri)_1 < 0 \ri\}$ or  $ \lf\{ \rv \in \Om \: | \: \lf( \phiv(\rv) \ri)_2 < 0 \ri\}$ or  $ \lf\{ \rv \in \Om \: | \: \lf( \phiv(\rv) \ri)_1 < 0, \lf( \phiv(\rv) \ri)_2 < 0 \ri\} $, where $ \lf( \phiv \ri)_j$ stands for the $ j-$th component of $ \phiv $.
		\end{enumerate}
	\end{asum}
	
	The inward normal $ \nuv $ to $ \partial \Omega $ is thus defined almost everywhere and jumps at the corners. More precisely, if $ \gamv(\ss): [0, |\partial \Omega|) \to \partial \Omega $, is a counterclockwise parametrization of $ \partial \Omega $ satisfying $ |\gamv^{\prime}(\ss)| = 1 $, we can define the mean curvature $ \mathfrak{K}(\ss) $ almost everywhere through the identity
\beq
	\label{eq: curvature}
	\gamv^{\prime\prime}(\ss) = \mathfrak{K}(\ss) \nuv(\ss).
\eeq
{Hence, we can introduce a convenient system of tubular coordinates in a neighborhood of the boundary (see also \cite[Appendix F]{FH1}): for any point $ \rv \in \Omega $ close enough to $ \partial \Omega $, we set
\beq
	\label{eq: tc}
	\rv(\ss, \tt) = : \gamv^{\prime}(\ss) + \tt \nuv(\ss),
\eeq
with
\beq
	\label{eq: normal}
	\tt = \dist\lf(\rv, \partial \Omega\ri).
\eeq
Where $ \partial \Omega $ is smooth, i.e., far enough from the corners, this change of coordinates is a diffeomorphism close to the boundary, e.g., {as far as} $ \dist\lf(\rv, \partial \Omega\ri) = o(1) $.}

	\begin{asum}[Boundary with corners]
		\label{asum: boundary 2}
		\mbox{}	\\
		We assume that the set  $ \Sigma : = \lf\{ \rv_1, \ldots, \rv_N \ri\} $ of corners of $ \partial \Omega $, i.e., the points where the normal $ \nuv $ does not exist, is non empty but finite and given by $ N $ points. We denote by $ \beta_j$ the angle of the $j-$th corner (measured towards the interior) and by $ \ss_j $ its boundary coordinate.
	\end{asum}

\subsection{Heuristics}
\label{sec: heuristics}

Before entering the discussion of what is mathematically known on the phenomenon of surface superconductivity, we resume here its key features for smooth domains, neglecting errors and remainders: if $ b > 1 $, as $ \eps \to 0 $, 
	\begin{itemize}
		\item the order parameter $ \glm $ is non-vanishing only close to $ \partial \Omega $; more precisely it is exponentially small in $ \eps $ at distances from the boundary much larger than $ \eps $;
		\item the induced magnetic field $ \curl \aavm $ is suitably close to the applied one, i.e., a uniform magnetic field of unit strength, and, consequently, one can find a local gauge close to $ \partial \Omega $ in which $ \aavm $ is purely tangential and $ |\aavm| = \dist\{\rv, \partial\Omega\} $;
		\item the modulus of $ \glm $ is essentially independent of the tangential coordinate $ \ss $ and therefore optimizes an effective one-dimensional problem where the only variable is the distance from the boundary;
		\item the phase of $ \glm $ is on the other hand constant in $ \tt $ and linear in $ \ss $, with rapid oscillations, or, more precisely, the current \eqref{eq: current} is constant along $ \ss $.
	\end{itemize}	
Summing up, we expect that $ \aavm $ can be locally replaced by $ - \tt \ev_{\ss} $ close to $ \partial \Omega $ and
\beq
	\lf| \glm(\rv) \ri| \simeq f\lf( \tx\frac{\tt}{\eps} \ri),	\qquad		\jv[\glm] \simeq \frac{|\Omega|}{\eps^2} - \frac{\alpha}{\eps},
\eeq
for some $ f $ positive and $ \alpha \in \R $, which leads to
\beq
	\label{eq: ansatz glm}
	\glm(\rv) \simeq f(\tt/\eps) e^{- \frac{i \alpha \ss}{\eps}} e^{ i \phi_{\eps}(\rv)},
\eeq
$ \phi_{\eps} $ standing for the gauge transformation mentioned above. Note the scaling factors $ 1/\eps $ we have extracted for later convenience, so that $ f $ and $ \alpha $ are quantities of order $ \OO(1) $. 

If we plug the ansatz \eqref{eq: ansatz glm} into the GL energy {\eqref{eq: glf}}, we get
\beq
	\label{eq: 1d model problem}
	 \frac{\lf| \partial \Omega \ri|}{\eps} \int^{\elle}_0 \mbox{dt} \left\{ |\partial_t f|^2 + (t+\alpha)^2 f^2 -\frac{1}{2b} (2f^2-f^4)\right\},
\eeq
i.e., up to the prefactor $ |\partial \Omega|{/\eps} $, a one-dimensional (1D) energy functional evaluated on $ f $ {and depending on the real parameter $ \alpha $}. The value $ \elle > 0 $ is to some extent arbitrary and is chosen much larger than $ 1 $ in order to cover all the superconducting layer: we make the following explicit choice
\beq
	\label{eq: elle}
	\elle : = c_1 |\log\eps|,
\eeq
for a large constant $ c_1 $. The minimization of the 1D functional above and some variants of it w.r.t. both $ f $ and $ \alpha $ is discussed in \cref{sec: 1d}. This identifies the leading term contribution in the GL energy asymptotics $ \eones/\eps $, the optimal 1D profile $ \fs(t) $ and the optimal phase $ \as $.

The next-to-leading order term in the GL energy asymptotics is of order $ \OO(1) $ and depends on the mean curvature of the boundary: one can indeed refine the 1D model problem \eqref{eq: 1d model problem} keeping track of $ \OO(\eps) $ contributions coming from the curvature-dependent terms due to the change of coordinates $ \rv \to (\ss, \tt) $. Indeed, if we define the  rescaled tubular coordinates as
\beq
	\label{eq: rescaled tc}
	\begin{cases}
		t : = \tt/\eps \in [0, \elle],		\\
		s : = \ss/\eps \in \lf[0, \frac{|\partial \Omega|}{\eps} \ri].
	\end{cases}
\eeq
we get $ \diff \rv = \diff \tt \diff \ss \lf(1 - \curv \tt \ri) $, or, equivalently,
\beq
	\label{eq: jacobian}
	\diff \rv = \eps^2 \diff t \diff s \lf(1 - \eps k(s) t \ri),
\eeq
where we have set
\beq
	k(s) : = \mathfrak{K}(\eps s).
\eeq

\subsection{Summary}
In this paper we continue the analysis started with \cite{CG}. The expansion \eqref{eq: glee asympt CG} provides indeed only the leading order term in the energy asymptotics and does not capture the corner effects{, that we are going to investigate. More precisely, we prove that the presence of corners modifies} the $ \OO(1) $ term in the expansion \eqref{eq: gle asympt smooth}. We also identify the model problem which yields such a new contribution {in terms of a genuine 2D model}. Finally, we prove that the pointwise estimate of $ |\glm| $ in terms of $ \fs $ still holds far from the corners, precisely as in the smooth case. 

After having introduced some notation in \cref{sec: notation}, we define in \cref{sec: main results} the effective model in the corner region and state our main results about the GL energy asymptotics and the pointwise estimate of the order parameter far from corners. Further comments about the effective model and a conjecture about the possible explicit expression of the effective energy are contained in \cref{sec: corner energy}. The well-posedness of the effective model problem is proven in detail in \cref{sec: corner effective}, whereas in \cref{sec: lower bound} and \cref{sec: other proofs} we provide the energy lower bound and the rest of the arguments needed to complete the proof of our main results, respectively. The Appendix is divided in three parts: in \cref{sec: 1d} we discuss the effective 1D problems and their related properties; the GL minimization and some useful technical estimates are treated in \cref{sec: technical}; finally, \cref{sec: smooth energy} recalls the salient steps of the proof of the energy asymptotics in domains with smooth boundaries, which are used to complete the proof of energy expansion.

\subsection{Notation} 
\label{sec: notation}
Given their key role in the rest of the paper, we recall the definitions \eqref{eq: tc} and \eqref{eq: rescaled tc} of tubular coordinates $ (\ss, \tt) $ and their rescaled counterparts $ (s,t) $. We stress that $ (\ss, \tt) $ or, equivalently, $ (s,t) $ provide a smooth diffeomorphism, e.g., in
\bdm
	\lf\{ \rv \in \anne \: \big| \: \dist(\rv, \Sigma) \geq \eps|\log\eps| \ri\},
\edm
where $ \Sigma $ is the set of corner positions and 
\beq
	\label{eq: anne}
	\anne : = \lf\{ \rv \in \Omega \: \big| \: \dist\lf(\rv, \partial \Omega\ri) \leq \eps \elle \ri\},
\eeq
for $ \eps \ll 1 $, where {(see \eqref{eq: elle})}
\bdm
	\elle : = c_1 |\log\eps|
\edm
and $ c_1 $ is large enough constant, which is set once for all (see next \eqref{eq: glm exp small}). Given a differentiable function $ \psi(\rv) $ and a vector $ \aav(\rv) $, the transformations induced by the change of coordinates $ \rv \to (\ss, \tt) $ are
\beq
	\label{eq: transformation gradient}
	\lf( \nabla \psi \ri) (\rv(\ss,\tt))   = \lf(1 - \curv \tt \ri)^{-1} \big( \partial_\ss \widetilde{\psi} \big)   \ev_{\ss} + \big( \partial_{\tt} \widetilde{\psi} \big) \ev_{\tt},
\eeq
where we have set $ \widetilde\psi(\ss,\tt) : = \psi(\rv(\ss,\tt)) $ and
\beq
	\ev_{\ss} : = \gamv^{\prime}(\ss),		\qquad		\ev_{\tt} : = \nuv(\ss),
\eeq
for short. As a consequence, for any vector $ \aav $,
\bml{
	\label{eq: curl}
	\lf( \curl \: \aav \ri)(\rv(\ss,\tt)) = - \partial_{\tt} \lf( \aav(\rv(\ss,\tt)) \cdot \ev_{\ss} \ri) \\
	+ \lf(1 - \curv \tt \ri)^{-1} \lf[ \partial_{\ss} \lf( \aav(\rv(\ss,\tt)) \cdot \ev_{\tt} \ri) + \curv \aav(\rv(\ss,\tt)) \cdot \ev_{\ss} \ri].
}

We are going to make use of Landau symbols, with the following convention: given two functions $ f(x), g(x) $, with $ g > 0 $,
\begin{itemize}
	\item $ f = \OO(g) $, if $ \disp\lim_{x \to 0^+} |f|/g \leq C $;
	\item $ f = o(g) $, if $  \disp\lim_{x \to 0^+} |f|/g = 0 $;
	\item $ f \sim g $, if $ f = \OO(g) $ and $ \disp\lim_{x \to 0^+} |f|/g > 0 $;
	\item for $ f > 0 $, $ f \ll g $ or $ f \gg g $, if $ f = o(g) $ or $ g = o(f) $, respectively;
	\item for $ f > 0 $, $ f \lesssim g $ or $ f \gtrsim g $, if $ f = \OO(g) $ or $ g = \OO(f) $, respectively.
\end{itemize}
We also commit a little abuse of the notation by using the symbols $ \OO( \: ) $ and $ o( \: ) $ inside an inequality to mean a precise direction of the estimate. As usual, $ \OO( \: ) $ and $ o( \: ) $ stand for quantities whose sign is not known. In case of functions of two or more variables, we point out the parameter, whose asymptotics we are considering, by adding a label, e.g., $ o_x(\:) $ or $ \OO_x( \: ) $ is meant to stress that we are estimating the behavior of the function as $ x \to 0^+ $. 
Finally, we say that a quantity is $ \OO(\eps^{\infty}) $, as $ \eps \to 0^+ $, if it is smaller than any power of $ \eps $, i.e., it is exponentially small in $ \eps $. We will also use the following convention: $ \OO(\eps^a \logi) $, $ a > 0 $, stands for a quantity which is bounded by $ \eps^a |\log\eps|^b $ for some large but finite power $ b > 0 $, which is however not relevant since  the $ |\log\eps|$-factor is always dominated by $ \eps $-powers.

\section{Main Results}
\label{sec: main results}

\subsection{State of the art} 
\label{sec: state of the art} 
We briefly review here the most recent and relevant results on surface superconductivity, which are related to the analysis carried on in this paper (see \cite{Cor} for a more detailed review). 
After the series of works \cite{CR1,CR2,CR3,CDR}, following \cite{Pa}, where the problem was first investigated, and \cite{AH,FHP}, the phenomenon of 2D surface superconductivity in domains with smooth boundaries is well understood: combining \cite[Thm. 1]{CR1} with \cite[Lemma 2.1]{CR2} (see also \cite{CDR}), one gets that, whenever
\beq
	1 < b < \theo^{-1},
\eeq
the GL energy asymptotics is given by
\beq
	\label{eq: gle asympt smooth}
	\glee = \frac{|\partial \Omega| \eones}{\eps} - 2 \pi \ecorr + o(1),
\eeq
where 
\beq
	\label{eq: eones}
	\eones : = \inf_{\alpha \in \R} \inf_{f \in \onedom}  \int^{+\infty}_0 \mbox{dt} \left\{ |\partial_t f|^2 + (t+\alpha)^2 f^2 -\frac{1}{2b} (2f^2-f^4)\right\},
\eeq
and
\beq
	\label{eq: ecorr}
	\ecorr : = \int_{0} ^{+\infty} \diff t \:  t\lf\{ \lf| \partial_t \fs \ri|^2 + f_0 ^2 \left( -\as (t+\as) -\frac{1}{b} + \frac{1}{2b} \fs ^2\right)\ri\} = \frac{1}{3} \fs^2(0) \as - \eones,
\eeq
$ \as, \fs $ being a pair of minimizers of \eqref{eq: eones} (see \cref{sec: 1d disc}). Note that \eqref{eq: gle asympt smooth} can also be rewritten as 
\bdm
	{\glee =  \int_0^{\frac{|\partial \Omega|}{\eps}} \diff s \: \eone_{k(s)} + \OO(\eps|\log\eps|^{\infty}),}
\edm
{with a more precise remainder term and where $ \eone_k $ is defined in \eqref{eq: eonek} in \cref{sec: 1d curvature}. Expanding further $ \eonek $}, the next-to-leading order correction in \eqref{eq: gle asympt smooth} can be shown to be
\beq
	\label{eq: gauss-bonnet smooth}
	-\ecorr \int_{\partial \Omega}  \diff \ss \: \mathfrak{K}(\ss) + o(1) = -2\pi \ecorr + o(1),
\eeq
by the Gauss-Bonnet theorem, because $ \Omega $ is flat and the Euler characteristic is equal to $ 1 $. Moreover, in \cite{CDR} the quantity $ \ecorr $ is numerically evaluated and it is shown that it is positive, which has some important consequences on the distribution of superconductivity near the boundary: regions with larger curvature attract Cooper pairs, which concentrate more there (to first order), although to leading order superconductivity is uniform at the boundary. 

Indeed, a consequence of \eqref{eq: eones} is that \cite[Thm. 1]{CR1} the density $ |\glm|^2 $ is $ L^2$-close to the reference density $ \fs $. Such an estimate can in fact be strengthened in two directions:
\begin{itemize}
	\item in \cite[Thm. 2]{CR1} it is proven that there exists a boundary layer $ \mathcal{A}_{\mathrm{bl}} \subset \lf\{ \rv | \dist(\rv, \partial \Omega) \leq \eps|\log\eps| \ri\}$, containing the bulk of superconductivity, where Pan's conjecture holds true, i.e.,
		\beq
			\label{eq: pan smooth}
			\lf\| \lf| \glm(\: \cdot \:) \ri| - \fs(\dist(\: \cdot \:,\partial \Omega)/\eps) \ri\|_{L^{\infty}(\mathcal{A}_{\mathrm{bl}})} = o(1);
		\eeq
	\item the approximation of $ |\glm| $ in terms of $ \fs $ holds also locally \cite[Thm 1.1]{CR2} and one can explicitly derive the asymptotics of the density of superconductivity (in fact, the $ L^4 $ norm of $ \glm $) in any reasonable subdomain contained in $ \Omega $.
\end{itemize}

It is expected that a regime of surface superconductivity with similar features occurs also for genuine 3D samples but so far only partial results are available \cite{FKP,FMP}. In particular, in \cite[Thm 1.1]{FKP} (see also \cite{FK2}) it is shown that such a regime does exist and the leading order term in the energy asymptotics can be identified, {although} in terms of a rather implicit effective problem. In \cite[Thm. 1.5]{FMP} it is then proven that, when the magnetic field is parallel to the 3D boundary, the effective model is still given by the 1D functional \eqref{eq: eones} above.

One of the major differences for samples with non-smooth boundary is that one expects \cite{BNF,HK,Ja,JRS,Pa2} a shift of the third critical field, provided there is at least one corner with acute opening angle $ 0< \beta < \pi $: more precisely, the transition to the normal state should occur \cite[Thm. 1.4]{BNF} for applied fields larger than 
\beq
	\label{eq: Hccc}
	\Hccc = \frac{1}{\mu(\beta) \eps^2}
\eeq
where 
\beq
	\mu(\beta) : = \inf \mathrm{spec}_{L^2(W_\beta)} \lf( - \lf(\nabla + \tx\frac{1}{2} i \xv^{\perp} \ri)^2 \ri), 
\eeq
is the ground state energy of a Schr\"{o}dinger operator with uniform magnetic field  in an infinite sector $ W_{\beta} $ of angle $ \beta $. The above result is however conditioned to the inequality $ \mu(\beta) < \theo $ (see also \cite[Chpt. 8.2]{Ra} and references therein), which is known to be true for $ 0 < \beta < \pi/2 + \epsilon $ \cite{Bo,Ja,ELP} but is expected to hold in the whole interval $ \beta \in (0, \pi) $, based on numerical simulations (see, e.g., \cite{Bo,BD,ELP}).

As the applied field gets closer to \eqref{eq: Hccc} from below, the order parameter concentrates around the corner of smallest opening angle and becomes smaller and smaller everywhere else. Hence, one can speak of a {\it corner superconductivity} regime occurring before the transition to the normal state. On the other hand, in \cite{CG}, we proved that, if $ 1 < b < \theo $, superconductivity is still uniform along the boundary (although only in $ L^2 $ sense), leading to the conjectured existence of another critical field
\beq
	\Hstar = \frac{1}{\theo \eps^2},
\eeq
which marks the transition from surface to corner concentration. Indeed, if $ 1 < b < \theo $, then \cite[Thm 1.1]{CG}
\beq
	\label{eq: glee asympt CG}
	\glee = \frac{|\partial \Omega| \eones}{\eps} + \OO(|\log\eps|^2),
\eeq
and, more importantly,
\beq
		\label{eq: l2 estimate}
	\lf\| \lf|\glm( \: \cdot \:) \ri|^2 - \fs^2\lf( \dist( \: \cdot \:, \partial \Omega)/\eps \ri) \ri\|_{L^2(\Omega)} = \OO(\eps|\log\eps|),
\eeq
which implies, to leading order, uniform distribution of {superconductivity} along the boundary layer. 

The result of \cite{BNF} has also been recently improved in \cite{HK}, where the presence of several corners is taken into account and shown that, under the same unproven assumption, one can identify several critical fields associated to the concentration of the order parameter close to the respective corner. We also stress that, as noted in \cite[Rmk. 1.9]{AKP} (see also \cite{Ass,AK}), the behavior of superconductivity in presence of corners is expected to be recovered in the case of magnetic steps, i.e., for applied magnetic fields with a jump singularity along a curve.

\subsection{GL energy and density asymptotics}

	Before stating our main results, we have to define the effective problem near the corners. Here we only provide a sketchy definition and in the next \cref{sec: corner energy} we comment further about its well-posedness and heuristic meaning. The model problem is given by first minimizing the GL functional with given magnetic potential and unit magnetic field in large wedge-like domain (see \cref{fig: corner}), and then subtracting the surface energy of the outer boundary of the wedge. The wedge domain is supposed to describe the rectified and rescaled area close to each corner, where the only relevant parameter is the opening angle $ \beta_j $.

	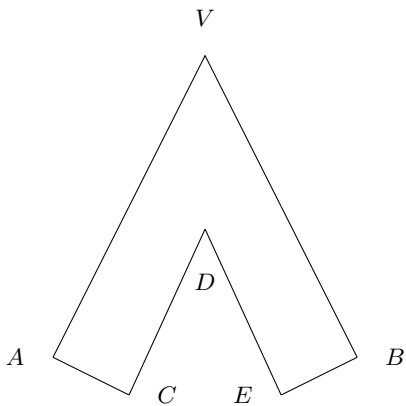
\begin{figure}[!ht]
		\begin{center}
		\begin{tikzpicture}
			\draw (0,0) -- (1,2);
			\draw (1,2) -- (1.5,3);
			\draw (1.5,3) -- (2,4) -- (2.5,3);
			\draw (2.5,3) -- (3,2);
			\draw (3,2) -- (4,0);
			\draw (0,0) -- (1,-0.5);
			\draw (4,0) -- (3,-0.5);
			\draw (1,-0.5) -- (2,1.7);
			\draw (3,-0.5) -- (2,1.7);
			\node at (2,4.5) {{\footnotesize $V$}};
			\node at (-0.5, 0) {{\footnotesize $A$}};
			\node at (4.5, 0) {{\footnotesize $B$}};
			\node at (1.5, -0.5) {{\footnotesize $C$}};
			\node at (2.5, -0.5) {{\footnotesize $E$}};
			\node at (2,1) {{\footnotesize $D$}};
		\end{tikzpicture}
		\caption{The region $ \corner $, where $ \beta $ is the opening angle $ \widehat{AVB} $, $ L = |\overline{AV}| = |\overline{VB}| $ and $ \ell = |\overline{AC}| = |\overline{EB}| $.}\label{fig: corner}
		\end{center}{}
	\end{figure}

We thus define the corner energy as
\beq
	\label{eq: ecorn}
	\boxed{
	\ecorn : = \displaystyle\lim_{\ell \to +\infty} \lim_{L \to + \infty} \lf( - 2 L \eoneo(\ell) + \inf_{\psi \in \doms(\corner)} \glf_{1}\lf[\psi, \fv; \corner \ri] \ri), }
\eeq
where $ \eoneo(\ell) $ is a 1D effective energy analogous to \eqref{eq: eones}, which is explicitly given by (see \cref{sec: 1d no curv} for further details)
\beq
	\label{eq: eoneo}
	\eoneo(\ell) : = \inf_{\alpha \in \R} \inf_{f \in H^1([0,\ell])}  \fone_{0,\alpha}[f],	
\eeq
with
\beq
	\fone_{0,\alpha}[f] : = \int^{\ell}_0 \mbox{dt} \left\{ |\partial_t f|^2 + (t+\alpha)^2 f^2 -\frac{1}{2b} (2f^2-f^4)\right\},
\eeq
and we denote by $ \al_0 \in \R $, $ \fol \in H^1([0,\ell]) $ a corresponding minimizing pair, i.e., $ \eoneo(\ell) = \fone_{0, \alpha_0}[\fol] $. The minimization domain is
\beq
	\label{eq: doms}
	 \doms(\corner) : = \lf\{ \psi \in H^1(\corner) \: \big| \: 
	 \lf. \psi\ri|_{\bdbd \cup \bdi} = \psi_{\star} \ri\}.
\eeq
where the boundaries are identified in \cref{fig: corner} by the segments $ \bdbd = \overline{AC} \cup \overline{EB} $ and $ \bdi = \overline{CD} \cup \overline{DE} $ and, in local tubular coordinates $ (s,t) \in [-L,L]\times [0,\ell] $, { we set, for $ |s| $ large enough, e.g., for $ |s| \geq \frac{\ell}{\tan(\beta/2)} $,}
\beq
	\label{eq: psi star}
	\psi_{\star}(\rv(s,t)) : = f_0(t) \exp \lf\{ -i \alpha_0 s - \tx\frac{1}{2} i st \ri\}.		
\eeq
{Note that $ \psi_{\star} $ is ill-defined in the whole corner region but we are going to use it (see \eqref{eq: doms}) only where tubular coordinates are meaningful, i.e., far enough from the corner.}
Any function in $ \doms(\corner) $ has thus to satisfy Dirichlet non-zero boundary conditions on $ \bdi $ and $ \bdbd $ in trace sense, whose role is going to be explained in next \cref{sec: corner energy}. Finally, the magnetic potential $ \fv $ is fixed and equals
\beq
	\label{eq: fv}
	\fv(\rv) : = \tx\frac{1}{2} \lf( - y, x \ri) = : \tx\frac{1}{2} \rv^{\perp},
\eeq
in a coordinate system chosen\footnote{In fact, any choice of the coordinate system would lead to the same energy because of rotational invariance of the GL functional and its gauge symmetry, which allows to incorporate any translation of the origin.} as in \cref{fig: axis}.
We also point out that the existence of the limit in \eqref{eq: ecorn} is not trivial at all and, in fact, it will be the main content of \cref{pro: ecorn}. {Furthermore, the GL functional in the second term on the r.h.s. of \eqref{eq: ecorn} is independent of $ \eps $, but still contains the parameter $ b \in (1, \theo^{-1}) $.

	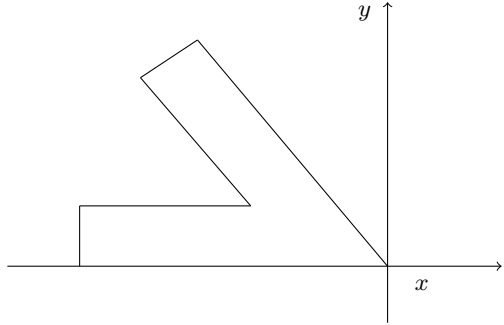
\begin{figure}[!ht]
		\begin{center}
		\begin{tikzpicture}[scale=0.5]
			\draw[->] (0.5,2) to (13.5,2);
			\draw[->](10.5,0.5) to (10.5,9);
			\draw (2.4,2) to (2.4,3.6);
			\draw (2.4,3.6) to (6.9,3.6);
			\draw (4,7) to (6.9,3.6);
			\draw (4,7) to (5.5, 8);
			\draw (5.5,8) to (10.5,2);
			\node at (11.4,1.5) {{\footnotesize $x$}};
			\node at (9.9,8.7) {{\footnotesize $y$}};
		\end{tikzpicture}
		\caption{Cartesian coordinates for the corner domain.}\label{fig: axis}
		\end{center}
	\end{figure}

The main result we prove in this paper is about the GL energy asymptotics as $ \eps \to 0 $, i.e., we derive the expansion of $ \glee $ up to correction of order $ o(1) $. Compared with the case of domains with smooth boundary, some new terms of order $ \OO(1) $ appear: each corner indeed contributes to the energy by $ E_{\mathrm{corner}, \beta_j} $, $ \beta_j $ being the corresponding opening angle.

	\begin{teo}[GL energy asymptotics]
		\label{teo: gle asympt}
		\mbox{}	\\
		Let $\Om\subset\R^2$ be any bounded simply connected domain satisfying \cref{asum: boundary 1} and \cref{asum: boundary 2}. Then, for any fixed
		\beq
			\label{eq: b condition}
			 1< b < \theo^{-1},
		\eeq
		 as $\eps \to 0$, it holds
		\beq
			\label{eq: gle asympt}
			\framebox{$
			\glee = \disp\frac{|\partial\Om| \eones}{\eps} - \ecorr \int_{0}^{|\partial \Omega|} \diff \ss \: \curv + \sum_{j = 1}^N E_{\mathrm{corner}, \beta_j} + o(1). $}
		\eeq
	\end{teo}

	\begin{remark}[Critical field $ \Hcc $]
		\label{rem: Hcc}
		\mbox{}	\\
		In smooth domains the regime of surface superconductivity corresponds to the parameter interval $ 1 < b < \theo^{-1} $, namely the second critical field is $ b = 1 $, while the third one is precisely $ b= \theo^{-1} $. This is motivated by the results in \cite{CR2,CR3} in combination with \cite{FK,FH3}, where it is proven that, for $ b < 1 $, there is still superconductivity in the bulk, while, for $ b > \theo^{-1} $, the normal state is a global minimizer of the GL functional, respectively. The condition $  b > 1  $ is expected to be sharp also for domains with corners and, consequently, we expect that the second critical field is given by
		\beq
			\Hcc = \frac{1}{\eps^2},		
		\eeq
		The value $ 1/\eps^2 $ can actually be taken as a definition of the second critical field, but, as for smooth domains, it would be necessary to show that, for $ b \leq 1 $, there is still superconductivity in the bulk. This has not yet been proven in case of samples with corners, but, based on the results proven in \cite{FK}, it is highly expected.
	\end{remark}

	\begin{remark}[Critical field $ \Hstar $]
		\label{rem: Hstar}
		\mbox{}	\\
		The result proven in \cref{teo: gle asympt} substantiates even more than \cite{CG} the conjecture about the appearance of an additional critical field 	
		\beq
				\Hstar = \frac{1}{\theo \eps^2},
		\eeq
		when corners are present along the sample boundary. Indeed, combining \eqref{eq: gle asympt} and, more importantly, next \cref{pro: pan}, with \cite[Thm. 1.6]{BNF} (see also \cite[Thm. 1.2]{HK}), which states the exponential decay of $ \glm $ in the distance from $ \Sigma $ (still, based on the unproved conjecture on the linear model), one concludes that superconductivity is uniform along the boundary layer until the threshold $ b = \theo^{-1} $ is crossed and, then, concentrates close to the corners with smallest opening angles. More precisely, assuming that all the angles $ \beta_j $ are acute and different, one can identify \cite[Rmk. 1.4]{HK} a sequence of $ N $ critical fields
		\beq
			\Hstar = H_{\mathrm{corner},0} \leq H_{\mathrm{corner},1}  \leq \ldots \leq H_{\mathrm{corner},N-1} \leq H_{\mathrm{corner},N} = \Hccc,
		\eeq
		with
		\beq
			 H_{\mathrm{corner},j} = \frac{1}{\mu(\beta_j) \eps^2},		\qquad 	\mbox{for } 1 \leq j \leq N,
		\eeq
		so that, in between $ H_{\mathrm{corner},j-1} $ and $ H_{\mathrm{corner},j} $, the material is superconducting only close to the $ j-$th corner $ \rv_j $. Let us stress that all these results are conditioned by the request $ \mu(\beta_j) < \theo $ for all the corners, which is expected to hold true (but not proven) for any acute angle $ 0 < \beta_j < \pi  $.
	\end{remark}

	Once the energy asymptotics is obtained, it is natural to ask whether one can extract information about the behavior of the order parameter, which would then give access to the physically relevant quantities, as the density of Cooper pairs. As already proven in \cite[Thm. 1.1]{CG}, the distribution of superconductivity along the boundary layer is uniform to leading order (see \eqref{eq: l2 estimate}). Note that such an estimate goes along with the exponential decay proven in \eqref{eq: agmon}, which implies that $ \glm = o(1)$  at distance much larger than $ \eps $ from the boundary $ \partial \Omega $: we can indeed restrict out attention to the boundary layer $ \anne $ defined in \eqref{eq: anne}, since
	\beq
		\label{eq: glm exp small}
		\glm(\rv) = \OO\lf( \eps^{c_1 \cdot c(b) + 1} \ri),		\qquad		\mbox{in } \Omega\setminus \anne,
	\eeq
	and by taking $ c_1 $ large, the above quantity can be made arbitrarily small. We thus denote it as $ \OO(\eps^{\infty}) $, to stress that it is an arbitrarily large power of $ \eps $. 
	
	{\begin{remark}[Refined $ L^2 $ estimate]
		\label{rem: refined l2 estimate}
		\mbox{}	\\
		An almost direct consequence of the energy asymptotics \eqref{eq: gle asympt} is an improvement of the bound \eqref{eq: l2 estimate}: setting
		\beq
			\label{eq: omega smooth}
			\Omega_{\mathrm{smooth}} : = \lf\{ \rv \in \Omega | \dist(\rv, \Sigma) \geq c_2 \eps |\log\eps| \ri\},
		\eeq
		for some large enough constant $ c_2 > 0 $, one has 
		\beq
			\label{eq: refined l2 estimate}
			\lf\| \lf| \glm(\: \cdot \:) \ri|^2 - {\fs}^2\lf( \dist(\: \cdot \:, \partial \Omega)/\eps \ri) \ri\|_{L^2(\Omega_{\mathrm{smooth}} )} = o(\eps|\log\eps|).
		\eeq
	\end{remark}}

	The {estimates \eqref{eq: l2 estimate} or \eqref{eq: refined l2 estimate} do} not exclude however the presence of vortices or region with very little superconductivity left close to the boundary and, therefore, one would like to prove a bound in a stronger norm, e.g., in $ L^{\infty} $, which is stated in the next \cref{pro: pan}.

	\begin{pro}[GL order parameter]
		\label{pro: pan}
		\mbox{}	\\
		Under the same assumptions of \cref{teo: gle asympt},
		\beq
			\label{eq: pan}
			\framebox{$ \lf\| |\glm(\rv)| - {\fs}(0) \ri\|_{L^{\infty}(\osmooth\cap \partial \Omega) } = o(1). $}
		\eeq
	\end{pro}
	
	\begin{remark}[Uniform distribution of superconductivity]
		\label{rem: uniform}
		\mbox{}	\\
		The estimate \eqref{eq: pan} can in fact be extended to the boundary layer of points $ \rv $ such that $ \dist(\rv, \osmooth) \leq \eps \sqrt{|\log\eps|} $, in the very same way as the analogous result in \cite[Thm. 2]{CR1}. An important consequence is the uniformity of superconductivity in $ \anne $, where one has
		\beq
			\lf|\glm(\: \cdot \:) \ri| \sim \fs\lf( \dist(\: \cdot \:, \partial \Omega)/ \eps \ri),
		\eeq
		not only in weak sense, as proven in \cite{CG}, but also pointwise. Strictly speaking, the corner regions are excluded, but, on the one hand, their overall area is $ \OO( N \eps^2 |\log\eps|^2) $, i.e., much smaller than $ |\anne| $, and, on the other, we do expect the minimizer of the corner problem to be close to $ \fs $ almost everywhere but very close to the corner. An interested reader might wonder whether it is possible to show that $ \glm $ is close to such an effective minimizer in the corner region, but this presumably requires to get some more information about the effective problem \eqref{eq: ecorn} as well as extract a more precise estimate of the remainders in \eqref{eq: gle asympt}.
	\end{remark}
	
	\begin{remark}[Current along $ \partial \Omega $]
		\label{rem: current}
		\mbox{}	\\
		An important consequence of \eqref{eq: pan smooth} in smooth domains is the non-vanishing of $ \glm $ close to the boundary, because of the strict positivity of $ \fs $, and thus surface superconductivity is robust w.r.t. the inclusion of the applied magnetic field. In addition, this allows to estimate the current \eqref{eq: current} along the boundary or, equivalently, the total winding number of $ \glm $ on $ \partial \Omega $ \cite[Thm. 3]{CR1}:
\beq
	\label{eq: deg}
	\deg \lf( \glm, \partial \Omega \ri) = \frac{|\Omega|}{\eps^2} - \frac{\as}{\eps}(1 + o(1)).
\eeq
Such a behavior is similar (although physically different) to the ultrafast rotation regime for angular velocities larger than the third critical one of rotating Bose-Einstein condensates, when vortices are expelled from the boundary region \cite{CRY,CD} (see also \cite{CRv,CY,CPRY} for further results on rotating condensates). In presence of corners, \eqref{eq: pan} guarantees the non-vanishing of $ \glm $ only far from the corners and prevents us to estimate the current on $ \partial \Omega $. Indeed, the pointwise estimate of the gradient \eqref{eq: est grad infty} allows a variation of order $ 1 $ of $ \glm $ on a scale $ \eps $, which is much smaller than the tangential length of the corner region, thus implying that $ \glm $ may a priori vanish there.
	\end{remark}

	\subsection{Corner effective energy}
	\label{sec: corner energy}
		
	We now give more details  about the corner effective problem. Let us start by identifying more precisely the corner region depicted in \cref{fig: corner}. It is meant as a suitable stretching and rescaling (on a scale $ \eps $) of a local area around any corner of $ \Omega $ of tangential and normal lengths both of order $ \eps |\log\eps| $, as $ \eps \to 0 $. For later convenience, however, we consider a region where the tangential length $ L $ along the angle is different from the normal length $ \ell $. Let then $ \corner $ be a triangle-like region as in \cref{fig: corner}, where  $ \beta $ is the opening angle at the vertex $ V $ and side lengths $ L,  \ell > 0 $. In order to reproduce the shape of \cref{fig: corner}, we always assume that
	\beq
		\label{eq: L ell condition}
		\ell \leq \tan\lf( \tx\frac{\beta}{2}\ri) \, L.
	\eeq
	{We recall the definition of the boundaries $ \bdi, \bdbd $ provided in \cref{sec: main results} and denote} by $ \bdo $ the outer boundary $ \overline{AVB} ${, so that $ \partial \corner = \bdo \cup \bdi \cup \bdbd $.}

The effective energy in the corner region is given by a suitably rescaled GL energy with fixed magnetic potential \eqref{eq: fv}. The effective variational model is then
	\beq
		\label{eq: ecornl}
		\ecornl : = - 2 L \eoneo(\ell) + \inf_{\psi \in \doms(\corner)} \glf_{1}\lf[\psi, \fv; \corner\ri],
	\eeq
	where $  \doms(\corner)  $ is defined in \eqref{eq: doms}.
	The heuristics behind the choice \eqref{eq: ecornl} is that in the surface superconductivity regime each portion of the boundary of the sample yields a (leading order) energy contribution proportional to $ \eones $ times its length, which equals $ \eones |\bdo| = 2 L \eones $ in the case of $ \corner $. Indeed, the boundaries $ \bdi $ and $ \bdbd $ are not expected to give any energy contribution. More precisely, $ \bdi $ is immersed in the bulk, where the order parameter is exponentially small in $ \ell $ and it could have been removed from the outset by consider a solid wedge; similarly, $ \bdbd $ is a fictitious boundary, whose role is to separate the corner region from the rest. Mathematically, the non-zero Dirichlet conditions on $ \bdi $ and $ \bdbd $ {in the minimization domain $ \doms $} guarantee that those portions of the boundary do not contribute {to surface superconductivity}. 
	
	Once the boundary energy $ 2 L \eones $ has been subtracted, what remains is precisely the additional energy due to the presence of the corner. Such an energy is indeed of purely geometric nature and is generated by the constraint on the boundary $ \bdo $: in order to reproduce the correct energy along $ \bdo $, the minimizer must behave like the model order parameter  $ \fs(t) e^{-i \as s} $ in a layer of width $ \OO(1) $ around $ \bdo $, but the coordinate $ s $ has a jump on the bisectrix of the domain and thus such a behavior is allowed only close far from the corner. The modulus of the minimizer $ \fs(t) $ is in fact well defined and continuous everywhere, since it depends on the normal coordinate which is continuous as well. Hence, in order to glue together the two model profiles, any minimizer must accommodate a non-trivial phase factor, which must be genuinely 2D, because no 1D function can adjust the jump of $ - i \as s $ along the bisectrix. Unfortunately, the explicit expression of such a phase remains unknown, expect in certain specific cases (for almost flat angles, see \cite{CG2}).
	
	The GL energy functional appearing in \eqref{eq: ecornl} is gauge invariant but we have chosen to work in a prescribed local gauge, i.e., we have made an explicit choice of the vector potential $ \fv $, generating a unit magnetic field. In this respect the GL energy in \eqref{eq: ecornl} is similar to the effective functional studied in \cite[Eq. (1.11)]{BNF}, although both the parameter regime and the domain are slightly different. {Such a difference reflects indeed the different behavior of the minimizer: in the present setting it decays in the distance from the outer boundary, whereas in \cite{BNF}, the decay is in the distance from the corner.}
	
	Recalling that $ L $ and $ \ell $ are obtained via a rescaling from the tangential and normal length of the corner region and thus, in the original problem in $ \Omega $, are actually of order $ |\log\eps| \gg 1 $, we have to study the limit $ L, \ell \to + \infty $ of \eqref{eq: ecornl}.

	\begin{pro}[Corner energy]
		\label{pro: ecorn}
		\mbox{}	\\
		{Let $ \lf\{ \ell_n \ri\}_{n \in \N} $, $ \lf\{ L_n \ri\}_{n \in \N} $ be two monotone sequences with $ \ell_n, L_n \to  + \infty $, as $ n \to + \infty $, and $ \beta \in (0, 2\pi) $, such that $ 1 \ll \ell_n \leq \tan\lf( \beta/2 \ri) \, L_n \leq C \ell_n^{a} $ for some $ a > 0 $. Then, for any $ 1 < b < \theo^{-1} $, the limit}
		\beq
			\label{eq: ecorn def}
			{\lim_{n \to + \infty} \ecorn(L_n,\ell_n) = : \ecorn}
		\eeq
		{exists, it is finite and independent of the chosen sequences.}
	\end{pro}

	As stated in \cref{pro: ecorn} (see also \cref{pro: ecorn bound}), the corner energy $ \ecorn $ is bounded for any $ \beta \in (0,2\pi) $, although we have no information on its sign. In fact, it might as well be zero. In a companion paper \cite{CG2} however we prove that, when $ \beta $ is close to $ \pi $, this is not the case (see also below). 
	
	Once the well-posedness of the model problem has been proven, it is then natural to ask whether one can derive the explicit dependence of $ \ecorn $ on the angle $ \beta $. So far we have not found such an expression but, based on some heuristic arguments, we formulate below an unproven conjecture, which is inspired again by the Gauss-Bonnet theorem. Indeed, the first order correction to the GL energy asymptotics in smooth domains reads equivalently 
	\beq
		\label{eq: smooth gb}
		-\ecorr \int_{0}^{|\partial \Omega|} \diff \ss \: \curv = -2 \pi \ecorr.
	\eeq
	In presence of corner singularities on $ \partial \Omega $, the Gauss-Bonnet theorem has to be modified to take into account the corners: the only correction is that the integral of the curvature must now be performed over the smooth part of $ \partial \Omega $ and each corner yields a contribution proportional to its opening angle 
	\bdm
		\int_{\partial \Omega_{\mathrm{smooth}}} \diff \ss \: \curv + \sum_{j =1}^N (\pi - \beta_j) = 2 \pi.
	\edm
	Therefore, one can think of the above identity as if each corner contributes to the mean curvature with a Dirac mass multiplied by $ \pi - \beta_j $ and the integral is meant in distributional sense, i.e., formally replacing the curvature $ \mathfrak{K}(\ss) $ with
	\bdm
		\mathfrak{K}(\ss) + \sum_{j =1}^N (\pi - \beta_j) \delta(\ss_j),
	\edm
	which,  if substituted on the r.h.s. of \eqref{eq: smooth gb}, yields
	\bdm
		-\ecorr \int_{\partial \Omega_{\mathrm{smooth}}} \diff \ss \: \curv - E_{\mathrm{corr}} \sum_{j =1}^N (\pi - \beta_j).
	\edm
	After a direct comparison with the asymptotics proven in \cref{teo: gle asympt}, i.e.,
	\bdm
		-\ecorr \int_{\partial \Omega_{\mathrm{smooth}}} \diff \ss \: \curv + \sum_{j =1}^N E_{\mathrm{corner}, \beta_j}
	\edm
	it is then very natural to state the conjecture below. Note that, if true, the conjecture would imply that the next-to-leading order term in the GL energy expansion would always be given by $ - 2\pi E_{\mathrm{corr}} $, irrespective of the presence of corners.

	\begin{conj}[Corner energy]
		\label{conj: ecorn}
		\mbox{}	\\
		For any $ 1 < b < \theo^{-1} $ and $ \beta \in (0,2\pi) $, one has
		\beq
			\label{eq: ecorn conj}
			\ecorn = -(\pi - \beta) \ecorr.
		\eeq
	\end{conj}

	\begin{remark}[Acute/obtuse angles]
		\mbox{}	\\
		In the linear case, i.e., for a magnetic Schr\"{o}dinger operator with uniform magnetic field in an infinite wedge, it is expected \cite[Rmk. 1.1]{BNF} and numerically verified \cite{ABN,BNDMV} that the ground state energy changes for acute or obtuse angles: for the former it is a strictly increasing function of the angle, which equals $ \theo $ for flat angles, while it is believed to remain constantly equal to $ \theo $ for any obtuse angle. On the opposite, in the nonlinear case, the above Conjecture would provide the same expression for acute and obtuse angles.
	\end{remark}

	As already anticipated, we prove in \cite{CG2} that in a wedge with opening angle $ \pi - \delta $, $ 0 < \delta \ll 1 $, the corner energy is given by
	\beq
		\ecorn = -\delta \ecorr + \OO(\delta^{4/3}|\log\delta|) + \OO(\ell^{-\infty}),
	\eeq
	i.e., it coincides to leading order in $ \delta $ with the conjectured expression. Furthermore, this also shows that the corner energy $ \ecorn $ is non-trivial, at least for angles close to the flat one.

\bigskip

\begin{footnotesize}
\noindent\textbf{Acknowledgments.} The authors are thankful to \textsc{S. Fournais} and \textsc{N. Rougerie} for useful comments and remarks about this work. The support of MIUR through the FIR grant 2013 ``Condensed Matter in Mathematical Physics (Cond-Math)'' (code RBFR13WAET) and of the National Group of Mathematical Physics (GNFM--INdAM) through Progetto Giovani 2016 ``Superfluidity and Superconductivity'' and Progetto Giovani 2018 ``Two-dimensional Phases'' is  acknowledged. The authors are especially grateful to the Institut Mittag-Leffler, where part of this work was completed.
\end{footnotesize}

\section{Corner Effective Problems}
\label{sec: corner effective}

This section is mainly devoted to the proof of \cref{pro: ecorn}, i.e., the existence of the limit defining the effective energy contribution of each corner, and the discussion of the properties of such a limit. For later convenience, we also study another minimization problem in $ \corn $ with different boundary conditions and show that it asymptotically provides the same effective energy (\cref{pro: DN corner}).

\subsection{Surface superconductivity in a finite strip}
\label{sec: strip}

We start by studying a simple minimization problem in a finite strip. Similar problems have already been studied in \cite{Pa,AH,CR1}, taking into account the limit of an infinite strip. Here, instead, the focus is more on boundary conditions and their effect on the ground state energy. We are going to apply the corresponding obtained results to the minimization in \eqref{eq: ecorn} to derive \cref{pro: DN corner}. 

After a local gauge transformation and blow-up on a scale $ \eps $, the leading contribution to the GL energy of a portion of the boundary layer of $ \Omega $ of tangential length $ \eps L $ and normal length $ \eps \ell $, suitably far from any corner, is {(see, e.g., \cite[Lemmas 2 \& 4]{CR1})}
\beq
	\label{eq: stripf}
	\G\lf[\psi; \rect\ri] : = \int_0^{L} \diff s \int_{0}^{\ell} \diff t \: \lf\{ \lf| \partial_t \psi \ri|^2 + \lf| \lf(\partial_s - i t\ri) \psi \ri|^2 - \frac{1}{2b} \lf( 2 |\psi|^2 - |\psi|^4 \ri) \ri\},
\eeq
where $ L,\ell > 0 $, $ b \in (1, \theo^{-1}) $ and $\rect $ stands for the rectangle
\beq
	\label{eq: strip}
	\rect : = \lf[ 0, L \ri] \times \lf[ 0, \ell \ri],	\qquad		\mbox{with }
	\ell  \gg 1.
\eeq
We study two simple minimization problems associated to the energy \eqref{eq: stripf}. First, we set
\beq
	\label{eq: stripe D}
	\ed(\rect) : = \inf_{\psi \in \domd(\rect)} \G\lf[\psi; \rect\ri],
\eeq
and denote by $ \psid $ any corresponding minimizer. The minimization domain is given by
\bml{
	\label{eq: strip domd}
	\domd(\rect) : = \lf\{ \psi \in H^1(\rect) \: \big| \: 
	\psi(0,t)  = \fol(t), \psi(L,t) = \fol(t) e^{-i \alpha_0 L},	\ri.	\\
	\lf. \psi(s,\ell) = f_0(\ell) e^{-i \alpha_0 s} \ri\},
}
where the boundary conditions are meant in trace $ H^{1/2}$-sense and we recall that $ \fol , \al_0 $ is a minimizing pair (see also \cref{sec: 1d curvature}) of \eqref{eq: eoneo}. The label $ \mathrm{D} $ stands for the Dirichlet-type conditions at
$ s = 0 $ and $ s = L $. The heuristic meaning of such conditions is the following:
\begin{itemize}
	\item on the boundary between the surface and the bulk region, i.e., for $ t = \ell $, the order parameter is exponentially small and the same holds true for $ \fol(\ell) $, so the contribution of the boundary conditions there is expected to be negligible; {for this reason we could as well have set $ \psi = 0 $ at $ t = \ell $, but this would make the analysis more complicated;}
	\item at the normal boundaries $ s = 0 $ or $ s = L $, the order parameter is set equal to the ideal minimizer (see \cref{sec: heuristics});
	\item no condition is set on the boundary $ t = 0 $, which is meant to coincide with a blow-up of a portion of $ \partial \Omega $: this is crucial to capture the key features of surface superconductivity and leads to Neumann conditions along the line $ t = 0 $.
\end{itemize}
	{By setting $ \psi : = \chi + f_0(t) e^{-i\alO s} $, one can reduce the variational problem 
	\eqref{eq: stripe D} to the minimization of a functional of $ \chi $ with zero Dirichlet conditions on the boundaries $ s = 0, L $ and $ t = \ell $. This easily allows to deduce}
	 (see, e.g., \cite[Chapt. 4]{G}) {the existence of} a minimizer, {its smoothness} and {the fact that any minimizer solves} 
\beq
	\label{eq: var eq strip}
	- \lf( \nabla - i t \ev_s \ri)^2 \psi = \frac{1}{b} \lf(1 - \lf|\psi\ri|^2 \ri) \psi.
\eeq

The alternative version of \eqref{eq: strip domd} is provided by a modification of the energy: we define
\beq
	\label{eq: stripft}
	\wG\lf[\psi; \rect\ri] : = \G[\psi; \rect] - \int_0^{\ell}\diff t\, \frac{\Fol(t)}{\fol^2(t)} j_t\lf[\psi\ri] \bigg\vert_{s=0}^{s=L},
\eeq
where $ \Fol $ is the potential function (see also \cref{sec: 1d no curv})
\beq
	\label{eq: Fol}
	\Fol(t) : = 2\int_0^t\diff\eta\, \lf(\eta + \al_0 \ri) \fol^2(\eta), 
\eeq
and $ j_t $ is the normal component of the current $ \jv[\psi] $ given in \eqref{eq: current}, i.e., $ j_t[\psi] = \frac{i}{2} \lf( \psi  \partial_t  \psi^* - \psi^*  \partial_t  \psi  \ri)$. The boundary terms appearing in \eqref{eq: stripft} are non-trivial only if the phase of $ \psi  $ varies  along the normal to the boundary, which is obviously not the case for the reference function $ \fol(t) e^{- i\alpha_0 s} $. The reason why such terms have been added to the energy will become clear later (see the proof of \cref{pro: strip} and in particular \eqref{eq: en Ett0}). The minimization of \eqref{eq: stripft} is performed on a domain without constraints on the boundaries $ s = 0 $ and $ s = L $, i.e., we set
\beq
	\label{eq: stripe N}
	\en(\rect) : = \inf_{\psi \in \domn(\rect)} \widetilde{\G}\lf[\psi; \rect\ri],
\eeq
where
\beq
	\label{eq: strip domn}
	\domn(\rect) : = \lf\{ \psi \in H^1(\rect) \: \big| \: \psi(s,\ell) = f_0(\ell) e^{-i\alpha_0 s} \ri\}
\eeq
and we denote by $ \psin $ any corresponding minimizer, which enjoys the same properties as $ \psid $, except for conditions of magnetic Neumann-type at $ s = 0 $ and $ s = L$, i.e.,
\beq
	\label{eq: neumann conditions}
	\lf. \lf[ \lf( \partial_s + i \alpha_0 \ri) \psin - i \frac{F_0(t)}{f_0^2(t)} \lf( \partial_t \psin \ri) \ri] \ri|_{s = 0, L} = 0.
\eeq

	The surface superconductivity behavior occurs for $ 1 < b < \theo^{-1} $ and is characterized by the emergence of the 1D effective model \eqref{eq: eoneo} or, equivalently, \eqref{eq: eones}.

	\begin{pro}[GL energies on a finite strip]
		\label{pro: strip}
		\mbox{}	\\
		For any $ 1 < b < \theo^{-1} $ and $ L > 0 $, as $ \ell \to \infty $,
		\beq
			\label{eq: strip energies}
			\edn(\rect) = L \lf( \eoneo(\ell) + \exl \ri)
			= L \lf( \eones + \exl) \ri).
		\eeq
	\end{pro}

	\begin{remark}[Boundary conditions]
		\label{rem: boundary conditions}
		\mbox{}	\\
		The boundary condition $ \psin(s,\ell) = f_0(\ell) e^{-i\al_0 s} $ is needed for the asymptotics \eqref{eq: strip energies} to hold true. The reason is that otherwise one would get an additional energy contribution from the boundary $ t = \ell $, i.e., the energy would be twice the value appearing in \eqref{eq: strip energies}. Indeed, without the condition at $ t = \ell $, exploiting the gauge invariance of \eqref{eq: stripf} and replacing $ \psi, - t \ev_s $ with $ \psi^* e^{i \ell s}, -(\ell - t) \ev_s $, one can exchange the boundaries $ t = 0 $ and $ t = \ell $, leaving the energy unaffected. 
	\end{remark}

	\begin{proof}
		We first observe that the last estimate is in fact stated in \cref{lem: eoneo approx eones} in  \cref{sec: 1d no curv}. The rest of the statement is actually proven by showing separately that the first estimate holds true for both $ \ed $ and $ \en $.

		Let us first consider $ \ed(\rect) $. For the upper bound, we test $ \G $ on the trial state $ \fol(t) e^{-i \alpha_0 s} $, which immediately yields $  \ed(\rect) \leq L \eoneo(\ell)  $.
		For the corresponding lower bound we use the same energy splitting used, e.g., in \cite{CR2}, i.e., we set
		\beq
			\label{eq: psi splitting strip}
			\psid(s,t) = : \fol(t) e^{-i\alpha_0 s} u(s,t),
		\eeq
		which, via an integration by parts and the variational equation \eqref{eq: fk var}, leads to
		\beq
			\label{eq: en splitting}
			\ed(\rect) = L \eoneo + \mathcal{E}_0\lf[u;\rect\ri],
		\eeq
		where
		\beq
			\label{eq: en E0}
			\mathcal{E}_0\lf[u;\rect\ri]:=\int_0^{L}\diff s\int_0^{\ell}\diff t\, \fol^2 \lf\{\lf |\nabla_{s,t} u \ri|^2 -  2(t+\alpha_0)j_s[u] +\frac{\fol^2}{2b}(1-|u|^2)^2\ri\},
		\eeq
		and $ j_s[\psi] $ is the tangential component of \eqref{eq: current}, i.e., explicitly $ j_s[\psi] = \tx\frac{i}{2} \lf( \psi  \partial_s  \psi^* - \psi^*  \partial_s  \psi \ri) $.
		We stress that the decoupling does not generate any boundary term because $\fol^{\prime} $ vanishes both at $ t = 0 $ and $ t = \ell $ by \eqref{eq: fol nbc}: the only non-trivial computation is the following integration by parts
		\bdm
			\int_0^{L}\diff s\int_0^{\ell}\diff t\,  \lf[ |u|^2 {f_0^{\prime}}^2 + f_0 f_0^{\prime} \partial_t \lf| u \ri|^2  \ri]= - \int_0^{L}\diff s\int_0^{\ell}\diff t\, |u|^2 f_0 f_0^{\prime\prime},
		\edm
		where $ f_0^{\prime\prime} $ can then be replaced via the variational equation \eqref{eq: fk var}.

		{The key ingredient} to bound from below $ \E_0[u] $ {is} the pointwise positivity of the cost function  (see {\eqref{eq: Kol} and} \eqref{eq: kol positive} in \cref{sec: 1d no curv})
		\beq
			\label{eq: K0}
			K_0 : =  \fol^2 + \Fol,	
		\eeq 
		in $ I_{\bar{\ell}} = [0,\bar\ell] $ given by \eqref{eq: annol} (recall that $ \bar{\ell} = \ell + \OO(1) $ by \eqref{eq: barell}).
		Indeed, we integrate by parts twice:
		\bml{
			\label{eq: int by parts}
			 -2 \int_0^{L}\diff s\int_0^{\ell} \diff t \:(t+\alpha_0) f_0^2(t) j_s[u]  = - \int_0^{L}\diff s\int_0^{\ell} \diff t \:  \Fol^{\prime}(t)   j_s[u]
			 = \int_0^{L}\diff s\int_0^{\ell} \diff t \:  \Fol(t) \partial_t j_s[u] \\
			 = 2 \int_0^{L}\diff s\int_0^{\ell} \diff t\, \Fol(t) \, \Im \lf( \partial_t  u^*   \partial_s u \ri) + \int_0^{\ell}\diff t\, \Fol(t)j_t[u]\bigg\vert_{s=0}^{s=L},
		}
		where the boundary terms of the first integration by parts {vanish}, because $ \Fol(0) = \Fol(\ell) = 0 $, and the last terms vanish as well, since, due to boundary conditions, $ u(0,t) = u(L,t) = 1 $ and thus $ j_t[u] = 0 $ there.

		Using \eqref{eq: fol bound} and the simple bound $ 2 |\Im  (a b) | \leq |a|^2 + |b|^2 $, one then obtains as in \cite[Eq. (4.38)]{CR1} (see also \cite[Sect. 2.3 \& Proof of Prop. 4.2]{CR1})
		\bml{
			\label{eq: lb E0}
			\E_0\lf[u; \rect \ri] \geq  \int_0^{L}\diff s\int_0^{\bar\ell}\diff t\,\lf\{ K_0(t) \lf( |\partial_s u|^2+|\partial_t u|^2 \ri) + \frac{1}{2b}f_0^4 (1-|u|^2)^2 \ri\} \\
			+ \int_0^{L}\diff s\int_{\bar{\ell}}^{\ell} \diff t\,  \lf\{ f_0^2 \lf| \nabla u \ri|^2 + 2 \Fol(t) \, \Im \lf( \partial_t  u^*   \partial_s u \ri) \ri\} \\
			\geq \int_0^{L}\diff s\int_{\bar{\ell}}^{\ell} \diff t\,  \lf\{ f_0^2 \lf| \nabla u \ri|^2 + 2 \Fol(t) \, \Im \lf( \partial_t  u^*   \partial_s u \ri) \ri\},
		}
		by \eqref{eq: kol positive} and the positivity of the last term {on the r.h.s. of the first line}. 		
		It thus remains to estimate the quantity on the r.h.s. of \eqref{eq: lb E0} above, which can be done by integrating by parts back:
		\bml{
			\label{eqp: kinetic interval}
			\int_0^{L}\diff s\int_{\bar{\ell}}^{\ell} \diff t\,  \lf\{ f_0^2 \lf| \nabla u \ri|^2 + 2 \Fol(t) \, \Im \lf( \partial_t  u^*   \partial_s u \ri) \ri\} 	\\
			=  \int_0^{L}\diff s\int_{\bar{\ell}}^{\ell} \diff t \: \lf\{ f_0^2 \lf| \nabla u \ri|^2  - 2 (t+\alpha_0) j_s[f_0 u] \ri\} - 2 F_0(\bar{\ell}) \int_0^{L}\diff s \: j_s[u]\bigg\vert_{t = \bar\ell}.
			}
		Now, exploiting \eqref{eq: fkprime decay} and the fact that $ \bar\ell = \ell + \OO(1) $, we deduce that
		\beq
			\fol(t) = \exl,		\qquad		\fol^{\prime}(t) = \exl,		\qquad		\mbox{for any } t \geq \bar\ell.
		\eeq
		Hence, $ \lf| \nabla \psid \ri| = f_0 \lf| \nabla u \ri| + \exl $ in $ \iell \setminus \iellb $. Now, since $ F_0(\ell) = 0 $, $ F_0(\bar\ell) \leq C \ell f^2_0(\bar\ell) $, we can bound the boundary term (last term in \eqref{eqp: kinetic interval}) by 
		\bdm
			C L \sup_{s \in [0,L]} \lf| \psid(s,\bar{\ell}) \ri| \lf| \nabla \psid(s,\bar\ell) \ri| = L \exl,
		\edm
		thanks to \eqref{eq: point agmon strip} and the bound $ \lf\| \nabla \psid \ri\|_{\infty} \leq C $ on the gradient of $ \psid $ (see \eqref{eq: est grad infty}). For the same reason, the first term on the r.h.s. of \eqref{eqp: kinetic interval} can be bounded from below via Cauchy inequality and \eqref{eq: agmon strip} by
		\bdm
			- C \int_0^{L}\diff s\int_{\bar{\ell}}^{\ell} \diff t\, \lf(t + \al_0 \ri)^2 \lf| \psid \ri|^2 = \OO(L \ell^{-\infty}),
		\edm
		which finally yields,
		\beq
			\int_0^{L}\diff s\int_{\bar{\ell}}^{\ell} \diff t\,  \lf\{ f_0^2 \lf| \nabla u \ri|^2 + 2 \Fol(t) \, \Im \lf( \partial_t  u^*   \partial_s u \ri) \ri\} = \OO(L \ell^{-\infty}),
		\eeq
		and thus the statement.

		The proof for the modified functional \eqref{eq: stripft} is very similar. The upper bound is obtained by evaluating the energy on the trial state $ f_0(t) e^{-i\alpha_0 s} $: notice that the phase of such a function is independent of $ t $, then the normal component $ j_t $ of its current is identically zero and therefore the boundary terms in $ \widetilde{\G} $ do not yield any additional contribution. 
		The final outcome is the very same bound $ E_{\mathrm{N}}(\rect) \leq L \eoneo(\ell) $ as before.

		One can then apply the splitting technique, setting (for a different $ u $ than before)
		\beq
			\label{eq: psi splitting strip N}
			\psin(s,t) = : f_0(t) e^{-i\alpha_0 s} u(s,t),
		\eeq
		to get the identity $ \en(\rect) = L \eoneo + \Ett_0[u;\rect] $, where
		\beq
			\label{eq: en Ett0}
			\Ett_0[u;\rect] := \E_0[u; \rect] - \int_0^{\ell}\diff t\, \Fol(t)j_t[u]\bigg\vert_{s=0}^{s=L}.
		\eeq
		The proof of the lower bound is then completely analogous to the one above: the only nontrivial observation is that the first integration by parts in \eqref{eq: int by parts} generates the same outcome, because of the vanishing of $ F_0 $ at the boundaries, and the last terms in \eqref{eq: int by parts} are exactly compensated by the boundary terms in the functional {\eqref{eq: en Ett0}}, so that they sum up to zero. Actually, this was the main reason to add those terms to \eqref{eq: stripft} in first place. The lower bound then follows from the positivity of $ K_0 $, exactly as above. 
	\end{proof}
	
	{A straightforward adaptation of the above arguments leads to the following result on a modified problem with twisted boundary conditions, which is going to play a role later.}
	
	{
	\begin{pro}[GL energy with twisted boundary conditions]
		\label{pro: twisted}
		\mbox{}	\\
		Let $ \varkappa \in [0,2\pi) $, $ 1 < b < \theo^{-1} $ and $ L > 0 $. Let also 
		\beq
			\edk(\rect)  : = \inf_{\psi \in \domk(\rect)} \G[\psi; \rect],	
		\eeq
		\beq
			\domk(\rect)  : = \lf\{ \psi \in \domn(\rect) \: \big| \: \psi(0,t) = f_0(t) e^{i \varkappa}, \psi(L,t) = f_0(t) e^{-i\alO L} \ri\}.
		\eeq
		Then, as $ \ell \to \infty $,
		\beq
			\label{eq: strip energies twisted}
			\eoneo(\ell) L + \OO(L \ell^{-\infty}) \leq \edk(\rect) \leq \eoneo(\ell) L + \frac{C}{L}.
		\eeq
	\end{pro}}

	{\begin{proof}
		The lower bound is obtained via the splitting technique and the positivity of the cost function as discussed in the proof of \cref{pro: strip}. For the upper bound it is sufficient to test the functional on the trial state
		\bdm
			f_0(t) e^{- i \alO s} e^{i \frac{\varkappa (L - s)}{L}},
		\edm
		and recall the optimality of the phase $ \alO $ yielding \eqref{eq: optimal alk}.	
	\end{proof}}

	We conclude this section with a result which will be used later in the paper. In extreme synthesis it states that, if one has an a priori upper bound on $ \E_0[u,\rect] $, then it is possible to extract some useful information on the corresponding order parameter $ \psi(s,t) $ and show for instance that it is pointwise close to $ f_0(t) e^{- i \alpha_0 s} $ up to a smooth phase factor.

	\begin{pro}[Order parameter estimates]
		\label{pro: point est psi}
		\mbox{}	\\
		Let $ \psi $ be a solution of \eqref{eq: var eq strip} in the strip $ \rect $, with $ \ell \geq t_0 > 0 $ and $ L > 0 $, satisfying the boundary conditions in \eqref{eq: strip domn} and \eqref{eq: neumann conditions}, and let $ u $ be defined as in \eqref{eq: psi splitting strip N}. Let also $ \Ett_0[u; \rect] $ be the functional defined in \eqref{eq: en Ett0} in the strip $ \rect $ and assume that
		\beq
			\label{eq: apriori bound Et0}
			\Ett_0[u;\rect] \leq \ee \leq 1,
		\eeq
		for some $ \ee > 0 $. Then, if $  1 < b < \theo^{-1} $,
		\beq
			\label{eq: kinetic u}
			{\lf\| f_0^2 \nabla u \ri\|_{L^2(\rect)}^2 \leq C \ee + \OO(L \ell^{-\infty}).}
		\eeq
		{Moreover, for any $ 0 < T \leq \bar\ell $, there exists a finite constant $ C > 0 $, such that}
				\beq
					\label{eq: point est psi}
					\lf| \lf|\psi(s,t)\ri| - f_0(t) \ri| \leq \frac{{C\ee^{1/4} + \OO(L\ell^{-\infty})}}{\sqrt{\min_{[0,T]} f_0}},	\qquad		\mbox{for any } (s,t) \in R_{L,T},
				\eeq
				\beq
					\label{eq: gradient s boundary}
					\lf. \lf| \partial_s \int_0^\ell \diff t \: \lf| \psi \ri|^2 \ri| \ri|_{s = L} \leq {C \lf\{ \ee  + \sqrt{\ee L }+ \frac{1}{L} \lf[ \frac{\ee^{1/4} + \OO(L\ell^{-\infty})}{\sqrt{\min_{[0,T]} f_0}} + e^{-c(b)T} \ri]  + L \ri\}.}
				\eeq
	\end{pro}
	
	\begin{proof}
		{Applying elliptic regularity theory to the equation satisfied by $ \psi$ one can prove as  in \cref{lem: est gradient infty} that}
		\beq
			\label{eqp: grad est infty}
			\lf\| \nabla \psi \ri\|_{L^{\infty}(\rect)} \leq C,
		\eeq
		Furthermore, $ \psi $ satisfies the Agmon estimates \eqref{eq: agmon strip} and \eqref{eq: point agmon strip}.
		
		The key estimate is then the positivity of the cost function $ K_0 $ in $ \iellb $, as well as the lower bound given by \eqref{eq: Kol coercive}, i.e., 
		\beq
			\label{eqp: K0 lb}
			K_0(t) = f_0^2(t) + F_0(t) \geq {c_b f_0^4(t)},	\qquad		\mbox{for any } t \in \iellb.
		\eeq
		Indeed, by acting as in the proof of \cref{pro: strip}, one immediately gets
		\bml{
			\label{eqp: point est u 0}
			\Ett_0[u, \rect] \geq \int_{R(L,\bar\ell)} \diff s \diff t \:  K_0(t) \lf|\nabla u \ri|^2 {+ \frac{1}{2b} \int_{R(L, \ell)} \diff s \diff t \:   f_0^4(t) \lf(1 - \lf|u\ri|^2 \ri)^2} \\
			+ \OO(L \ell^{-\infty}).
		}
		Plugging in \eqref{eqp: K0 lb} above{, one obtains \eqref{eq: kinetic u} and  
		\beq
			\label{eqp: u close to 1}
			\int_{R(L, \ell)} \diff s \diff t \:   f_0^4(t) \lf(1 - \lf|u\ri|^2 \ri)^2 \leq C \ee + \OO(L\ell^{-\infty}).
		\eeq}
		
		We now address \eqref{eq: point est psi}: the starting point is {provided by \eqref{eqp: u close to 1}}, which essentially implies that $ |u| $ is approximately constant and equal to $ 1 $. The idea of proof goes back to \cite{BBH2} and it has been used several times since then (see, e.g., \cite{CRY}). Fix $  0 < T \leq \bar\ell $ and assume by contradiction that there was a point $ (s_0,t_0) \in R_{L,T}  $, where
		\beq
			\lf| \lf|\psi(s_0,t_0)\ri| - f_0(t_0) \ri| \geq \frac{c \: \bar{\ee}^{1/4}}{\sqrt{\min_{[0,T]} f_0}},
		\eeq
		for suitable $ c > 0 $ and $ \bar{\ee} \geq \ee $ to be adjusted later. Then, by \eqref{eqp: grad est infty} and the analogous bound for $ |f_0^{\prime}(t)| $ (see \eqref{eq: fkprime decay}), we deduce that there would exist also a ball of radius $ \varrho : = c'  \bar{\ee}^{1/4}/\sqrt{f_0(t_0)} $ centered in $ (s_0,t_0) $, with $ c' $ a constant proportional to $ c $ and depending only on the a priori bounds on the gradients, so that
		\bdm
			\lf| \lf|\psi(s,t)\ri| - f_0(t) \ri| \geq \frac{1}{2}\frac{c \:  \bar{\ee}^{1/4}}{\sqrt{\min_{[0,T]} f_0}},	\qquad		\mbox{in } \B_{\varrho}(s_0,t_0) \cap R_{L,T}.
		\edm
		Furthermore, we can also assume that at least one quarter of the ball is contained inside $ R_{L,T} $. Hence,
		\bmln{
			\int_{R(\ell,T)} \diff s \diff t \: f_0^4(t) \lf(1 - \lf|u\ri|^2 \ri)^2 = \int_{R(\ell,T)} \diff s \diff t \: \lf(f_0^2(t)  - \lf|\psi\ri|^2 \ri)^2 \\
			\geq \frac{\pi c^2}{16} \varrho^2  \bar{\ee}^{1/2} \min_{[0,T]} f_0 \geq C c^4 \bar{\ee}
		}
		where $ C $ is a positive constant independent of $ c $. Therefore, by taking $ c $ large enough and $ \bar{\ee} = \ee + \OO(L \ell^{-\infty}) $, we would get a contradiction with \eqref{eqp: u close to 1}, which completes the proof.

		In order to finally get \eqref{eq: gradient s boundary}, we can restrict the integration to the interval $ t \in [0, \bar\ell] $, since the rest is exponentially small. We then compute
		\beq
			\label{eqp: integration of the boundary term}
			  \partial_s \int_0^{\bar\ell} \diff t \: \lf| \psi(s,t) \ri|^2 \bigg|_{s = L} = \int_0^{L} \diff s  \int_0^{\bar\ell} \diff t \: \lf[ \chi(s) \partial_s^2 \lf| \psi(s,t) \ri|^2 + \chi^{\prime}(s) \partial_s \lf| \psi(s,t) \ri|^2 \ri],
		\eeq
		{for any smooth $ \chi  $ such that $ \chi(L) = 1 $. Taking $ \chi(s) = s/L $, we get}
		\bmln{
			\bigg| \int_0^{\delta} \diff s  \int_0^{\bar\ell} \diff t \: \chi^{\prime}(s) \partial_s \lf| \psi(s,t) \ri|^2 \bigg| \leq {\frac{1}{L} \int_0^{\bar{\ell}} \diff t \: \lf| \lf| \psi(\delta,t)\ri|^2 - \lf| \psi(0,t) \ri|^2 \ri|} \\
			{\leq \frac{C}{L} \lf[ \frac{\ee^{1/4} + \OO(L\ell^{-\infty})}{\sqrt{\min_{[0,T]} f_0}} + e^{-c(b)T} \ri],}
		}
		{by \eqref{eq: point est psi} and \eqref{eq: point agmon strip}}. For the first term on the r.h.s. of \eqref{eqp: integration of the boundary term}, we extend the integration in $ t $ to $ \ell $: using Neumann boundary conditions at $ t = \ell $, one gets
		\bdm
			\int_{\bar{\ell}}^{\ell} \diff t \: \partial_s^2 \lf| \psi  \ri|^2 = \int_{\bar{\ell}}^{\ell} \diff t \: \Delta \lf| \psi  \ri|^2 + \lf. \partial_t \lf| \psi \ri|^2 \ri|_{t = \bar\ell} = \int_{\bar{\ell}}^{\ell} \diff t \: \Delta \lf| \psi  \ri|^2 + \exl,
		\edm 
		by \eqref{eq: point agmon strip} and the pointwise bound on the gradient of $ \psi $. Hence, exploiting the Neumann conditions also at $ t = 0 $ and \eqref{eq: var eq strip}, we obtain
		\bmln{
			\int_{0}^{\bar\ell} \diff t \: \partial_s^2 \lf| \psi  \ri|^2 = \int_{0}^{\ell} \diff t \: \Delta \lf| \psi \ri|^2 + \exl = \int_{0}^{\ell} \diff t \: \lf[ 2 \Re \lf( \psi^* \Delta \psi \ri) + 2 \lf| \nabla \psi \ri|^2 \ri] + \exl \\
			= 2 \int_{0}^{\ell} \diff t \: \lf[ \lf| {\lf( \nabla - i t \ev_s \ri)} \psi \ri|^2  - \tx\frac{1}{b} \lf(1 - \lf| \psi \ri|^2 \ri) \lf| \psi \ri|^2 \ri] + \exl,
		}
		which yields, after integration in $ s $,
		\bml{
			\label{eqp: point est u 3}
			\int_0^{L} \diff s \int_0^{\ell} \diff t \: \chi(s) \partial_s^2 \lf| \psi(s,t) \ri|^2 \\
			= 2 \int_{\rect} \diff s \diff t \: \chi(s) \lf[ \lf| {\lf( \nabla - i t \ev_s \ri)} \psi \ri|^2  - \tx\frac{1}{b} \lf(1 - \lf| \psi \ri|^2 \ri) \lf| \psi \ri|^2\ri] + \OO(L \ell^{-\infty}).
		}
		In order to estimate the quantity on the r.h.s. of the expression above we observe that
		{\bmln{
			\int_{\rect} \diff s \diff t \: \chi(s) \lf[ \lf| {\lf( \nabla - i t \ev_s \ri)} \psi \ri|^2  - \tx\frac{1}{b} \lf(1 - \lf| \psi \ri|^2 \ri) \lf| \psi \ri|^2\ri] 	\\
			\leq \Ett_0[u] +\frac{1}{2b}\int_{\rect} \diff s \diff t \lf| f_0^4 - \lf| \psi \ri|^4 \ri|+ 2 \int_{\rect} \diff s \diff t \: \lf(\chi(s) - 1 \ri) (t + \alO) f_0^2 j_s[u]	\\
			 \leq C  \lf[ \ee + \sqrt{\ee L}+ L \ri],	
			}}
		{by \eqref{eqp: grad est infty} and \eqref{eqp: u close to 1}.} Altogether we get \eqref{eq: gradient s boundary}.
	\end{proof}

	\subsection{Properties of $  \ecornl $}
	\label{sec: existence}
	
	In the present and following Sections, we study the effective model introduced in \eqref{eq: ecorn} and specifically prove the existence of the limit as well as its boundedness. The key properties we are going to use in the proof of \cref{pro: ecorn} are:
	\begin{itemize}
		\item change of gauge to replace the magnetic potential $ \fv $ with $ \av \simeq - t \ev_s $ (\cref{lem: gauge choice});
		\item uniform boundedness of $ \ecorn(L,\ell) $ and existence of the limit $ L, \ell \to + \infty $ over suitable subsequences (\cref{pro: ecorn bound});
		\item further properties of the effective model and, in particular, its dependence on the boundary conditions (\cref{sec: further properties}).
	\end{itemize}
	
	We recall the corner energy defined in \eqref{eq: ecornl}
	and set 
	\beq
		\label{eq: gfv}
		\G_{\fv}[\psi] : = \glf_{1}\lf[\psi, \fv; \corner \ri];
	\eeq
	\beq
		\label{eq: glecorn doms}
		\glecorn(L,\ell) : =  \inf_{\psi \in \doms(\corner)} \G_{\fv}[\psi],
	\eeq
	where both the energy $ \glf_{1} $, the minimization domain and the corner region are introduced in \cref{sec: main results}. Any corresponding minimizer is denoted by $ \psilf $.	
{Before proceeding further, we introduce an auxiliary problem in $ \corner $, modified by the addition of analogous boundary terms as in \eqref{eq: stripft}. Such a problem will appear in the proof of the main theorem. We set
\beq
	\label{eq: gltfv en corner}
	\Gt_{\fv}[\psi] : = \glf_{1}\lf[\psi, \fv; \corner\ri]  - \int_0^{\ell}\diff t\, \frac{F_0(t)}{f_0^2(t)} j_t\lf[\psi(\rv(s,t)) e^{\frac{i}{2} st} \ri]  \bigg\vert_{s=-L}^{s=L};
\eeq
\beq
	\label{eq: en n fv corner}
	\glecornt(L,\ell) : =  \inf_{\psi \in \widetilde{\dom}_{\star}(\corner)} \Gt_{{\fv}}[\psi],
\eeq
where
\beq
	\label{eq: domns}
	\widetilde{\dom}_{\star}(\corner) : = \lf\{ \psi \in H^1(\corner) \: \big| \: 
	 \lf. \psi\ri|_{\bdi} = \psi_{\star} \ri\},
\eeq
and $ \psi_\star  $ is defined in \eqref{eq: psi star}.
Note that the boundary terms are slightly different than the ones considered in \eqref{eq: stripft}, which is due to the presence of an additional phase in $ \psi_{\star} $ compared to $ f_0 e^{-i \alO s} $, due to the different choice of the magnetic potential.}

	In the next \cref{lem: gauge choice}, we show that the vector potential $ \fv $ can be replaced with $ \av $, such that far from the corners
	\beq
		\label{eq: av simeq}
		\av(\rv(s,t)) \simeq - t \ev_s,
	\eeq
	in boundary coordinates $ (s,t) $. It is not difficult to figure out that there exists no smooth gauge transformation implementing the above change globally in $ \corner $, in particular close to the bisectrix. More precisely, we define the wedge-domain $ \corner \setminus \cornert $ (as depicted in \cref{fig: cornert}) through
	\beq
		\label{eq: cornert}
		\corner \setminus \cornert : = \lf\{ \rv \in \corner \: \big| \: \tx\frac{1}{2} \beta - \tx\frac{1}{\ell^3} \leq \vartheta \leq \tx\frac{1}{2} \beta + \tx\frac{1}{\ell^3} \ri\},
	\eeq
	in polar coordinates $ (\varrho,\vartheta) \in [0,\ell] \times [0,\beta] $. Hence, we obviously have
	\beq
		\label{eq: cornert area}
		\lf| \corner \setminus \cornert \ri| = \OO(\ell^{-1}).
	\eeq
	The potential $ \av $ is thus such that there exists a gauge phase $ \phi_{\fv} \in H^1(\corner) $ so that 
	\beq
		\label{eq: av gauge}
		\av = \fv + \nabla \phi_{\fv},
	\eeq
	and
	\beq
		\label{eq: av}
		\av = - t \ev_s,	\qquad		\mbox{in } \cornert.
	\eeq
	As already explained, because of the jump of $ \ev_s $ along the bisectrix of the sector, one can not set $ \av = - t \ev_s $ everywhere. However, we require that 
	\beq
		\label{eq: av bound}
		\av = \OO(\ell^4),	\qquad		\mbox{in } \corner,
	\eeq
	which is in fact a constraint only in $ \corner \setminus \cornert $. In next \cref{lem: gauge choice} we investigate the existence of such a phase $ \phi_{\fv} $. Note that 
	\beq
		\label{eq: curl av}
		\curl \: \av  = \curl ( - t \ev_s ) =  1, 	\qquad \mbox{in } {\corner},
	\eeq
	thanks to \eqref{eq: curl} and the gauge invariance of the $ \curl $.

\begin{figure}[!ht]
		\begin{center}
		\begin{tikzpicture}
			\draw (0,0) -- (1,2);
			\draw (1,2) -- (1.5,3);
			\draw (1.5,3) -- (2,4) -- (2.5,3);
			\draw (2.5,3) -- (3,2);
			\draw (3,2) -- (4,0);
			\draw (0,0) -- (1,-0.5);
			\draw (4,0) -- (3,-0.5);
			\draw (1,-0.5) -- (2,1.7);
			\draw (3,-0.5) -- (2,1.7);
			\draw[fill=black!20!white] (2,4) -- (1.6,0.82) -- (2,1.7) --  (2.4,0.82) -- (2,4);
			\node at (2,4.5) {{\footnotesize $V$}};
			\node at (-0.5, 0) {{\footnotesize $A$}};
			\node at (4.5, 0) {{\footnotesize $B$}};
			\node at (1.5, -0.5) {{\footnotesize $C$}};
			\node at (2.5, -0.5) {{\footnotesize $E$}};
			\node at (2,1) {{\footnotesize $D$}};
		\end{tikzpicture}
		\caption{The region $ \corner \setminus \cornert $ (shaded area).}\label{fig: cornert}
		\end{center}{}
	\end{figure}

	\begin{lem}[Gauge choice]
		\label{lem: gauge choice}
		\mbox{}	\\
		For any $ L, \ell > 0 $ satisfying \eqref{eq: L ell condition} and so that
		\beq
			\label{eq: integer condition}
			\frac{\lf|\corner \ri|}{2\pi \lf| \partial \corner \ri|} \in \Z,
		\eeq
		there exists a vector potential $ \av  \in L^{\infty}(\corner) $ and a phase $ \phi_{\fv} \in H^1(\corner) $ satisfying \eqref{eq: av gauge}, \eqref{eq: av} and \eqref{eq: av bound}, such that
		\beqn
			\label{eq: gauge choice}
			\inf_{\psi \in \doms(\corner)} \G_{\mathbf{F}}[\psi]  &=& \inf_{\psi \in \mathscr{D}_{\mathrm{D}}(\corner)} \G_{\mathbf{a}}[\psi],	\\
			\label{eq: gauge choice N}
			\inf_{\psi \in \widetilde{\dom}_{\star}(\corner)} \Gt_{{\fv}}[\psi] &=& \inf_{\psi \in \mathscr{D}_{\mathrm{N}}(\corner)} \Gt_{\mathbf{a}}[\psi],
		\eeqn
		where
		\beqn
			\label{eq: domd}
			\domd(\corner) & : = & \lf\{ \psi \in H^1(\corner) \: \big| \: \lf. \psi\ri|_{\bdbd \cup \bdi} = \psi_{0} \ri\},	\\
	 		\domn(\corner) & : = & \lf\{ \psi \in H^1(\corner) \: \big| \: 
	 \lf. \psi\ri|_{\bdi} = \psi_{0} \ri\},
		\eeqn
		\beq		
	 		\label{eq: psi0}
	 		 {\psi_0(s,t) : = f_0(t) e^{- i \al_0 s}.}
		\eeq
	\end{lem}
	
	\begin{remark}[Constraint on $ L, \ell $]
		\label{rem: integer condition}
		\mbox{}	\\
		The condition \eqref{eq: integer condition} {reads $ L - \frac{2 \ell}{\tan \beta} = c(\beta,L \ell) \mathbb{Z} $ where $ c(\beta,\ell) $ is uniformly bounded as $ {L},\ell \to + \infty ${. More precisely
		\bdm
			c(\beta,L,\ell) \xrightarrow[L \to +\infty]{} 0,	\qquad		c(\beta,L,\ell) \xrightarrow[\ell \to + \infty]{} c(\beta),
		\edm
		uniformly in the other parameters}. Hence, given generic $ \ell, L \to + \infty $, it suffices to replace $ L $ with $ L + \OO(1) $ to enforce \eqref{eq: integer condition}.} 
	\end{remark}
	
	\begin{proof}
		The two different minimization problem can be treated in the same way. It suffices to prove the existence of the gauge phase  $ \phi_{\fv} $ and, in order to recover \eqref{eq: av}, we set
		\beq
			\label{eq: phif}
			\phi_{\fv}(s,t) : = - \tx\frac{1}{2}  s t,	\qquad		\mbox{in } \cornert.
		\eeq
		Note that such a phase is actually the same gauge phase used in \cite[Appendix F]{FH1} or in \cite[Eq. (4.7)]{CR2} with vector potential set equal to $ \fv $ and recovers the additional phase factor in the boundary terms in \eqref{eq: gltfv en corner}.
		Such a phase is in $ H^1(\cornert) $ but its definition can not be extended to the whole $ \corner $. We can however continue $ \phi_{\fv} $ arbitrarily in $ \corner \setminus \cornert $, just requiring continuity through the boundary of the region. There are infinitely many ways of doing that and at least one such that the bound \eqref{eq: av bound} is satisfied (e.g., a linear interpolation).
		
		In order to complete the proof, we need to show that $ \psi e^{i \phi_\fv} $ is still a single-valued function. It is not difficult to see \cite[Appendix F]{FH1} that, to this purpose, one has to correct \eqref{eq: phif} by $ \varpi s $, where
		\bmln{
			\varpi = \frac{1}{2\pi \lf| \partial \corner \ri|} \int_{\corner} \diff \rv \: \curl \fv - \lf\lfloor \frac{1}{2\pi \lf| \partial \corner \ri|} \int_{\corner} \diff \rv \: \curl \fv \ri\rfloor \\
			= \frac{\lf| \corner \ri|}{2\pi \lf| \partial \corner \ri|}  - \lf\lfloor \frac{\lf| \corner \ri|}{2\pi \lf| \partial \corner \ri|}  \ri\rfloor,
		}
		where $ \lf\lfloor \: \cdot \: \ri\rfloor $ stands for the integer part. However, by the assumption \eqref{eq: integer condition}, $ \varpi = 0 $ and no additional phase is needed.
	\end{proof}

	From now on we are to going to study only the minimization on the r.h.s. of \eqref{eq: gauge choice} in \cref{lem: gauge choice}, with the vector potential $ \av $ satisfying \eqref{eq: av gauge}, \eqref{eq: av} and \eqref{eq: av bound}. {In order to guarantee that \eqref{eq: integer condition} is satisfied, however, we restrict the analysis to suitable monotone sequences $ \lf\{ \ell_n \ri\}_{n \in \N} $, $ \lf\{ L_n \ri\}_{n \in \N} $, such that}
	{\beq
		\label{eq: subsequence}
		\ell_n, L_n \xrightarrow[n \to +\infty]{} + \infty,	\qquad		\eqref{eq: L ell condition} \mbox{ and } \eqref{eq: integer condition} \mbox{ hold},
	\eeq
	and consider $ \ecorn(L_n,\ell_n) $ in the following}. More precisely, we are going to study the quantity
	\beq
		\label{eq: en d corner}
		\glecorn(L_n,\ell_n) : =  \inf_{\psi \in \domd(\Gamma_{\beta}(L_n,\ell_n))} \G_{\av}[\psi; \Gamma_{\beta}(L_n,\ell_n))].
	\eeq
	Any minimizer of \eqref{eq: en d corner} is denoted by $ \psil $, i.e.,
	\beq
		\label{eq: psil}
		\glecorn(L_n,\ell_n) : =  \G_{\av}\lf[\psil; \Gamma_{\beta}(L_n,\ell_n))\ri].
	\eeq 
	The existence of such a minimizer follows by standard arguments as well as the fact that any $ \psil $ solves the variational equation
	\beq
		\label{eq: var eq corner}
		\begin{cases}
			- \lf( \nabla + i \av \ri)^2 \psil = \tx\frac{1}{b}(1-|\psil|^2)\psil,		&		\mbox{in } \corner,	\\
			\psil = f_0(t(\rv)) e^{-i \alpha_0 s(\rv)},														&		\mbox{on } \bdbd \cup \bdi,	\\
			\nv \cdot \lf(\nabla + i \av \ri) \psil = 0,								&		\mbox{on } \bdo.
		\end{cases}
\eeq
Note that the equation above coincides with \eqref{eq: var eq strip} far from the vertex, where boundary coordinates are well posed and $ \av = - t \ev_s $. We can thus apply to $ \psil $  the results in \cref{lem: est gradient infty}, \cref{lem: agmon 1}, \cref{lem: agmon 2} and \cref{lem: exp estimate}.

	  \subsection{Boundedness and existence of the limit}
	  \label{sec: boundedness}

We start by proving the uniform boundedness of $ \ecornl $ as a function of $ \ell, L $.

	\begin{pro}[Boundedness of $ \ecornln $]
		\label{pro: ecorn bound}
		\mbox{}	\\
		{Let $ \lf\{ \ell_n \ri\}_{n \in \N}, \lf\{ L_n \ri\}_{n \in \N} $ satisfy \eqref{eq: subsequence}.} Then, for any $ 1 < b < \Theta_0^{-1} $, there exists a finite constant $ C < + \infty $ independent of $ n $, such that
		\beq
			\label{eq: ecorn bound}
			\lf| \ecornln \ri| \leq C.
		\eeq
	\end{pro}

	\begin{proof}
		{We first discuss the boundedness from below, which is the most difficult property to prove, and show that}
		\beq
			\label{eq: lb boundedness}
			\ecorn(L_n,\ell_n) \geq - C,
		\eeq
		for some finite $ 0 < C < +\infty $.
		The key tool is a suitable partition of unity, which isolates the region where we want to retain the energy and allow us to discard the rest. We thus consider two smooth positive functions $ \chi $ and $ \eta $, such that $\chi^2 + \eta^2=1$ and whose supports are described, e.g., in \cref{fig: partition}: we assume that $ \eta \equiv 1 $ inside  the shaded area, while $ \chi \equiv 1 $ in the white area.
		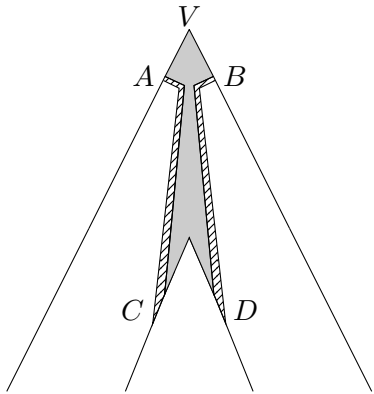
\begin{figure}[ht!]
			\begin{center}
			\begin{tikzpicture}[scale=1.2]
				\draw (1.735, 3.48) -- (0.9,1.8);
				\draw (2.265,3.48) -- (3.1,1.8);
				\draw (0.9,1.8) -- (0.5, 1);
				\draw (3.1,1.8) -- (4,0);
				\draw (1.948,3.38) -- (1.8,1.803);
				\draw (1.8,1.803) -- (1.73, 1.07);
				\draw (2.16,1.33) -- (2,1.7);
				\draw  (1.84, 1.33) -- (1.7, 1);
				\draw (0.5, 1) -- (0,0);
				\draw  (2.16,1.33) -- (2.3, 1);
				\draw (2.3, 1) -- (2.7,0);
				\draw(1.7,1) -- (1.3,0);
				\node at (2, 4.15) {$V$};
				\node at (1.5, 3.48) {$A$};
				\node at (2.5,3.48) {$B$};
				\node at (1.38, 0.9) {$C$};
				\node at (2.62, 0.9) {$D$};
				\draw [fill = black!20!white] (2,4) -- (1.735, 3.48) -- (1.948,3.38) -- (1.8, 1.803)  -- (1.8,1.803) -- (1.73, 1.07) -- (2,1.7) -- (2.27, 1.07) -- (2.052,3.38) -- (2.265,3.48) -- (2,4);
				\draw [pattern=north east lines] (1.735, 3.48) -- (1.948, 3.38) --  (1.8,1.803)  -- (1.73, 1.07) -- (1.6, 0.75)  -- (1.889, 3.35) -- (1.718,3.435);
				\draw [pattern=north east lines] (2.265,3.48) -- (2.052,3.38) -- (2.2,1.803) -- (2.27, 1.07) -- (2.4,0.75) -- (2.111, 3.35) -- (2.282, 3.435);
			\end{tikzpicture}
			\caption{The partition of unity $ \chi $, $\eta $.}
			\label{fig: partition}
			\end{center}
		\end{figure}
		The dashed regions is where the supports of the two functions overlap. We choose the angle $ \widehat{CVD} $ equal to $ \beta/2 $ for concreteness but any angle of order $ 1 $ would work. The distance of the points $ A $ and $ B $ from the vertex $ V $ is also taken of order $ 1 $. Furthermore, the width of the transition regions can also be taken in such a way that
		\beq
			\label{eq: bounds partition}
			|\nabla\chi| = \OO(1),	\qquad |\nabla\eta| = \OO(1).
		\eeq
		
		{The rationale behind the choice of the partition of unity is that the energy contribution coming from the support of $ \chi $ reconstructs the leading term $ 2 L \eoneo(\ell) $, up to an $ \OO(1) $ error, while the rest provides a correction of order $ \OO(1) $. Therefore, the support of $ \chi $ must contain the outer boundary $ \bdo $ up to $ \OO(1) $ regions and the magnetic potential must be equal to $ - t \ev_s $ there. Hence, the area close to the bisectrix is included in the support $ \eta $, because there the magnetic potential is unknown.}

		{The key ingredient of the proof is then the} IMS formula \cite[Thm. 3.2]{CFKS}{, which} yields
		\bml{
			\label{eq: ims corner}
			\glecorn(L_n,\ell_n) = \G \lf[\chi\psil\ri] + \G \lf[\eta\psil\ri]
			 - \int_{ \corner }\diff\rv\, |\nabla\chi|^2|\psil|^2	
			 -\int_{ \corner   } \diff \rv \: |\nabla\eta|^2|\psil|^2		\\
			 {= \G \lf[\chi\psil\ri] + \G \lf[\eta\psil\ri] + \OO(1),}
		}
		{where we have exploited the decay \eqref{eq: decay corner} to bound the contributions on the supports of $ \nabla \chi, \nabla \eta $.} 		
		We now claim that there exists a finite constant independent of $ n  $ so that
		\beqn
			\G_{\av}\lf[\eta\psil\ri]  & \geq &  -C,		\\
			\G_{\av}\lf[\chi\psil\ri] & = & 2 L_n \eoneo + \OO(1), \label{eq: ecorn main energy}
		\eeqn
		which combined with \eqref{eq: ims corner} yields \eqref{eq: ecorn bound}. 
		
		Let us first consider the first estimate above: dropping from the energy all the positive terms, we get
		\beq
			\G_{\av}\lf[\eta\psil\ri] \geq - C \int_{\supp(\eta)} \diff \rv \: \lf| \psil(\rv) \ri|^2 \geq - C,
		\eeq
		by the decay of $ \psil $ as above.
		To complete the proof it remains only to deal with \eqref{eq: ecorn main energy}: since $ \supp(\chi) $ is actually composed of two disconnected sets, denoted by $ T_- $ (on the right of \cref{fig: partition}) and $ T_+ $, we can use boundary coordinates in both regions $ T_{\pm} $. We can then apply the splitting technique described in the proof of \cref{pro: strip} and set
		\beq
			\label{eq: splitting boundedness}
			\chi(s,t) \psil(\rv(s,t)) = :
			\begin{cases}
				f_0(t) e^{-i \alpha_0 s} u_-(s,t),	&	\mbox{in } T_-,	\\
				f_0(t) e^{-i \alpha_0 s} u_+(s,t),	&	\mbox{in } T_+.
			\end{cases}
		\eeq
		The same computation which leads to \eqref{eq: en splitting} yields now (recall \eqref{eq: en E0})
		\beq
			\G_{\av}\lf[\chi\psil\ri] = - \frac{1}{b} \int_{T_- \cup T_+} \diff s \diff t \: f_0^4(t) + \E_0[u_-;T_-] + \E_0[u_+;T_+].
		\eeq
		Finally, as long as $ 1 < b < \theo^{-1} $, one can prove that the energies $ \E_0[u_-;T_-] $ and $ \E_0[u_+;T_+] $ are both positive, exactly as in \eqref{eq: lb E0}, leading to
		\beq
			\label{eq: int Tpm}
			\G_{\av}\lf[\chi\psil\ri] \geq - \frac{1}{b} \int_{T_- } \diff s \diff t \: f_0^4(t) - \frac{1}{b} \int_{T_+ } \diff s \diff t \: f_0^4(t).
		\eeq
		The last step is the estimate of the two integrals on the r.h.s. of \eqref{eq: int Tpm} above: the identity \eqref{eq: eonek identity} and the exponential decay \eqref{eq: fk decay 1} (both with $ k = 0 $)  imply
  		\bdm
  			-\frac{1}{b} \int_{T_\pm}\diff s\diff t \: f_0^4(t) \geq L_n \eoneo - C \int_0^{\ell}\diff t \: t  e^{-2(t+\alpha_0)^2} + \OO(1) \geq L_n \eoneo + \OO(1),
  		\edm
  		which together with \eqref{eq: int Tpm} completes the lower bound proof.
  		
  		{The opposite side of the inequality \eqref{eq: lb boundedness} can be proven by simply using $ \chi f_0(t) e^{-i \alpha_0 s} $ as a trial state (more precisely, setting $ u_{\pm} = 1 $ in \eqref{eq: splitting boundedness}). We omit the calculations, since they are totally analogous to the ones above.}
	\end{proof}

	We are now in position to prove the first important result of this section.
	
	\begin{proof}[Proof of \cref{pro: ecorn}]
		{The first important observation is that $ \ecorn(L,\ell) $ is a monotone non-increasing function of $ L $ and as such it admits a limit. Indeed, for any $ L_a < L_b $, one can easily construct a trial state for the energy in $ \corn(L_b, \ell) $ by extending the minimizer in $ \corn(L_a, \ell) $ and setting the trial state equal to $ f_0(t) e^{-i \alpha_0 s} $ where the minimizer is not defined. The outcome of the trivial computation is the inequality $ \ecorn(L_a, \ell) \leq \ecorn(L_b,\ell) $.}
		
		{Let $ \lf\{ \ell_n \ri\}_{n \in \N}, \lf\{ L_n \ri\}_{n \in \N} $ be two monotone subsequences such that $ \lim_{n \to +\infty} \ell_n = \lim_{n \to +\infty} L_n = + \infty $ and \eqref{eq: subsequence} is satisfied (see \cref{rem: integer condition}). By the monotonicity in $ L $ of the energy and its boundedness, we know that for any $ \eps > 0 $ and any given $ \bar{n} \in \N $, there exists $ \bar{n}_2(\bar{n}) \in \N $, such that}
		\beq
			\label{eq: cauchy 1}
			{\lf| \ecorn(L_n, \ell_{\bar{n}}) - \ecorn(L_m, \ell_{\bar{n}}) \ri| < \tx\frac{1}{3} \eps,}
		\eeq
		{for any $ n,m > \bar{n}_2 $.}
		
		{Furthermore, by the exponential decay of the minimizer and its derivatives \eqref{eq: agmon 2}, one gets}
		\bdm
			{\lf| \ecorn(L_n, \ell_n) - \ecorn(L_n, \ell_m) \ri| \leq C L_n e^{- c \min\lf\{ \ell_n, \ell_m \ri\}}.}
		\edm
		{Hence, if the sequences satisfy the condition}
		\beq
			\label{eq: Ln condition}
			{L_n \leq C \ell_n^a,		\qquad		\mbox{for some } a > 0,}
		\eeq
		{we can conclude that there exists $ \bar{n}_1 \in \N $, such that}
		\beq
			\label{eq: cauchy 2}
			{\lf| \ecorn(L_n, \ell_n) - \ecorn(L_n, \ell_m) \ri| < \tx\frac{1}{3} \eps}
		\eeq
		{for $ n,m > \bar{n}_1 $.}
		
		{In conclusion, we can estimate}
		\bml{
			{\lf|  \ecorn(L_n, \ell_n) - \ecorn(L_m, \ell_m) \ri| \leq \lf|  \ecorn(L_n, \ell_n) - \ecorn(L_n, \ell_{\bar{n}_1+1}) \ri|}	\\
			{ +  \lf| \ecorn(L_n, \ell_{\bar{n}_1+1}) - \ecorn(L_m, \ell_{\bar{n}_1+1})  \ri| } \\
			{+  \lf| \ecorn(L_m, \ell_{\bar{n}_1+1}) - \ecorn(L_m, \ell_m)  \ri| < \eps}
		}
		{for any $ n, m > \max\lf\{ \bar{n}_1, \bar{n}_2(\bar{n_1}+1) \ri\} $, so that the sequence is Cauchy and the limit exists. The independence of the chosen subsequences relies on the uniqueness of the limit, while the uniform boundedness has been proven in \cref{pro: ecorn bound}.}
	\end{proof}

	\subsection{{Neumann and Dirichlet problems} in $ \corner $}
	\label{sec: further properties}

{We are going to study the Neumann problem \eqref{eq: en n fv corner} on the monotone subsequences $ \lf\{ \ell_n \ri\}_{n \in \N} $, $ \lf\{ L_n \ri\}_{n \in \N} $ introduced in the previous \cref{sec: existence}, i.e., such that \eqref{eq: subsequence} holds. Our main goal here is to} show {that, as in the case of the strip,} the Dirichlet and Neumann energies coincide asymptotically as  $ n \to + \infty $.  This is going to play a key role in the proof of our main result{, since it implies the identity
\bml{
			\label{eq: ecorn alt}
			\ecorn  = \lim_{n \to +\infty} \lf( - 2 L \eoneo(\ell_n) + \glecorn(L_n,\ell_n) \ri) \\
			= \lim_{n \to +\infty}  \lf( - 2 L_n \eoneo(\ell_n) + \glecornt(L_n,\ell_n) \ri).
		}}
	
	{Before proving the result we need however a technical lemma on a variational problem with twisted boundary conditions, whose proof is postponed at the end of the section. Let then $ \varkappa \in [0,2\pi) $ as in \cref{pro: twisted} and set
	\beq
			\label{eq: edx}
			E_{\beta,\varkappa}(L,\ell)  : = \inf_{\psi \in \domdx(\corner)} \Gt_{\av}[\psi],	
	\eeq
	\beq
		\label{eq: domdx}
			\domdx(\corner) : = \lf\{ \psi \in H^1(\corner) \: \Big| \:
	 \lf. \psi\ri|_{\bdi \cup \{ s = - L \}} = \psi_{0},  \lf. \psi \ri|_{s = L} = \psi_0 e^{i \varkappa} \ri\}.
	\eeq}

	\begin{lem}
		\label{lem: twisted D energies}
		\mbox{}\\
		{Let $ \lf\{ \ell_n \ri\}_{n \in \N} $, $ \lf\{ L_n \ri\}_{n \in \N} $ be two monotone subsequences such that \eqref{eq: subsequence} holds. Then,
		\beq
			\label{eq: twisted D energies}
			  \glecorn(L_n,\ell_n) = E_{\beta,\varkappa}(L_n,\ell_n) +  o_n(1).
		\eeq}
	\end{lem}

	\begin{pro}[{Dirichlet and Neumann energies}]
		\label{pro: DN corner}
		\mbox{}\\
		{Let $ \lf\{ \ell_n \ri\}_{n \in \N} $, $ \lf\{ L_n \ri\}_{n \in \N} $ be two monotone subsequences such that \eqref{eq: subsequence} holds. Then,} for any $ 1 < b < \Theta_0^{-1} $,
		\beq
			\label{eq: DN corner}
			\glecorn(L_n,\ell_n) - \glecornt(L_n,\ell_n)  = {o_{n}(1)}.
		\eeq
	\end{pro}

	\begin{proof}
		{In view of the vanishing of the boundary terms in the functional $ \Gt_{\av}[\psi] $ on any $ \psi $ belonging to $ \domd(\cornern) $ (see also the proof of \cref{pro: strip}) and the trivial inclusion $ \domd(\cornern) \subset  \domn(\cornern) $, we deduce the inequality
		\beq
			\label{eq: ND up bd}
			\glecornt(L_n,\ell_n) \leq \glecorn(L_n,\ell_n).
		\eeq}
		Hence, we only have to prove the opposite inequality, i.e., 
		\beq
			E_\beta(L_n,\ell_n) \leq \widetilde{E}_\beta(L_n,\ell_n)  + o_{n}(1).
		\eeq
		
		{Preliminarily, we observe that the quantity $ \ecornt(L_n,\ell_n) $ admits a limit, which is independent of the chosen sequences, exactly as $ \ecorn(L_n,\ell_n) $. The argument to prove it is the same as in the proof of \cref{pro: ecorn}; therefore we spell in detail only the estimates showing that $ \ecornt(L,\ell) $ is monotone in $ L $ for fixed $ \ell $, up to an exponentially small error term: let $ L_a < L_b $, then we have
		\beq
			\label{eq: glecorn bd below}
			\glecornt(L_b,\ell) = \Gt_{\av}\lf[ \psit_{L_b,\ell}; \corn(L_a,\ell) \ri] + \Gt_{\av}\lf[ \psit_{L_b,\ell};  R_{\pm} \ri] 			\geq \glecornt(L_a,\ell) + \Gt_{\av}\lf[ \psit_{L_b,\ell}; R_{\pm}  \ri]{,}
		\eeq
		where $ R_\pm $ are the rectangular regions $ [L_a, L_b] \times [0, \ell] $ and $ [-L_b, -L_a] \times [ 0, \ell] $, respectively. Applying \cref{pro: strip}, we get}
		\bdm
			\Gt_{\av}\lf[ \psit_{L_b,\ell}; R_{\pm}(L_b-L_a,\ell) \ri] \geq (L_b - L_a) \eoneo(\ell) + \OO\lf((L_b-L_a) \ell^{-\infty}\ri),
		\edm
		which, plugged into \eqref{eq: glecorn bd below}, yields
		\beq
			\ecornt(L_b,\ell) \geq \ecornt(L_a,\ell) + \OO\lf((L_b-L_a) \ell^{-\infty}\ri).
		\eeq

		{Let $ \lf\{ \delta_n \ri\}_{n \in \N} $ be such that $ 0 \leq \delta_n \leq 1 $ and the pair of sequences $ \lf\{ L_n - \delta_n \ri\}_{n \in \N} $, $ \lf\{ \ell'_n \ri\}_{n \in \N} $ satisfies the same conditions \eqref{eq: subsequence} as $ \lf\{ L_n \ri\} $, $ \lf\{ \ell_n \ri\} $. Note that we have also $ \ell'_n = \ell_n + \OO(\delta_n) $, because of \eqref{eq: integer condition} (see also \cref{rem: integer condition}). We denote by $ \psit_n $ and $ \psit_{n,\delta_n} $ for short any energy minimizer in $ \cornern $ and $ \corn(L_n-\delta_n,\ell'_n) $, respectively. The splitting technique used to derive \eqref{eq: en splitting}, yields (recall \eqref{eq: psi splitting strip}, \eqref{eq: en Ett0} and \eqref{eq: splitting boundedness})
		\bml{			
			\label{eq: larger energy lb}
			\glecornt(L_n, \ell_n) = \Gt_{\av} \big[\psit_{n}; \corn(L_n-\delta_n,\ell_n^{\prime}) \big] + 2 \eoneo(\ell^{\prime}_n) \delta_n  
			+ \Ett_0\lf[{u_-}; R_-\ri] + \Ett_0\lf[{u_+}; R_+\ri]+ \OO(\ell_n^{-\infty})	\\
			\geq \glecornt(L_n-\delta_n,\ell^{\prime}_n) + 2 \eoneo(\ell^{\prime}_n) \delta_n + \Ett_0\lf[{u_-}; R_-\ri] + \Ett_0\lf[{u_+}; R_+\ri] + \OO(\ell_n^{-\infty}),
		}
		where $ R_- = [-L_n,-L_n+\delta_n]\times[0,\ell_n^{\prime}]  $ and $ R_+ = [L_n-\delta_n,L_n]\times[0,\ell_n^{\prime}]  $ and $ u_{\pm} $ are defined as in \eqref{eq: psi splitting strip}. Hence, we get that (recall that $ \Ett_0[u] \geq \exl $ if $ 1< b < \theo^{-1} $)
		\beq
			\label{eq: est reduced en}
			 \Ett_0\lf[{u_{\pm}}; R_{\pm} \ri]  \leq \ecornt(L_n,\ell_n) - \ecornt(L_n - \delta_n,\ell'_n) + \OO(\ell_n^{-\infty})
			 = : \ee_n = o_n(1),
		\eeq
		for any $ \delta_n \leq 1 $, since the two quantities $ \ecornt(L_n,\ell_n) $, $ \ecornt(L_n - \delta_n,\ell'_n) $ admit the same limit, as proven above.}
		
		Now, we claim that \eqref{eq: est reduced en} implies that, up to a phase, $ \psit_{n} $ is pointwise close to $ f_0(t) e^{-i\alpha_0 s} $ in the region $ R_- \cup R_+ $ and, in particular, along the boundary $ \bdbd $. Indeed, applying \cref{pro: point est psi} to the functionals $  \Ett_0\lf[u_{\pm}; R_{\pm}\ri] $ (with $ \delta_n $ in place of $ L $), we get that
		{\beq
			\label{eqp: kinetic u}
			\int_{- L_n}^{-L_n + \delta_n} \diff s \int_{0}^{\bar{\ell}_n} \diff t \: f_0^4 \lf| \nabla u_- \ri|^2 + \int_{L_n - \delta_n}^{L_n} \diff s \int_{0}^{\bar{\ell}_n} \diff t \: f_0^4 \lf| \nabla u_+ \ri|^2  \leq C \ee_n + \exln.
		\eeq}
		Furthermore, fixing some $ 0 < T_n \leq \bar\ell_n $, then, for any $ t \in [0,T_n] $ and any $s \in [L_n- \delta_n, L_n] $ or $ s \in [ -L_n,-L_n+\delta_n] $,
		\beq
			\label{eqp: point est psi}
			\lf| \big|\psit_{n} (\rv(s,t))\big| - f_0(t) \ri| \leq {\frac{C\ee_n^{1/4} + \OO(\ell_n^{-\infty})}{\sqrt{\min_{[0,T_n]} f_0}}},
		\eeq
		\beq
			\label{eqp: gradient s boundary}
			\lf. \lf| \partial_s \int_0^{\ell_n} \diff t \: \big| \psit_{n} \big|^2 \ri| \ri|_{s = \pm L_n} \leq {C \lf\{ \sqrt{\ee_n} + \frac{1}{\delta_n} \lf[ \frac{{\ee_n}^{1/4} + \OO(\delta_n \ell_n^{-\infty})}{\sqrt{\min_{[0,T_n]} f_0}} + e^{-c(b)T_n} \ri]  + \delta_n \ri\}.}
		\eeq
		
		{In order to simplify the discussion, let us assume that the errors $ \exln $ appearing on the r.h.s. of \eqref{eqp: kinetic u} and \eqref{eqp: point est psi} are much smaller than $ \ee_n $, since, if this is not the case, i.e., $ \ee_n $ is exponentially small in $ \ell_n $, then the argument is actually much simpler.
		Then, if we pick $ T_n $ in such a way that
		\beq
			\label{eq: positivity of psit 2}
			f_0(T_n) = \ee_n^{1/12} = o_n(1),
		\eeq}
		if the r.h.s. is larger than $ f_0(\bar\ell_n) $, or $ T_n = \bar\ell_n $ otherwise,  then 
		{\beq
			\label{eqp: kinetic u refined}
			\lf\| \nabla u_{\pm} \ri\|_{L^2(\widetilde{R}_{\pm})}^2 \leq C \ee_n^{2/3} = o_n(1),
		\eeq
		\beq
			\label{eqp: point est psi refined}
			\lf\| 1 - |u_{\pm}|  \ri\|_{L^{\infty}(\widetilde{R}_{\pm})} \leq C \ee_n^{1/8} = o_n(1),	\qquad		\big\| \psit_n \big\|_{L^{\infty}(R_{\pm} \setminus \widetilde{R}_{\pm})} \leq  C e^{-\frac{1}{2} c(b) T_n},
		\eeq
		where $ \widetilde{R}_+ : = [L_n - \delta_n, L_n] \times [0, T_n] $ and we used \eqref{eq: decay corner}. Note that, by the pointwise lower bound on $ f_0 $ stated in \eqref{eq: fk decay 1}, we find that
		\beq
			\label{eqp: Tn}
			C \sqrt{\lf| \log \ee_n \ri|} \geq T_n \geq 2 \sqrt{\lf| \log \ee_n \ri|} (1 + o_n(1)) \gg 1.
		\eeq}
		
		{Now, we claim that \eqref{eqp: kinetic u refined} and \eqref{eqp: point est psi refined} imply that $ u_{\pm} $ is close in $ L^2$ sense to a constant phase factor $ e^{i \varkappa_{\pm}} $, $ \varkappa_{\pm} \in \R $, or, equivalently, $ \psit_n \simeq f_0(t) e^{-i(\alO s - \varkappa_{\pm})} $ in $ \widetilde{R}^{\pm} $. By applying the Poincar\'{e} inequality
		\bdm
			\int_{\widetilde{R}_+} \diff s \diff t \: \lf| h - \langle h \rangle \ri|^2 \leq C \int_{\widetilde{R}_+} \diff s \diff t \: \lf\{ T_n^2 \lf| \partial_t h \ri|^2 + \delta_n^2 \lf| \partial_s h \ri|^2 \ri\},
		\edm
		where $ \langle h \rangle $ is the average of $ h $ over $ \widetilde{R}_+ $, to $ h = u_+/|u_+| $, which is well posed since $ u_+ $ does not vanish in $ \widetilde{R}_+ $ by \eqref{eqp: point est psi refined}, we obtain that there exists $ \varkappa_+ \in [0,2\pi) $ such that
		\bdm
			\lf\| \frac{u_+}{|u_+|} - e^{i \varkappa_+} \ri\|_{L^2(\widetilde{R}_+)}^2 \leq C T_n^2 \ee_n^{2/3} = o_n(1),
		\edm
		thanks to \eqref{eqp: Tn}. This in turn yields the desired estimate via \eqref{eqp: point est psi refined}:
		\beq
			\label{eqp: L2 est u_+}
			\lf\| u_+ - e^{i \varkappa_+} \ri\|_{L^2(\widetilde{R}_+)}^2 \leq C \lf[ T_n^2 \ee_n^{2/3} + \delta_n T_n \ee_n^{1/4} \ri]  = o_n(1).
		\eeq}
		
		The idea is now to exploit the information collected above to construct a trial state and prove an upper bound on $  E_{\beta}(L_n,\ell_n) $ in terms of $ \glecornt(L_n-\delta_n,\ell_n) $ via \cref{lem: twisted D energies}: we set $ \trial(\rv) := \psit_n(\rv) $ close to the corner, while sufficiently far from it,
		\beq
			\label{eqp: DN corner trial}
			\trial(\rv) : = \psit_n(\rv) + \eta(s(\rv)) \lf( f_0(t(\rv)) e^{-i\alpha_0 s(\rv)} e^{i\varkappa_{\pm}} - \psit_n(\rv) \ri),
		\eeq
		where the phases $ \varkappa_{\pm} $ are the constants appearing in \eqref{eqp: L2 est u_+}. {The function $ \eta $ is smooth and satisfies $ \eta(\pm L_n) = 1 $, $ \supp(\eta) \subset [-L_n, -L_n + \delta_n] \cup [L_n - \delta_n, L_n] \times [0, \ell_n] $ and $ |\nabla \eta| = \OO(\delta_n^{-1}) $. Obviously, the trial state $ \trial $ does not belong to $ \domd $ but $ e^{- i \varkappa_-} \trial \in \domdx $ (recall \eqref{eq: domdx}) with $ \varkappa = \varkappa_+ - \varkappa_- $. Hence, using \cref{lem: twisted D energies}, we can estimate
		\bml{
			\label{eqp: transition region}
			\glecorn(L_n,\ell_n) \leq E_{\beta, \varkappa_+ - \varkappa_-}(L_n,\ell_n) + o_n(1) \leq \G_{\av}[\trial;\cornern]+ o_n(1) 	\\
			\leq \glecornt(L_n-\delta_n,\ell_n) +\lf. \int_0^{\ell_n}\diff t\, \frac{F_0(t)}{f_0^2(t)} j_t \big[\psit_{n} \big] \ri\vert_{s=-L_n-\delta_n}^{s=L_n-\delta_n}	 + \G_{\av}\lf[\trial; R_+ \cup R_-\ri]+ o_n(1).
		}}
		
		Let us first consider the boundary terms at $ L_n $, since the ones at $ - L_n $ are perfectly equivalent: thanks to the boundary conditions \eqref{eq: neumann conditions} satisfied by $ \psit_{n} $, we get
		{\bdm
			\lf. \int_0^{\ell_n}\diff t\, \frac{F_0(t)}{f_0^2(t)} j_t \big[\psit_{n} \big] \ri\vert_{s=L_n-\delta_n} = - \lf. \frac{1}{2} \int_0^{\ell_n}\diff t\, \partial_s \big| \psit_{n} \big|^2 \ri\vert_{s=L_n} + \int_{L_n - \delta_n}^{L_n} \int_0^{\ell_n}\diff t\, \frac{F_0(t)}{f_0^2(t)} \partial_s j_t \big[\psit_{n} \big].
		\edm
		Integrating by parts as in \eqref{eq: int by parts} and using the Agmon bound provided by \cref{lem: agmon 2} as well as the inequalities \eqref{eq: kol positive} and \eqref{eq: fol bound} , one can show that the second term on the r.h.s. of the expression above is bounded by
		\bml{
			 \lf| \int_{L_n - \delta_n}^{L_n} \int_0^{\ell_n}\diff t\, \frac{F_0(t)}{f_0^2(t)} \partial_s j_t \big[\psit_{n} \big] \ri| \\
			 \leq 2 \int_{L_n - \delta_n}^{L_n} \int_0^{T_n}\diff t\, \big| \partial_t \psit_n \big| \big| \partial_s \psit_n \big| 
			 + 2 e^{-c(b) T_n} \int_{R_+ \setminus \widetilde{R}_+} \diff s \diff t \: e^{c(b) t} \big| \nabla \psit_n \big|^2 	\\
			 \leq C \lf[ \lf\| f_0 \nabla u \ri\|_{L^2(\widetilde{R}_+)}^2 + e^{-c(b) T_n} \ri]  \leq C \lf( \ee_n^{5/6} + e^{-c(b) T_n} \ri),
		}
		 while the first one can be estimated via \eqref{eqp: gradient s boundary}, so obtaining 
		\beq
			\label{eqp: boundary terms transition}
			\lf| \lf. \int_0^{\ell_n}\diff t\, \frac{F_0(t)}{f_0^2(t)} j_t \big[\psit_{n} \big] \ri\vert_{s=L_n-\delta_n} \ri| \leq C \lf[ \sqrt{\ee_n} + \frac{{\ee_n}^{5/24}+ e^{-c(b)T_n}}{\delta_n}  + \delta_n \ri].
		\eeq}
	
		We now focus on the energy contributions of the regions $ R_{\pm} $ (third term on the r.h.s. of \eqref{eqp: transition region}): For simplicity, we are going to consider only the energy in the region $ R_+ $, since the corresponding one in $ R_- $ can be bounded in the very same way. We have
		\bml{
			\label{eqp: energy R+}
			\G_{\av}\lf[\trial; R_+ \ri] \leq 10 \lf\| \lf( \nabla - i t \ev_s \ri) \psit_n \ri\|_{L^2(R_+)}^2 + 4 \lf\| \nabla \eta \lf(f_0  e^{-i\alpha_0 s + i\varkappa} - \psit_n \ri)  \ri\|_{L^2(R_+)}^2 \\
			+ 8 \lf\| \lf( \nabla - i t \ev_s \ri) f_0 e^{-i\al_0 s + i \varkappa} \ri\|_{L^2(R_+)}^2 + \frac{1}{b} \lf( \big\| \psit \big\|_{L^4(R_+)}^4 + \lf\| f_0 \ri\|_{L^4(R_+)}^4 \ri)	\\
			\leq {\frac{C}{\delta_n^{2}} \lf\| f_0  e^{-i\alpha_0 s + i\varkappa} - \psit_n  \ri\|_{L^2(R_+)}^2 + C \lf( \ee_n^{5/6} + e^{-c(b) T_n} \ri) + \OO(\delta_n)}
		}
		thanks to \eqref{eq: fk decay 1} and \eqref{eq: agmon 2} and where {the first term on the r.h.s. has been bounded by Cauchy inequality, exploiting \eqref{eqp: kinetic u}, \eqref{eq: agmon 2}, \eqref{eq: decay corner} and the splitting technique:
		\bml{
			\lf\| \lf( \nabla - i t \ev_s \ri) \psit_n \ri\|_{L^2(R_+)}^2 \leq \int_{R_+} \diff s \diff t \: f_0^2 \lf\{ \lf| \nabla u_+ \ri|^2  -  2(t+\alpha_0)  j_s[u_+] + \tx\frac{1}{b} \lf(1 - f_0^2 \ri)  |u_+|^2 \ri\} \\
			\leq C \lf[ \lf\| f_0 \nabla u \ri\|_{L^2(\widetilde{R}_+)}^2 + e^{-c(b) T_n} \ri] + \OO(\delta_n)
			\leq C \lf( \ee_n^{5/6} + e^{-c(b) T_n} \ri) + \OO(\delta_n).
		}
		We now exploit \eqref{eqp: point est psi refined} and \eqref{eqp: L2 est u_+} to deduce that
		\bml{
			\label{eqp: distance on R+}
			\lf\| f_0  e^{-i\alpha_0 s + i\varkappa} - \psit_n  \ri\|_{L^2(R_+)}^2 \leq \int_{\widetilde{R}_+} \diff s \diff t \: f_0^2 \lf| u - e^{i \varkappa} \ri|^2 + C \delta_n e^{- \frac{1}{2} c(b) T_n}	\\
			\leq C \lf[ T_n^2 \ee_n^{2/3} + \delta_n T_n \ee_n^{1/4} + e^{- \frac{1}{2} c(b) T_n }\ri].	
		}}

		{Putting together \eqref{eqp: transition region} with \eqref{eqp: boundary terms transition}, \eqref{eqp: energy R+} and \eqref{eqp: distance on R+}, we finally get
		\bml{
			\glecorn(L_n,\ell_n) \leq \glecornt(L_n-\delta_n,\ell_n) \\
			+C \lf\{ \sqrt{\ee_n} + \frac{{\ee_n}^{5/24}+ e^{-c(b)T_n}}{\delta_n}   + \frac{T_n^2 \ee_n^{2/3} + \delta_n T_n \ee_n^{1/4} + e^{- \frac{1}{2} c(b) T_n }}{\delta_n^2} + \delta_n \ri\}  \\
			\leq \glecornt(L_n-\delta_n,\ell_n) + C \lf[  \frac{T_n^2 \ee_n^{2/3} + \delta_n T_n \ee_n^{1/4} + e^{- \frac{1}{2} c(b) T_n }}{\delta_n^2} + \delta_n \ri] + o_n(1)	\\
			\leq \glecornt(L_n-\delta_n,\ell_n) + C \lf[ \max\lf\{ T_n^{2/3} \ee_n^{2/9}, T_n^2 \ee_n^{1/8} \ri\} + e^{- \frac{1}{10} c(b) T_n } \ri] + o_n(1) \\
			= \glecornt(L_n,\ell_n) + o_n(1),
		}
		where we have optimized over $\delta_n $ by taking $ \delta_n = \max\{ T_n^{2/3} \ee_n^{2/9}, T_n^2 \ee_n^{1/8} \} + e^{- \frac{1}{5} c(b) T_n} $ and used that $ T_n = \OO(\sqrt{|\log \ee_n|}) $.}
		\end{proof}

	{\begin{proof}[Proof of \cref{lem: twisted D energies}]
		We first observe that the existence of the limit as $ n \to + \infty $ of $ E_{\beta,\varkappa}(L_n,\ell_n) - 2 \eoneo(\ell_n) L_n $ can be shown as in the proof of \cref{pro: ecorn}. Hence for any $ 1 \ll \delta_n \ll \min\{\ell_n, L_n\} $, we have
		\bdm
			E_{\beta,\varkappa}(L_n,\ell_n) - E_{\beta,\varkappa}(L_n-\delta_n,\ell_n)  + 2 \eoneo(\ell_n) \delta_n = o_n(1).
		\edm
		By a trivial testing of the functional, exploiting the above estimate as well as \cref{pro: strip} and \cref{pro: twisted}, one gets
		\bml{
			\glecorn(L_n,\ell_n) \leq E_{\beta,\varkappa}(L_n-\delta_n,\ell_n) + \edk(R(\delta_n,\ell_n)) \\
			\leq E_{\beta,\varkappa}(L_n-\delta_n,\ell_n) + 2 \eoneo(\ell_n) \delta_n + \OO(\delta_n^{-1}) + o_n(1) = E_{\beta,\varkappa}(L_n,\ell_n) + o_n(1).
		}
		The proof of the opposite inequality is identical.
	\end{proof}}

\section{Proof of the Energy Lower Bound}
\label{sec: lower bound}

In this Section we prove the lower bound to the GL energy which in combination with the upper bound proven in \cref{pro: upper bound}, stated in next Section, will provide the proof of \cref{teo: gle asympt}. 

	\begin{pro}[GL energy lower bound]
		\label{pro: lower bound}
		\mbox{}	\\
		Let $\Om\subset\R^2$ be any bounded simply connected domain satisfying \cref{asum: boundary 1} and \cref{asum: boundary 2}. Then, for any fixed $ 1< b < \theo^{-1} $,		 as $\eps \to 0$, it holds
		\beq
			\label{eq: lower bound}
			\gle  \geq \disp\frac{|\partial\Om| \eoneo}{\eps} - \ecorr \int_{0}^{|\partial \Omega|} \diff \ss \: \curv + \sum_{j = 1}^N E_{\mathrm{corner}, \beta_j} + o(1). 
		\eeq
	\end{pro}

We recall the definition of the superconducting boundary layer
\bdm
	\anne : = \lf\{ \rv \in \Omega \: \big| \: \dist\lf(\rv, \partial \Omega\ri) \leq \eps \elle \ri\},
\edm
with (see \eqref{eq: elle}) $ \elle = c_1  |\log\eps| $, for a large constant $ c_1 $. {The smooth part of the boundary layer is defined as
\beq
	\label{eq: acute}
	\acute : = \lf\{ \rv \in \anne \: \big| \lf| s(\rv) - \ss_j \ri| \geq \eps \leps \ri\},
\eeq
where $ \ss_j $ is the coordinate along $ \partial \Omega $ of the $j$-th corner and
\beq
	\label{eq: leps}
	\leps = c_2(\eps)  |\log\eps|,
\eeq
for some 
\beq
	\tx\frac{c_1}{\tan\lf(\beta/2\ri)} \leq c_2(\eps) \leq C ,
\eeq	
so that \eqref{eq: subsequence} holds. The corner regions are denoted by $ \corneru $, $ j \in \lf\{ 1, \ldots, N \ri\} $, and coincide with the complement of $ \acute $:
\beq
	\label{eq: corneru}
	\corneru : =  \lf( \anne \setminus \acute \ri) \cap \lf\{ \rv \in \anne \: | \: \dist(\rv, \rv_j) \leq C \eps \leps \ri\}.
\eeq}
In $ \acute $, one can use the tubular coordinates $ (\ss, \tt) $ defined in \eqref{eq: tc} as well as their rescaled counterparts given in \eqref{eq: rescaled tc}. We denote by $ \ann $ the rescaling of the boundary layer $ \anne $. Similarly, the set obtained via rescaling of the domain $ \acute $ is denoted by $ \acut $, i.e., with a little abuse of notation,
\beq
	\label{eq: acut}
	\acut : = \lf( \lf[0, s_{1} - \leps \ri] \cup \lf[ s_1 + \leps, s_2 - \leps \ri] \cup \cdots \cup \lf[ s_N + \leps, \tx\frac{|\partial \Omega|}{\eps} \ri] \ri) \times [0, c_1 |\log\eps|],
\eeq
while $ \Gamma_j $ stands for the rescaling of the domain $ \corneru $, i.e., $ \Gamma_j : =\lf\{ \rv' \in \R^2 \: \big| \: \rv_j + \eps \rv' \in \corneru \ri\} $.

\begin{figure}[!ht]
		\begin{center}
		\begin{tikzpicture}
			\draw (0,0) to[bend right=11] (2,4);
			\draw (2,4) to[bend left=15] (4.2,-0.5);
			\draw (0,0) -- (1,-0.7);
			\draw (4.2,-0.5) -- (3,-0.8);
			\draw (1,-0.7) to[bend right=7] (2.3,1.5);
			\draw (3,-0.8) to[bend right=8] (2.3,1.5);
		\end{tikzpicture}
		\caption{A typical corner region $ \corneru $ (or, after rescaling, $ \Gamma_j $) before the rectification.}
		\label{fig: corner curve}
		\end{center}{}
	\end{figure}
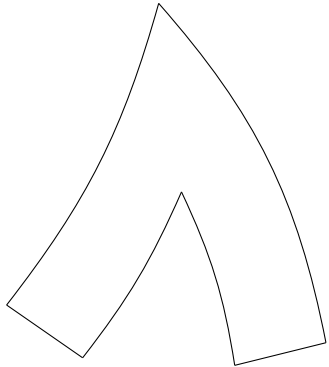

Before proceeding further, we summarize the main steps of the proof of the lower bound. We are going to treat the smooth part of the layer and the corner regions differently. In order to extract the $ \OO(1) $ contributions to the energy, it is indeed necessary to retain in the smooth part of the layer the terms depending on the boundary curvature. The same precision is not needed close to the corners. There, however, the procedure is more involved, since we have to reconstruct the model problem discussed in \cref{sec: corner energy}. 
\begin{itemize}
	\item The first step is the \textit{replacement of the magnetic vector potential} (\cref{sec: replacement magnetic field}). The idea is to replace $\aav^{\mathrm{GL}}$ with $-\tt \ev_{\ss} {+ \frac{1}{2} \curv \tt^2 + \eps \deps} $ far enough from the corners by means of a suitable gauge {change}. Close to the corners, on the other hand, we replace $ \aavm $ with $ \fv $ (\cref{lem: replacement corners}) by means of a priori bounds of the difference between $ \aavm $ and $ \fv $ (see \cref{sec: elliptic}); 
	\item The second step is the \textit{rectification of the corner regions} (\cref{sec: rectification}): via a suitable diffeomorphism, we map the corner region as in \cref{fig: corner curve} onto a domain with the same shape as $ \corner $ in \cref{fig: corner}; this allows us to reduce the lower bound to the corner effective problem introduced in \eqref{eq: ecorn};
	\item The third step is simply the \textit{completion of the lower bound} (\cref{sec: completion}), where we just glue together the lower bounds near the corners with the one in the smooth part of the domain discussed in \cref{sec: smooth energy}.
\end{itemize}

\subsection{Replacement of the magnetic field}
\label{sec: replacement magnetic field}

In $ \acute $ we aim at bounding from below {the GL energy} by the reduced energy functional $ \gep[\psi; \acut] $, where
\bml{
	\label{eq: gep}
	\gep[\psi; \acut] : = \int_{\acut} \diff s \diff t \: \lf(1 - \eps k(s) t \ri) \lf\{ \lf| \partial_t \psi \ri|^2 + \tx\frac{1}{(1 - \eps k(s) t)^2} \lf| \lf( \partial_s + i \aeps(s,t) \ri) \psi \ri|^2 \ri.	\\
	\lf. - \tx\frac{1}{2 b} \lf( 2 |\psi|^2 - |\psi|^4 \ri) \ri\}.
}
\beq
	\aave(s,t) : = \aeps(s,t) \ev_{s},		\qquad			\aeps(s,t) = - t + \tx\frac{1}{2} k(s) t^2 + \eps \deps,
\eeq
and $ \psi(s,t) = \glm(\rv(s,t)) e^{-i\phi_\eps(\rv(s,t))} $, with $ \phi_{\eps} $ a suitable gauge phase (see \eqref{eq: gauge phase} below). 

The replacement procedure by means of a local gauge choice is well described in \cite[Appendix F]{FH1}  for smooth domains and, in more details, in \cite[Sect. 5.1]{CR2}. A similar discussion is extended in presence of corners at the boundary in \cite[Sect. 2.4]{CG}, where however the energy of the corner regions is dropped.

	\begin{lem}[Replacement of the magnetic potential in $ \acute $]
		\label{lem: replacement smooth}
		\mbox{}	\\
		Under the assumptions of \cref{pro: lower bound}, there exists $ \phi_{\eps} \in C^{\infty}(\Omega) $ such that, setting $ \psi(s,t) : = \glm(\rv(s,t)) e^{i \phi_{\eps}(\rv(s,t))} $, we get, as $ \eps \to 0 $,
		\beq
			\label{eq: replacement smooth}
			\glfe\lf[\glm, \aavm; \acute\ri] \geq \gep[\psi; \acut] + \OO(\eps^2 |\log\eps|^2).
		\eeq
	\end{lem}
	
	\begin{proof}
		As described above there are three operations, which are performed simultaneously, to get \eqref{eq: replacement smooth}:
		\begin{itemize}
			\item change to boundary tubular coordinates $ (\ss,\tt) $;
			\item extraction of a suitable gauge phase to replace $ \aavm $ with $ - \tt \ev_{\ss} {+ \frac{1}{2} \curv \tt^2 + \eps \deps} $;
			\item rescaling of all the lengths (e.g., via \eqref{eq: rescaled tc}).
		\end{itemize}
		As anticipated, the above procedures have been already discussed in the literature{, therefore we omit the details for the sake of brevity. We only provide the expression of the gauge phase for later convenience}
		\begin{equation}
			\label{eq: gauge phase}
			\phi_\varepsilon(s,t) := -\frac{1}{\varepsilon}\int_0^{t} \, \diff\eta\, \aavm(\rv(\eps s, \eps \eta))\cdot \nuv(\eps s) - \frac{1}{\eps}\int_0^{s} \, \diff \xi\, \aavm(\rv(\eps \xi,0))\cdot \gamv^\prime(\eps \xi) + \OO(\eps) s.
		\end{equation}
	\end{proof}

In the corner regions, on the opposite, it suffices to use a priori bounds on the solutions of the GL equations to substitute $ \aavm $ with $ \fv $ (recall \eqref{eq: fv}). Before doing that, we need however a preparatory lemma:

	\begin{lem}
		\label{lem: kinetic energy corner}
		\mbox{}	\\
		For any $ j  {= 1, \ldots, N} $, as $ \eps \to 0 $,
		\beq
			\label{eq: kinetic energy corner}
			\lf\| \lf( \nabla + i \tx\frac{\aavm}{\eps^2} \ri) \glm \ri\|_{L^2(\corneru)} = \OO(|\log\eps|).
		\eeq
	\end{lem}
	
	\begin{proof}
		The idea is to exploit the variational equation for $ \glm $ in \eqref{eq: GL eqs}, to compute
		\bml{
			\lf\| \lf( \nabla + i \tx\frac{\aavm}{\eps^2} \ri) \glm \ri\|^2_{L^2(\corneru)} = \frac{1}{\eps^2} \int_{\corneru} \diff \rv \: \lf( 1 - \lf| \glm \ri|^2 \ri) \lf| \glm \ri|^2 \\
			+ \int_{\partial \corneru \setminus \partial \Omega} \diff x \: {\glm}^* \nuv \cdot \lf( \nabla +  i \tx\frac{\aavm}{\eps^2} \ri) \glm = \OO(|\log\eps|^2),
		}
		by the bounds \eqref{eq: est infty}, \eqref{eq: est grad infty} and the boundary conditions on $ \glm $ (recall \eqref{eq: GL eqs}).
	\end{proof}
	
	We can now perform the vector potential replacement:

	\begin{lem}[Replacement of the magnetic field in $ \corneru $]
		\label{lem: replacement corners}
		\mbox{}	\\
		{For any $ j = 1, \ldots, N $,} there exists $ \psi_j \in H^1(\anne) $, so that, as $ \eps \to 0 $,
		\beq
			\label{eq: replacement corners}
			\glfe\lf[\glm, \aavm; \corneru \ri] \geq \glf_1\lf[\psi_j, \fv; \Gamma_j \ri] + \OO(\eps^{3/5}).
		\eeq
	\end{lem}
	
	\begin{remark}[Kinetic energy in the corner regions]
		\label{rem: kinetic energy corner}
		\mbox{}	\\
		Combining \cref{lem: kinetic energy corner} with \cref{lem: replacement corners}, one can easily deduce that \eqref{eq: kinetic energy corner} holds true with $ (\psi_j, \fv) $ in place of $( \glm, \aavm) $. More precisely, let $ {\psi_j} $ as in \eqref{eq: replacement corners}, then
		\beq
			\label{eq: kinetic energy corner fv}
			\lf\| \lf( \nabla + i \fv \ri) \psi_j \ri\|_{L^2(\Gamma_j)} = \OO(|\log\eps|).
		\eeq
	\end{remark}

	\begin{proof}
		A straightforward computation yields
		\bml{
			\label{eqp: replacement}
 			\lf\| \lf( \nabla + i \tx\frac{\aavm}{\eps^2} \ri) \glm \ri\|^2_{L^2(\corneru)} - \lf\| \lf( \nabla + i \tx\frac{\fv}{\eps^2} \ri) \glm \ri\|^2_{L^2(\corneru)}  \\
			= - 2 \Im \int_{\corneru} \diff \rv \:   \lf[ \lf( \nabla + i \tx\frac{\aavm}{\eps^2} \ri) \glm \ri]^*  \lf( \tx\frac{\aavm}{\eps^2} - \tx\frac{\fv}{\eps^2} \ri) \glm - \int_{\corneru} \diff \rv \:  \lf| \tx\frac{\aavm}{\eps^2} - \tx\frac{\fv}{\eps^2} \ri|^2 \lf| \glm \ri|^2	\\
		\geq - \delta \lf\| \lf( \nabla + i \tx\frac{\aavm}{\eps^2} \ri) \glm \ri\|^2_{L^2(\corneru)}  - \tx\frac{1}{\eps^4} \lf( \tx\frac{1}{\delta} + 1\ri) \lf\| \aavm - \fv \ri\|^2_{L^2(\corneru)} 	\\
		\geq - C \delta |\log\eps|^2 - \tx\frac{1}{\eps^4} \lf(\tx\frac{1}{\delta} + 1\ri) \lf\| \aavm - \fv \ri\|^2_{L^{2p}(\Omega)} \lf| \corneru \ri|^{1 - \frac{1}{p}},
 		}
 		for any $ p \in [2, \infty) $, where we have used \cref{lem: kinetic energy corner}. Plugging now \eqref{eq: vector potential p}, which reads for $ p' \in [2, +\infty) $
 		\beq
 			\label{eq: vector potential p'}
 			\lf\| \aavm - \fv \ri\|_{L^{p'}(\Omega)} = \OO(\eps^{7/4}),
		\eeq
		{thanks to} the Agmon decay \eqref{eq: agmon}, we get
		\bml{
			\lf\| \lf( \nabla + i \tx\frac{\aavm}{\eps^2} \ri) \glm \ri\|^2_{L^2(\corneru)} - \lf\| \lf( \nabla + i \tx\frac{\fv}{\eps^2} \ri) \glm \ri\|^2_{L^2(\corneru)}\geq - C \lf[ \delta |\log\eps|^2 \ri. \\
			\lf. - \tx\frac{1}{\sqrt{\eps}} \lf(\tx\frac{1}{\delta} + 1\ri)  \lf(\eps^2 |\log\eps|^2 \ri)^{1 - \frac{1}{p}} \ri] \geq - C \eps^{3/4 - 1/p} |\log\eps|^2
		}
		after an optimization over $ \delta $ (i.e., taking $ \delta = \eps^{3/4-1/p} |\log\eps|^{{-1/p}} $). The proof is then completed by exploiting the invariance of the energy under combined translations and global gauge change.		
	\end{proof}

\subsection{Rectification of the corner regions}
\label{sec: rectification}

We aim at estimating from below the energy close to the corners ({first} term on the r.h.s. of \eqref{eq: replacement corners}) by the minimal energy of the model problem introduced in \eqref{eq: ecorn} and discussed in \cref{sec: corner effective}. To this purpose there are two difficulties to overcome. First, one needs to force the boundary conditions in the minimization of the corner energies appearing in \eqref{eq: replacement corners}, since an unconstrained minimization would lead to unwanted contribution from the normal boundaries at $ s_j \pm \leps $. Such a problem will however be solved  by exploiting \cref{pro: DN corner}.

The second issue is less trivial: the model problem is indeed defined on a domain whose boundary is straight, while typically the boundaries of the domains $ \Gamma_j $ have a non-trivial curvature, as in \cref{fig: corner curve}. Of course, being the corner regions rather small and the curvature bounded,  the corrections induced by this adjustment are of lower order (\cref{lem: rectification}). 

Let us introduce some notation: we are going to denote by $ \cornerr $ the corner region of opening angle $ \beta_j $ with straight sides, longitudinal length $ \leps $ and normal width $ \elle (1 + o(1)) $. We do not require the inner boundaries of $ \cornerr $ to be straight since the exponential decay of any GL minimizer makes such a boundary irrelevant. We also choose the coordinates in such a way that the corner coincides with the origin. Hence, except for the inner boundaries, $ \cornerr $ coincide with the region described in \cref{fig: corner}, up to a rotation:
\beq
	\label{eq: rotation}
	\cornerr \simeq \mathcal{R} \cornere,
\eeq
with $ \mathcal{R} $ a rotation around the axis perpendicular to the plane passing through the corner.

	\begin{lem}[Rectification of the corner]
		\label{lem: rectification}
		\mbox{}	\\
		Let $ \psi_j $ be the $ H^1 $ function in \eqref{eq: replacement corners}. Then, there exists a diffeomorphism $ \rrv: \Gamma_j \to \cornerr $, so that, setting $ \widetilde\psi_j(\rrv) : = \psi_j(\rv(\rrv)) $,
		\beq
			\label{eq: rectification}
			 \glf_1\lf[\psi_j, \fv; \Gamma_j \ri] = \glf_1\big[\widetilde\psi_j, \fv; \cornere \big] + \OO(\eps \logi).
		\eeq
	\end{lem}
	
	\begin{proof}
		We want to map the region $ \Gamma_j $ onto $ \cornerr $ via a suitable diffeomorphism and exploit the fact that, thanks to the boundedness of curvature and the size of the region, such a map is suitably close to the identity. A similar a trick has already been used, e.g., in \cite{BNF}. Indeed, there exists a smooth map $ \rrv(\rv) : \Gamma_j \to \cornerr $ which is one-to-one, such that $ \rrv(0) = 0 $ and 
		\beq
			\label{eq: diffeomorphism}
			R_j(\rv) = r_j \lf( 1 + \OO(\eps |\log\eps|) \ri),		\qquad		\partial_j R_k(\rv) = \delta_{jk} +  \OO(\eps \logi).
		\eeq	
		Using such a map, we get
		\beq
			\label{eqp: rectification 1}
			\glf_1\lf[\psi_j, \fv; \Gamma_j \ri] = \int_{\cornerr} \diff \rrv \lf\{ \lf| \lf( J \nabla_{\rrv} + i \widetilde\fv\ri) \widetilde\psi_j \ri|^2 - \tx\frac{1}{2b} \lf( 2 \big| \widetilde\psi_j \big|^2 - \big| \widetilde\psi_j \big|^4 \ri) \ri\},
		\eeq
		where $ J $ is the jacobian matrix associated to the change of coordinates $ \rv \to \rrv $ and
		\beq
			\label{eq: rectification 1}
			\widetilde\psi_j(\rrv) : = \psi_j(\rv(\rrv)),		\qquad			\widetilde\fv(\rrv) : = \fv(\rv(\rrv)).
		\eeq
		By \eqref{eq: diffeomorphism}, 
		\beq
			\widetilde\fv(\rrv) = \tx\frac{1}{2} \rrv^{\perp} \lf(1 + \OO(\eps|\log\eps|) \ri).
		\eeq
		Therefore, we can estimate from below the r.h.s. of \eqref{eqp: rectification 1} exactly as in \eqref{eqp: replacement}, using \eqref{eq: kinetic energy corner fv} in \cref{rem: kinetic energy corner}, to get
		\bml{
			\glf_1\lf[\psi_j, \fv; \Gamma_j \ri] \geq \int_{\cornerr} \diff \rv \lf\{ \lf| \lf( \nabla + i \fv \ri) \widetilde\psi_j \ri|^2 - \tx\frac{1}{2b} \lf( 2 \big| \widetilde\psi_j \big|^2 - \big| \widetilde\psi_j \big|^4 \ri) \ri\} \\
			 - C \lf[ \delta |\log\eps|^2 + \lf(1 + \tx\frac{1}{\delta}\ri) \eps^2 |\log\eps|^4 \ri] 			\geq \glf_1\big[\widetilde\psi_j, \fv; \cornerr \big] + \OO(\eps\logi).
		}
		The last step is then the replacement of the region $ \cornerr $ with $ \cornere $, which can be done exploiting the rotational invariance of the GL functional, and the exponential decay of $ \glm $ given by \eqref{eq: agmon} (and thus of $ \psi $).
	\end{proof}

\subsection{Completion of the lower bound}
\label{sec: completion}
 We are now in position to complete the proof of the lower bound.

	\begin{proof}[Proof of \cref{pro: lower bound}]
		Combining the results proven in \cref{lem: replacement smooth}, \cref{lem: replacement corners} and \cref{lem: rectification}, we get
		\beq
			\label{eqp: lb start}
			\gle \geq  \gep[\psi; \acut] + \sum_{j = 1}^N \glf_1\big[\widetilde\psi_j, \fv; \cornere \big]  {+  \OO(\eps^{3/5})}.
		\eeq
		At this stage the energy contributions of the smooth part of the domain and its complement have been completely decoupled, so we can bound them from below separately. In fact, the lower bound to $ \gep[\psi; \acut] $ can be simply taken from \cite[Prop. 4]{CR2}: 
		\bml{
			\label{eqp: first lb}
			\gle \geq  \disp\frac{|\partial\Om| \eoneo}{\eps} - 2 \leps \eoneo - \ecorr \int_{0}^{|\partial \Omega|} \diff \ss \: \curv  \\
			+ \sum_{j = 1}^N \lf[ \glf_1\big[\widetilde\psi_j, \fv; \cornere \big] -  \int_0^{c_0|\log\eps|}\diff t\, \lf.\frac{F_0(t)}{f_0^2(t)} j_t \lf[\psi(s,t)\ri] \ri|_{s=s_j - L_\eps}^{s= s_j +L_\eps} \ri]
			{+  \OO(\eps^{3/5})}. 
		}	
		Now, recalling the definitions of $ \psi $ in \cref{lem: replacement smooth} and $ \widetilde{\psi}_j $ in \cref{lem: rectification}, we claim that
		\beq
			\label{eq: reconstruction gt}
			\glf_1\big[\widetilde\psi_j, \fv; \cornere \big]  - \int_0^{c_0|\log\eps|}\diff t\, \lf.\frac{F_0(t)}{f_0^2(t)} j_t \lf[\psi(s,t)\ri] \ri|_{s=s_j - L_\eps}^{s= s_j +L_\eps}
			= \Gt_{\fv}\big[\widetilde{\psi}_j \big] + \OO({\eps^{3/5}}),
		\eeq
		where the functional $ \Gt_{\fv} $ is defined in \eqref{eq: gltfv en corner}.
		
		We note that, because of the rigid translation and rotation, the boundaries $ \partial\Gamma_{\beta_j, \mathrm{bd}} $ coincides with the portion of the lines $ s = s_j \pm \leps $ in $ \mathcal{A} $. Therefore, in order to replace $ \psi $ with $ \widetilde\psi_j $ in the boundary terms in \eqref{eq: reconstruction gt}, we need to estimate the contribution of the gauge phase as well as the effect of the rectification. 
		Next, we observe that
		\bml{
			\lf. j_t[\psi] \ri|_{s = s_j \pm \leps} 
			= \lf.  j_t\lf[ \psi_j((\rv(\eps s, \eps t)) - \rv_j)/\eps) e^{-i\phi_\eps(s,t)} e^{ - i \frac{ \fv(\rv_j) \cdot \rv'}{\eps} } \ri] \ri|_{s = s_j \pm \leps}	\\
			= \lf.  \lf[ j_t\lf[ \psi_j((\rv(\eps s, \eps t)) - \rv_j)/\eps) \ri] - \partial_t \lf( \phi_\eps(s,t) + \tx\frac{1}{\eps} \fv(\rv_j) \cdot \rv' \ri) \ri] \ri|_{s = s_j \pm \leps}.
		}
		Recalling \eqref{eq: rectification 1}, we immediately see that the first term on the r.h.s. of the expression above in fact reconstructs the boundary terms in $ \Gt_{\fv} $. Moreover, by direct computation
		\bdm
			- \partial_t \lf( \phi_\eps + \tx\frac{1}{\eps} \fv(\rv_j) \cdot \rv' \ri) = \tx\frac{1}{\eps} \lf( \aavm(\rv) - \fv(\rv_j) \ri) \cdot \ev_t  = \tx\frac{1}{\eps} \lf( \aavm(\rv) - \fv(\rv) \ri) \cdot \ev_t + \frac{1}{2} {\rv^{\prime}}^{\perp} \cdot \ev_t,
		\edm
		where we have used the change of coordinates $ \rv = \rv_j + \eps \rv^{\prime} $. Furthermore, the properties of the diffeomorphism discussed in the proof of \cref{lem: rectification} imply that
		\bdm
			\tx\frac{1}{2} {\rv^{\prime}}^{\perp} \cdot \ev_t = \tx\frac{1}{2} s + \OO(\eps \logi).
		\edm
		Hence, each boundary term can be rewritten
		\bml{
			\int_0^{c_0|\log\eps|}\diff t\, \lf. \frac{F_0(t)}{f_0^2(t)} j_t \lf[\psi(s,t) \ri] \ri|_{s=s_j - L_\eps}^{s= s_j +L_\eps}  =  \int_0^{c_0|\log\eps|}\diff t\, \lf. \frac{F_0(t)}{f_0^2(t)} \lf[  j_t \lf[\widetilde\psi(\rv(s,t))\ri] + \tx\frac{1}{2} s \ri] \ri|_{s=- L_\eps}^{s= +L_\eps}  \\
			 - \frac{1}{\eps} \int_0^{c_0|\log\eps|}\diff t\, \lf. \frac{F_0(t)}{f_0^2(t)} \lf( \aavm(\rv(s,t)) - \fv(\rv(s,t)) \ri) \cdot \ev_t \ri|_{s=s_j - L_\eps}^{s= s_j +L_\eps} + \OO(\eps\logi)
		}
		and it only remains to bound the last term. This can be done exploiting once more \eqref{eq: vector potential p}: setting $ g(\rv) : = F_0(t(\rv))/f_0^2(t(\rv)) $ for short and using the vanishing of  $ {F_0} $ at $  t = 0, c_1 |\log\eps| $, we {get}
		\bml{
			\label{eqp: reconstruction}
			\frac{1}{\eps^3} \lf| \int_{\corneru} \diff \rv \: \nabla \cdot \lf[ g \lf( \aavm  - \fv  \ri) \ri] \ri| 			= \frac{1}{\eps^3} \lf| \int_{\corneru} \diff \rv \: \lf( \nabla g  \ri) \cdot \lf( \aavm - \fv  \ri) \ri| 	\\
			\leq \frac{C |\log\eps|^5}{\eps^3} \lf\| \aavm - \fv \ri\|_{L^1(\corneru)},
		}
		where we have used that both $ \aavm $ and $ \fv $ are divergence free and the estimate
		\beq
			\lf| \nabla g(\rv) \ri| = \lf| \partial_t  \frac{F_0(t)}{f_0^2(t)} \ri| \leq 2 \lf| t + \al_0 \ri| + \lf| \frac{F_0(t) f_0^{\prime}(t)}{f_0^3(t)} \ri| \leq C |\log\eps|^5,
		\eeq
		by \eqref{eq: fol}, \eqref{eq: fkprime decay} and the simple bound \eqref{eq: fol bound}. On the other hand, by \eqref{eq: vector potential p'}, 
		\beq
			\lf\| \aavm - \fv\ri\|_{L^1(\corneru)} \leq \lf\| \aavm - \fv\ri\|_{L^p(\corneru)} \lf| \corneru \ri|^{1 - \frac{1}{p}} 
			\leq C \eps^{7/4} \lf( \eps^2 |\log\eps| \ri)^{1 - \frac{1}{p}},
		\eeq
		for $ p\in [2, +\infty) $, which implies \eqref{eq: reconstruction gt} via \eqref{eqp: reconstruction}.
		
		Putting together \eqref{eqp: first lb} with \eqref{eq: reconstruction gt} and observing that, by the Agmon decay, we can impose the boundary condition $ \widetilde{\psi} = \psi_0  $ along the interior boundary $ \bdi $ up to $ \OO(\eps^{\infty}) $ errors, we thus get
		\bml{
			\gle \geq  \disp\frac{|\partial\Om| \eoneo}{\eps} - 2 \leps \eoneo - \ecorr \int_{0}^{|\partial \Omega|} \diff \ss \: \curv \\
			+ \sum_{j \in \Sigma} \inf_{\psi \in \widetilde\dom_{\star}(\Gamma_{\beta_j}(\leps,\elle))} \Gt_{{\fv}}\lf[ \widetilde\psi, \Gamma_{\beta_j}(\leps,\elle) \ri]
			+ \OO(\eps^{3/5}) \\
			= \disp\frac{|\partial\Om| \eoneo}{\eps} - \ecorr \int_{0}^{|\partial \Omega|} \diff \ss \: \curv
			+ \sum_{j \in \Sigma} \lf( \Et_{\beta_j}(\leps,\elle) - 2 \leps \eoneo \ri)
			+ \OO(\eps^{3/5}).
		}
		The final step of the proof is the application of \cref{pro: DN corner}, which yields (see \eqref{eq: ecorn alt})
		\beq
			\Et_{\beta_j}(\leps,\elle) - 2 \leps \eoneo = E_{\mathrm{corner},\beta_j} + o(1),
		\eeq
		and thus the result.
	\end{proof}

\section{Other Proofs}	
\label{sec: other proofs}

We complete in this Section the proofs of the results proven in the paper, i.e., specifically, we prove the energy upper bound matching the lower bound proven in \cref{pro: lower bound}. Finally, we show how the energy asymptotics can be used to deduce a {pointwise estimate} of the order parameter.

\subsection{Upper bound and energy asymptotics}
\label{sec: upper bound}

We state the main result of this section in the following \cref{pro: upper bound}. Note that \cref{pro: lower bound} and \cref{pro: upper bound} together completes the proof of \cref{teo: gle asympt}, with the simple exception of the replacement of $ \eoneo $ with $ \eones $, which can be done up to remainders of order $ \OO(\eps^{\infty}) $ by \cref{lem: eoneo approx eones}.

	\begin{pro}[GL energy upper bound]
		\label{pro: upper bound}
		\mbox{}	\\
		Let $\Om\subset\R^2$ be any bounded simply connected domain satisfying \cref{asum: boundary 1} and \cref{asum: boundary 2}. Then, for any fixed $ 1< b < \theo^{-1} $,		 as $\eps \to 0$, it holds
		\beq
			\label{eq: upper bound}
			\gle  \leq \disp\frac{|\partial\Om| \eoneo}{\eps} - \ecorr \int_{0}^{|\partial \Omega|} \diff \ss \: \curv + \sum_{j = 1}^N E_{\mathrm{corner}, \beta_j} + o(1). 
		\eeq
	\end{pro}
	
	\begin{proof}
		As usual the upper bound is obtained by testing the GL functional on a suitable trial state. As a vector potential, we pick $ \mathbf{F} = \frac{1}{2} (-y,x) $. The order parameter on the other hand is much more involved: the idea is to recover the trial state given in \cite[Eq. (4.14)]{CR2} far from the corners and glue it to the minimizers of the effective energies in every corner. {To retain the curvature corrections, as in \cite{CR1} (see also \cref{sec: smooth energy}), we decompose the smooth part of the layer into cells of tangential length of order $ \eps $.} The order parameter is constructed in such a way that its modulus is close to $ f_{k_n}(\eps \tt) $ (see \cref{sec: 1d curvature}) in each cell, $ k_n $ being the average curvature, and to the modulus of the corner minimizer $ \psi_{\beta_j}(\rv) $ in the $j-$th corner region, respectively. The phase of $ \trial $ on the other hand is given by a gauge phase analogous to \eqref{eq: gauge phase}, but defined in terms of the vector potential $ \mathbf{F} $, plus the optimal phase $ \exp\{- i \alpha_{k_n} \ss / \eps \} $  in each  cell. An additional phase is then added to patch together such factors.
		Explicitly, we set
		\beq
			\label{eq: trial order parameter}
			\trial(\rv) : = \chi\lf( t(\rv) \ri) \cdot
			\begin{cases}
				g(s(\rv),t(\rv)) e^{- i S(s(\rv))} e^{i \phitrial(s(\rv),t(\rv))},		&	\mbox{for } \rv \in \acutet,	\\
				\trialjt\lf(\mathcal{R}^{-1}(\rv-\rv_j)/\eps \ri),									&	\mbox{for } \rv \in \corneru,
			\end{cases}
		\eeq
		where $ \mathcal{R} $ is the rotation defined in \eqref{eq: rotation} and $ \varkappa_j $ a suitable phase factor, given in \eqref{eq: varkappaj} below. Moreover, for some $ a > 0 $,
		\beq
			\label{eq: acutet}
			\acutet : = \lf\{ \rv \in \acute \: \big| \: \lf| s(\rv) - s_j \ri| \geq \leps + \eps^{a}, \forall j {= 1, \ldots, N} \ri\},
		\eeq
		is a subdomain of $ \acute $ where boundary coordinates are well defined. In $ \acute \setminus \acutet $, we take care of the transition from the smooth part of the layer to the corner region: for any given $ j {= 1, \ldots, N}  $, we concretely set 
		\beq
			\trial(\rv) : = \zeta(s) f_0 \lf(t\ri) e^{- i S(s_j \pm \leps)} +\lf(1 - \zeta(s) \ri) f_{k_{\pm}}\lf(t\ri) e^{- i S(s_j \pm \leps \pm \eps^{a})},	
		\eeq
		for any $ |s(\rv) - s_j \pm \leps| \leq \eps^{a} $ and where we have denoted $ k_{\pm} $ the average curvature in the cells $ \cell_{\pm} $ adjacent to the $j$-th corner region. The smooth cut-off function $ \zeta $ is chosen in such a way that it is positive and
		\bdm
			\zeta(s_j \pm \leps) = 1,		\qquad		\zeta(s_j \pm \leps \pm \eps^{a}) = 0,	\qquad		|\zeta^{\prime}| \leq C \eps^{-a} .
		\edm
		Furthermore,
		\begin{itemize}
			\item the smooth cut-off function $ \chi(t) $ equals $ 1 $ for any $ t \leq c_1 |\log\eps| $ and vanishes for $ t \geq |\log\eps|^2 $, so that its gradient is bounded by $ \OO(|\log\eps|^{-2}) $;
			\item the gauge phase $ \phitrial $ is given by \eqref{eq: gauge phase} with $ \fv $ in place of $ \aavm $; to leading order such a phase equals $ - \frac{1}{2} st $, so recovering part of the phase of $ \psi_{\star} $ (recall \eqref{eq: psi star});
			\item the function $ g(s,t) $ is taken directly from \cite[Eq. (4.15)]{CR2} (see also  \cite[Eq. (4.18)]{CR2}): it equals $ f_{k_n} $ in $ \cell_n $, up to a smaller correction $ \chi_n $, which allows the continuous transition from $ f_{k_n} $ to $ f_{k_{n+1}} $;
			\item similarly, the phase $ S(s) $ is given by \cite[Eqs. (4.20) \& (4.21)]{CR2}: to leading order $ S(s) = - i \alpha_{k_n} s $ in $ \cell_n $, but, as for the density, one needs to add a higher order correction taking into account the jump from $ \alpha_{k_n} $ to $ \alpha_{k_{n+1}} $;
			\item $ \trialjt(\rv') $ is close to the minimizer of the effective energy \eqref{eq: edx} in $ \cornerj $, where $ \varkappa_j : = \varkappa_{+,j} - \varkappa_{-,j} $ and
			\beq
				\label{eq: varkappaj}
				\varkappa_{\pm, j} : = - S(s_j \pm \leps) + \al_0 (s_j \pm \leps) = - \int_0^{s_j \pm \leps} \diff \xi \: \lf( \alpha_{k(\xi)} - \alO \ri) + \OO(\logi).
			\eeq
			Applying the rectification procedure described in \cref{lem: rectification}, we set
			\beq
				\label{eq: trialjt}
				\trialjt(\rv') = \psi_{\beta_j,\varkappa_j}(\rrv(\rv)),
			\eeq
			where $ \rrv $ is the diffeomorphism of \cref{lem: rectification} and $ \psi_{\beta,\varkappa} $ any minimizer of \eqref{eq: edx}.
		\end{itemize}
		
		We now sketch the main steps in the computation of the energy of $ \lf( \trial, \mathbf{F} \ri) $, which were already discussed elsewhere and focus afterwards on the new estimates:		
		\begin{itemize}
			\item since $ \curl \mathbf{F} = 1 $ in $ \Omega $, the last term in the GL energy functional \eqref{eq: glf} vanishes;
			\item the integration can be restricted to $ \anne $, where the cut-off function $ \chi $ is 1 and all the rest of the energy can be discarded thanks to the exponential decay of the modulus of $ \trial $ as well as its derivatives (inherited from $ f_0$  and $ f_k $, see \eqref{eq: fk decay 1});
			\item the gauge phase $ \phitrial $ allows to replace $ \fv $ with $ \lf( - t + \frac{1}{2} \eps k(s) t^2 \ri) \ev_s $ in $ \acutet $, as in \cref{lem: replacement smooth} up to an error of order $ \OO(\eps \logi) $;
			\item the energy bound in $ \acutet $ is taken from \cite{CR2} and stated in \cref{pro: ub smooth}:
			\bdm
				\glf \lf[ \trial, \fv; \acute \ri] \leq \disp\frac{|\partial \Omega| \eoneo}{\eps} - 2 \leps N \eoneo - \eps \ecorr \int_{0}^{\frac{|\partial\Om|}{\eps}} \diff s \: k(s) + o(1).
			\edm
		\end{itemize}
		
		Given the discussion above, it remains to compute the energy of $ \trial $ in the region $ \acute \setminus \acutet $ as well as the energy contributions of all the corner regions $ \corneru $. Let us start by considering the latter: getting rid of the diffeomorphism up to small errors, close to each corner we recover $ E_{\beta_j,\varkappa_j}(\leps,\elle) = E_{\beta_j}(\leps,\elle) + o(1) $ by \cref{lem: twisted D energies}. Summing up, we get
		\beq
			\label{eqp: corner en}
			{\sum_{j = 1}^N} E_{\beta_j}(\leps,\elle) + o(1) 
			= {\sum_{j = 1}^N} E_{\mathrm{corner}, \beta_j} +  2 N \eoneo \leps + o(1),
		\eeq
		where we have exploited the existence of the limit proven in \cref{pro: ecorn}.
		
		Finally, let us consider the energy in $ \acute \setminus \acutet $ and restrict ourselves to the interval $ [s_j - \leps - \eps^a, s_j - \leps] $: the area of the region is of order $ \eps^{2+a} |\log\eps| $ and we can thus discard all the terms involving $ f_0 $, $ f_{k_{\pm}} $ and their derivatives there up to errors of order $ \OO(\eps^a\logi) $ by \eqref{eq: fkprime bound} and \eqref{eq: fk decay 1}. The only non-trivial term to estimate is thus the kinetic energy of the cut-off function $ \zeta $: by grouping together the terms in a convenient way, one has to bound at the boundary $ s_j - \leps $ the quantity
		\bml{
			C \int_{s_j - \leps - \eps^a}^{s_j -\leps} \diff s \int_0^{\elle} \diff t \: \lf|\zeta'(s)\ri|^2 \lf| f_0 \lf(t\ri) e^{-iS(s_j - \leps)} - f_{k_{-}}\lf(t\ri) e^{- i S(s_j - \leps - \eps^{a})} \ri|^2 \\
			\leq C \eps^{-2a} \int_{s_j - \leps - \eps^a}^{s_j -\leps} \diff s \int_0^{\elle} \diff t \:  \lf[ \lf| f_0 \lf(t\ri)  - f_{k_{-}}\lf(t\ri) \ri|^2  + f_0^2(t) \lf| e^{- i S(s_j - \leps)} - e^{- i S(s_j - \leps - \eps^{a})} \ri|^2 \ri]	\\
			\leq C \lf[ \eps^{1-a} + \eps^a \logi \ri],
		}
		by the identity
		\beq
			S(s) = \al_0 s + \int_0^{s} \diff \xi \: \lf( \alpha_{k(\xi)} - \alO \ri) + \OO(\logi). 
		\eeq

		Putting together all the energy contributions, we get \eqref{eq: upper bound}.
	\end{proof}

\subsection{Order parameter}
\label{sec: order parameter}
	
As proven in \cite{CG}, any minimizing $ \glm $ is such that its modulus is suitably close in $ L^2(\anne) $ to the 1D profile $ f_0(\dist(\: \cdot \:, \partial \Omega)/\eps) $. The presence of the corners affects the estimate only at the precision one can approximate $ |\glm| $ with $ f_0 $, since the result in \cite[Thm. 1.1]{CG} is proven by neglecting the corner regions. The improved energy asymptotics of \cref{teo: gle asympt} obviously suggests that such an estimate can in fact be strengthened. Indeed, we prove here that a pointwise estimate of the difference $ |\glm| - f_0 $ holds true in the smooth part of the boundary layer.

	\begin{proof}[Proof of \cref{pro: pan}]
		The starting point is the combination of the energy upper bound \eqref{eq: upper bound}, with the stronger lower bound which can be obtained by combining the arguments of the proof of \cref{pro: lower bound} with \cref{lem: sum En}: in each cell contained in the smooth part of the boundary layer, we can retain the positive contribution appearing on the r.h.s. of \eqref{eq: est sum En}. The final outcome is the estimate
		\beq	
			\label{eqp: starting BBH}
			\sum_{{n} = 1}^{M_{\eps}} \int_{\celln}\diff s\diff t\, (1-\eps k_nt)f_n^4(1-|u_n|^2)^2 = o(1).
		\eeq
		A direct consequence is the estimate stated in \eqref{eq: refined l2 estimate} in \cref{rem: refined l2 estimate}. Furthermore, \eqref{eqp: starting BBH} is the key ingredient of a typical argument (first used in \cite{BBH2}) to deduce a pointwise estimate of $ |u_n| $ and thus $ |\glm| = f_0 |u_n| $ (see \cite[Proof of Thm. 2, Step 2]{CR2} but also the proof of \eqref{eq: point est psi} in \cref{pro: point est psi}).
		
		Instead of providing all the details, we comment only on the needed adaptations. First of all, one has to select a subdomain of $ \celln $, where a suitable lower bound on the density $ f_n $ holds true. In our case, we can restrict the analysis to the layer $ \lf\{ \dist(\rv, \partial \Omega) \leq c \eps \ri\} $, where $ f_n $ is bounded from below by a positive constant independent of $ \eps $. As a consequence, the argument of \cite[Proof of Thm. 2, Step 1]{CR2} leads to
		\beq
			\lf| \nabla \lf|u_n \ri| \ri| \leq C,	\qquad		\mbox{in } \lf\{ \dist(\rv, \partial \Omega) \leq c \eps \ri\} \cap \celln.
		\eeq
		With such a bound at disposal, the aforementioned argument applies straightforwardly in the boundary region $ \lf\{ \dist(\rv, \partial \Omega) \leq c \eps \ri\} $, leading to the pointwise estimate
		\beq
			\lf| \lf| u_n \ri|  - 1 \ri| = o(1),		\qquad		\mbox{in } \lf\{ \dist(\rv, \partial \Omega) \leq c \eps \ri\} \cap \celln,
		\eeq
		which immediately yields \eqref{eq: pan}.
	\end{proof}

\appendix

\section{One-dimensional Effective Energies}
\label{sec: 1d}

In this Appendix we recall some known results about the effective one-dimensional problems, which are known to play a role in surface superconductivity. More details can be found, e.g., in \cite{CR1,CR2,CR3,CDR}.

\subsection{Effective model on the half-line}
\label{sec: 1d disc}

The model problem describing the behavior of the order parameter along the normal direction to the boundary $ \partial \Omega $ in the surface superconductivity regime is given in first approximation by the energy
\begin{equation}
	\label{eq: fone}
	\fone_{\star,\alpha}[f] := \displaystyle\int^{+\infty}_0 \mbox{dt} \left\{ |\partial_t f|^2 + (t+\alpha)^2 f^2 -\frac{1}{2b} (2f^2-f^4)\right\},
\end{equation}
where $ t $ is the rescaled distance to the boundary and $ \alpha \in \R  $ is a parameter.

For any $ \alpha \in \R $, the functional \eqref{eq: fone} admits (see, e.g., \cite[Prop. 3.1]{CR1} or \cite[Prop. 5]{CR2}) a unique minimizer in the domain $ \onedom = \lf\{ f \in H^1(\R^{{+}}; \R) \: | \: t f(t) \in L^2(\R^{{+}}) \ri\} $ with ground state energy $ \eone_{\star,\alpha} $. The optimal profile is obtained by optimizing over $ \alpha $, as in \eqref{eq: eones}. The infimum can be easily shown to actually be a minimum, i.e., there exists an $ \as \in \R $, where the minimum is achieved. We also denote by $ \fs $ the corresponding profile, i.e., the minimizer of $ \eone_{\star,\as} ${, which satisfies
\beq
	\label{eq: fs decay}
	0 < \fs \leq C  \exp \lf\{ - \tx\frac{1}{2} \lf( t + \as \ri)^2 \ri\}.
\eeq}

\subsection{Curvature-dependent one-dimensional models}
\label{sec: 1d curvature}

It is convenient to introduce a generalization of \eqref{eq: fone}, which takes into account the effects of the boundary curvature:
\begin{equation}
	\label{eq: fonekal}
	\fonekal[f] := \displaystyle\int^{\ell}_0 \mbox{dt} (1 - \eps k t) \left\{ |\partial_t f|^2 + \potkal(t) f^2 -\frac{1}{2b} (2f^2-f^4)\right\},
\end{equation}
where $ k \in \R $ is the rescaled mean curvature, which is assumed to be constant here,
\beq
	\label{eq: potkal}
	\potkal(t) = \frac{1}{(1 - \eps k t)^2} \lf(t + \alpha - \tx\frac{1}{2} \eps k t^2 \ri)^2
\eeq
and $ \ell = \ell(\eps) \gg 1 $ is an $ \eps$-dependent quantity satisfying
\beq
	\label{eq: ell conditions}
	|\log\eps| \lesssim \ell \ll \eps^{-1}.
\eeq
For any $ \alpha \in \R $ we denote the ground state energy of \eqref{eq: fonekal} by $ \eone_{k,\alpha} $. The corresponding optimal energy is
\beq
	\label{eq: eonek}
	\eonek : = \inf_{\alpha \in \R} \eone_{k,\alpha} = \inf_{\alpha \in \R} \inf_{f \in \onedomk} \fonekal[f],
\eeq
where $ \onedomk = H^1(\iell; \R) $, with $ \iell : = [0, \ell] $, and one can prove the existence of a minimizing $ \alpha_k \in \R $, i.e., $ \eonek = \eone_{k,\alpha_k} $. The corresponding profile is then denoted by $ \fk $, which is therefore the unique minimizer of $ \eone_{k,\alpha_k} $. We also set $ \fone_{k} : = \fone_{k, \alpha_k} $, accordingly. Note that, unlike $ \fone_{\star,\alpha} $, the new energy functional $ \fonekal $ depends on $ \eps $ in the measure, in the potential $ \potkal $ and possibly in the upper extreme of the integration domain $ \ell $. 

The dependence on the curvature $ k $ of the model problem \eqref{eq: eonek} is investigated in \cite[Props. 1 \& 2]{CR2}, where it is shown that all the relevant quantities are essentially continuous in $ k $. We sum up here the main properties of the limiting functionals \eqref{eq: fonekal} and the corresponding minimizers (see \cite[Sect. 3]{CR1} and \cite[Appendix A]{CR2} for the proofs):
\begin{itemize}
	\item $ \fk $ is a smooth non-negative function monotonically decreasing for $ t \gg 1 $, such that $ \lf\| \fk \ri\|_{\infty} \leq 1 $ and
		\beq
			\label{eq: fk var}
			- \fk^{\prime\prime} + \tx\frac{\eps k}{1 - \eps k t} \fk^{\prime} + \potkal(t) \fk = \tx\frac{1}{b} \lf( 1 - \fk^2 \ri) \fk,
		\eeq
		in $ \iell $ with Neumann boundary conditions $ \fk^{\prime}(0) = \fk^{\prime}(\ell) = 0 $;
	\item for any $ 1 \leq b < \theo^{-1} $, $ \fk $ is strictly positive and satisfies the {\it optimality condition}
		\beq
			\label{eq: optimal alk}
			\int_0^{\ell} \diff t \: \frac{1}{1 - \eps k t} \lf(t + \alk - \tx\frac{1}{2} \eps k t^2 \ri) \fk^2(t) = 0;
		\eeq
	\item for any $ k \in \R $,
		\beq
			\label{eq: eonek identity}
			\eonek = - \frac{1}{2b} \int_0^{\ell} \diff t \: (1 - \eps k t) \fk^4(t).
		\eeq

	\item for any $ 1 \leq b < \theo^{-1} $ and for $ \eps $ small enough, there exist two {positive and} finite constants $ {c, c'} > 0 $, such that
		\beq
			\label{eq: fk decay 1}
			{c} \exp \lf\{ - \tx\frac{1}{2} \lf( t + \frac{1}{2} \ri)^2 \ri\} \leq \fk(t) \leq {c'} \exp \lf\{ - \tx\frac{1}{2} \lf( t + \alk \ri)^2 \ri\}.
		\eeq
	
\end{itemize}

We add to the above bounds the following estimate, which is quite similar to what is proven in \cite[Lemma 9]{CR2}:

	\begin{lem}
		\label{lem: fkprime decay}
		\mbox{}	\\
		For any $ 1 \leq b < \theo^{-1} $ and for $ \eps \ll 1 $, there exists a finite constant $ C > 0  $, such that
		\beqn
			\label{eq: fkprime decay}
			\lf|\fk^\prime(t)\ri| & \leq & C\, e^{-\frac{1}{4}t^2},	\qquad		\mbox{for any } t \in [0,\ell],	\\
			\label{eq: fkprime bound}
			\lf|\fk^\prime(t) \ri| & \leq & C t^3 \fk(t),		\qquad 	\mbox{for any } t \in [1,\ell].
		\eeqn
	\end{lem}

	\begin{proof}
		For the proof of \eqref{eq: fkprime decay} we simply notice that, by integrating the variational equation \eqref{eq: fk var} multiplied by $ \fk(t)$ in $ [t, \ell] $ and using Neumann boundary conditions, we obtain
		\[
			\big|f_k^\prime(t)\big|\leq C\int_t^{\ell}\diff\eta\, \eta^2 \: f_k(\eta)
		\]
		Then, the result is a consequence of the decay of $ \fk $ \eqref{eq: fk decay 1}.

		The proof of \eqref{eq: fkprime bound} follows along the same lines of \cite[Proof of Lemma 9]{CR2}.
	\end{proof}

As first discussed in \cite[Sect. 3]{CR1}, a key role in the study of the effective 1D models is played by the following {\it potential function}
\bml{
	\label{eq: pot function}
	F_k(t):=2\int_0^t\diff\eta\,\frac{1}{1-\eps k \eta} \lf(\eta +\alk-\tx\frac{1}{2}\eps k \eta^2 \ri) f^2_k(\eta)	\\
	= - {f_k^{\prime}}^2(t) +  (t + \alk)^2 f_k^2(t) - \tx\frac{1}{b} \lf(1 - \frac{1}{2} f_k^2(t) \ri) f_k^2(t) + \OO(\eps k).
}
which heuristically provides the energy gain of a single vortex at a distance $ \eps t $ from the boundary. Similarly, the overall energy cost of a vortex is given by the {\it cost function}
\beq
	\label{eq: cost function}
	K_k(t) : =  \fk^2(t) + F_{k}(t),
\eeq

The properties of both functions are summed up below (see \cite[Sect. 3]{CR1} and \cite[Appendix A]{CR2}): for any $ 1 < b < \theo^{-1} $ and $ k \in \R $,
\begin{itemize}
	\item $ F_k(t) \leq 0 $, for any $ t \in [0, \ell] $;
	\item $ F_k(0) = F_k(\ell) = 0 $;
	\item let $ \btik > 0  $ be such that ($ \btik $ is uniquely defined by monotonicity of $ \fk $ for large $ t $)
		\beq
			\label{eq: annk}
			\annbk : = \lf\{ t \in (0, \ell) \: \big| \: \fk(t) \geq \ell^3 \fk(\ell) \ri\} = \lf[0, \btik \ri],
		\eeq
		then,
		\beq
			\label{eq: cost positive}
			K_k(t) \geq 0, 	\qquad		\mbox{for any } t \in [0,\btik].
		\eeq
\end{itemize}

	\subsection{Effective model on an interval with $ k = 0 $}
	\label{sec: 1d no curv}
	A special case of the 1D models discussed in the previous Section is the one obtained for $ k = 0 $. It is in fact an approximation of the 1D effective energy $ \eones $ obtained by minimizing the energy on a finite interval $ [0,\ell] $, $ \ell \gg 1 $, rather than in the whole of $ \R^+ $ (see \cite[Sect. 3]{CR1}). There is indeed a unique minimizing pair $ \fol, \alpha_0 $ of $ \fone_{0,\alpha}[f] $ over $f \in H^1([0,\ell]) $ positive and $ \alpha \in \R $. Like $ \fk $, $ \fol $ solves the variational equation \eqref{eq: fk var} in the interval $ [0,\ell] $ with $ \alpha_0 $ in place of $ \alk $. In addition, $ \fol $ satisfies Neumann boundary conditions
	\beq
		\label{eq: fol nbc}
		\fol^{\prime}(0) = \fol^{\prime}(\ell) = 0.
	\eeq
	Furthermore, all the properties \eqref{eq: optimal alk} -- \eqref{eq: fkprime bound} (with $ k = 0 $) hold true for $ \fol $ as well. In particular, $ f_0 $ is monotonically decreasing for $ t \geq t_0 $, where  $ t_0 $ is the unique maximum of $ f_0 $, satisfying
	\beq
		\label{eq: t0}
		0 < t_0 \leq |\alpha_0| + \tx\frac{1}{\sqrt{b}}.
	\eeq

	\begin{lem}
		\label{lem: eoneo approx eones}
		\mbox{}	\\
		For any $ 1 < b < \theo^{-1} $, if $ \ell \gg 1 $, then
		\beq
			\label{eq: eoneno approx eones}
			\eoneo = \eones + \OO(\ell^{-\infty}).
		\eeq
	\end{lem}

	\begin{proof}
		It suffices to prove that for any finite $ \alpha \in \R $,
		\beq
			\label{eq: en fol approx eone}
			\eone_{\star,\alpha} - \eone_{0,\alpha} = \exl,
		\eeq
		which immediately implies \eqref{eq: eoneno approx eones}, since the minima of both functionals are achieved for bounded $ \alpha $ (see, e.g., \cite[Cor. 3.2 \& Lemma 3.1]{CR1}).
		However, \eqref{eq: en fol approx eone} above is a trivial consequence of the exponential decays \eqref{eq: fs decay} and \eqref{eq: fk decay 1}.
	\end{proof}

	A very important consequence of 	the properties of $ \fol $  \cite[Prop. 3.5]{CR1} is that, as for $ \fk $, if we set (recall \eqref{eq: Fol})
	\beq
		\label{eq: Kol}
		K_0(t) : =  \fol^2(t) + \Fol(t),
	\eeq
	where 
	\bdm	
		\label{eq: fol}
		\Fol(t) : = 2\int_0^t\diff\eta\, \lf(\eta + \al_0 \ri) \fol^2(\eta) = - 2 \int_t^\ell \diff \eta \, \lf(\eta + \al_0 \ri) \fol^2(\eta),
	\edm
	then, for $ 1 {<} b < \theo^{-1} $,
	\beq
		\label{eq: kol positive}
		K_0(t) \geq 0, 	\qquad		\mbox{for any } t \in \annol,
	\eeq
	with
	\beqn
		\label{eq: annol}
		\annol &: =& \lf\{ t \in (0,\ell) \: \big| \: \fol(t) \geq \ell^{3} \fol(\ell) \ri\} = \lf[0, \bar{\ell} \ri];	\\
		\label{eq: barell}
		\bar{\ell} &=& \ell + \OO(1).
	\eeqn
	In the whole interval $ [0, \ell] $, we can use \eqref{eq: pot function} in combination with \eqref{eq: fkprime bound} to estimate
	\beq
		\label{eq: fol bound}
		\lf| F_0(t) \ri| = - F_0(t) \leq C t^6 f_0^2(t),	\qquad		\forall t \in [1, \ell].
	\eeq

	\begin{remark}[Positivity of the cost function]
		\label{rem: negativity K0}
		\mbox{}	\\
		An interested reader might wonder whether $ K_0(t)  $ is in fact positive in the whole of $ [0, \ell] $. There is however a simple argument showing that this is not the case and $ \exists  t_{m} \in {\iell} $, where
		\beq
			K_0(t_m) < 0.
		\eeq
		To prove this, one first note that $ K_0 $ is convex. Hence, since $ K_0^{\prime}(0) = 2\alO f_0^2(0) < 0 $, $ K_0^{\prime}(\ell) = 2(\ell + \alO)^2 f_0^2(\ell) > 0 $, there must be a minimum point at $ 0 < t_m < \ell $. In fact, by a close inspection of the condition $ K_0^{\prime}(t_m) = 0 $, it is possible to prove that $ t_m = \ell + \OO(1) $. However, the analogue of \eqref{eq: pot function} and the criticality condition $ K^{\prime}(t_m) $ imply that
		\bdm
			K_0(t_m) = \lf(1 - \tx\frac{1}{b} + \tx\frac{1}{2b} f_0^2(t_m) \ri) f_0^2(t_m) - \ell^2f_0^2(\ell) (1 + o(1)) = - \lf(\ell^2 + o\lf(\ell^2\ri) \ri) f_0^2(\ell)< 0.
		\edm
	\end{remark}
	
	{We complete the section, by showing that the positivity in \eqref{eq: kol positive} can in fact be strengthen and promoted to a sort of coercivity of $ K_0 $.}

	{\begin{pro}[Coercivity of $ K_0 $]
		\label{pro: new positive}
		\mbox{}	\\
		For any $ 1 < b < \theo^{-1} $, if $ \ell \gg 1 $, there exists a constant $ c_b > 0 $, such that
		\beq
			\label{eq: Kol coercive}
			K_0(t) \geq c_b f_0^4(t),	\qquad \mbox{for any } t \in \annol.
		\eeq	
	\end{pro}}
	
	\begin{proof}
		{The proof idea is quite similar to the one used in the proof of \cite[Prop. 3.5]{CR1}. We provide the details for the sake of completeness. We set
		\beq
			\label{eqp: Kt positive}
			\Kt(t) : = K_0(t) - \tx\frac{1}{\ell^{4}} f_0^2(t) - c_b f_0^4(t) + \gamma_{\ell},
		\eeq
		where $ \gamma_{\ell} : =  (\ell + \alO)^2 \fol^2(\ell) - \tx\frac{1}{b} \lf(1 - \frac{1}{2} \fol^2(\ell) \ri) \fol^2(\ell) $, so that, by the identity \eqref{eq: pot function}, one gets
		\beq
			\label{eqp: Kt alternative}
			\Kt(t) = \lf[ 1 - \tx\frac{1}{\ell^{4}} - \tx\frac{1}{b} + \lf(\tx\frac{1}{2b} - c_b \ri) f_0^2(t) \ri] f_0^2(t) - {\fol^{\prime}}^2(t) +  (t + \alO)^2 \fol^2(t).
		\eeq
		Now, if one can prove that $ \Kt \geq 0 $ in $ [0,\ell] $, the result  then easily follows because
		\bdm 
			\min_{\annol} \lf( \tx\frac{1}{\ell^{4}} f_0^2(\ell) - \gamma_{\ell} \ri) \geq \min_{\annol} \lf[  \tx\frac{1}{\ell^{4}} f_0^2(t)  - \ell^2 \fol^2(\ell)  \ri] \geq 0,
		\edm
		for $ \ell $ large enough.}
		
		{Let us then address the positivity of \eqref{eqp: Kt positive}: at the boundary of the interval we have
		\beq
			\Kt(0) \geq \lf( 1 - \tx\frac{1}{\ell^4} - c_b f_0^2(0) \ri) f_0^2(0) > 0,	\qquad		\Kt(\ell) \geq \lf(1 - \tx\frac{1}{\ell^4} - c_b f_0(\ell)^2 \ri) f_0^2(\ell) > 0,
		\eeq
		if $ c_b < 1/f_0^2(0) $ and $ \ell \gg 1 $. Hence, the function can become negative only in the interior of $ I_{\ell} $, so let us look for its minimum points $ t_m $, which must satisfy $ \Kt^{\prime}(t_m) = 0 $, yielding
		\bdm
			\lf(1 - \tx\frac{1}{\ell^4} - 4 c_b f_0^2(t_m) \ri) f_0^{\prime}(t_m) = - \lf( t_m + \alO \ri) f_0(t_m).
		\edm
		If we now take $ c_b < 1/(4 \lf\| f_0 \ri\|_{\infty}^2) $, we can solve the above identity w.r.t. to $ f'(t_m) $ and plug it into \eqref{eqp: Kt alternative}, so obtaining
		\bml{
			\Kt(t_m) = \lf[ 1 - \tx\frac{1}{\ell^{4}} - \tx\frac{1}{b} - \frac{2 \lf(\tx\frac{1}{\ell^4} + 4c_b f_0^2(t_m)\ri)\lf(1 - \tx\frac{1}{2 \ell^4} - 2c_b f_0^2(t_m) \ri)}{\lf(1 - \frac{1}{\ell^4} - 4 c_b f_0(t_m) \ri)^2} (t_m + \alO)^2 \ri] \fol^2(t_m)	\\
			+  \lf(\tx\frac{1}{2b} - c_b \ri) f_0^4(t_m) \geq \lf[ 1 - \tx\frac{C}{\ell^2} - \tx\frac{1}{b} - 8 c_b t_m^2 f_0^2(t_m) \ri] \fol^2(t_m) \geq 0,
		}
		if $ \ell \gg 1 $ and $ c_b $ is taken small enough so that $ c_b \leq \frac{1}{2b} $ and $ c_b < \lf(1 - \tx\frac{1}{b} \ri)/ ( 8 \lf\| t f_0 \ri\|_{\infty}^2 ) $ (recall \eqref{eq: fk decay 1}). Putting all the conditions together, we see that the result is proven if we take
		\bdm
			c_b < \min \lf\{ \tx\frac{b - 1}{8 b \lf\| t f_0 \ri\|_{\infty}^2}, \frac{1}{2b}, \frac{1}{4 \lf\| f_0 \ri\|_{\infty}^2}, \frac{1}{f_0^2(0)} \ri\}.
		\edm}
	\end{proof}

\section{Technical Estimates}
\label{sec: technical}

In this Appendix we collect several technical estimates, which are used in the paper. Throughout this Appendix, $ \Omega $ will denote a bounded and simply connected domain $\Omega\subset \mathbb{R}^{2}$ satisfying \cref{asum: boundary 1} and \cref{asum: boundary 2}.
We recall that for any bounded domain $\Om\subset\mathbb{R}^2$ with locally Lipschitz boundary, all the usual Sobolev embeddings hold true \cite[Thm. 5.4]{Ad}. In particular, in what follows, we often use that, given a domain $ \Omega $ with the strong local Lipschitz property (see \cite[Def. 4.5]{Ad}), for all $p\in [2,\infty)$ and for all $\alpha\in [0,1)$,
\beq\label{eq: Sobolev emb}
	H^{1}(\Om)\hookrightarrow L^p(\Om), \qquad H^2(\Om)\hookrightarrow W^{1,p}(\Om), \qquad W^{2,p}(\Omega)\hookrightarrow C^{0,\alpha}(\overline{\Om}),
\eeq
where $ C^{0,\alpha} $ stands for the space of H\"{o}lder continuous functions with exponent $ \alpha $.
We also note that the diamagnetic inequality is verified in a piecewise smooth domain as well, i.e., for every $\aav\in L^2_{\mathrm{loc}}(\mathbb{R}^2; \R^2)$, $\psi\in L^2_{\mathrm{loc}}(\mathbb{R}^2)$ such that $(\nabla + i \aav)\psi\in L^2_{\mathrm{loc}}(\mathbb{R}^2)$, one has
\beq\label{eq: diamagnetic}
	|\nabla|\psi||\leq |(\nabla + i \aav)\psi|,	\qquad		\mbox{for a.e. }  \rv \in \Omega.
\eeq

\subsection{Minimization of the GL energy} 
\label{sec: minimization}
For the sake of completeness, we briefly discuss the minimization of the GL functional in domains with Lipschitz boundary. The material is mostly taken from \cite{FH1} (see in particular \cite[Chpt. 15 \& Sect. D.2.3]{FH1}).

As proven in \cite[Thm. 15.3.1]{FH1}, there exists a minimizing pair $ (\glm, \aavm) $ for $ \glfk[\psi,\aav] $, such that $ (\psi, \aav - \mathbf{F}) \in H^1(\Omega) \times W^{1,2}_{0,0}(\R^2) ${, where $ W^{1,2}_{0,0}(\R^2) $ is a suitable Sobolev space properly defined in \cite[Eq. (D.12)]{FH1} and $ \mathbf{F} $ given in \eqref{eq: fv})}. In addition, we may fix the gauge in such a way that
\beq
	\label{eq: no div}
	\nabla \cdot \aavm = 0.
\eeq
This determines the potential up to an additive constant, which can be chosen so that
\beq
	\label{eq: A est curl}
	\lf\| \aavm - \mathbf{F} \ri\|_{H^1(\Omega; \R^2)} \leq C \lf\| \curl \aavm - 1 \ri\|_{L^2(\R^2)},
\eeq
which in turn implies \cite[Lemma 15.3.2]{FH1} that $ \curl \lf( \aavm - \mathbf{F} \ri) = 0 $ or, equivalently,
\beq
	\curl \aavm = 1,		\qquad		\mbox{in } \R^2 \setminus \Omega.
\eeq
Hence, when we evaluate $ {\glfe} $ on the minimizing configuration, we may restrict the integration domain in the last term in \eqref{eq: glf} to $ \Omega $. 

Finally, any critical point $ (\psi, \aav) $ of $ \glf $ and in particular the minimizing pair $ (\glm, \aavm) $ satisfies the GL variational equations
\beq
	\label{eq: GL eqs}
	\begin{cases}
		- \lf( \nabla + i \tx\frac{\aav}{\eps^2} \ri)^2 \psi = \tx\frac{1}{\eps^2} \lf(1 - |\psi|^2 \ri) \psi,		& \mbox{in } \Omega,		\\
		- \tx\frac{1}{\eps^2} \nabla^{\perp} \curl \aav = \jv_{\aav}[\psi] \one_{\Omega},							& \mbox{in } \R^2,	\\
		\nuv \cdot \lf( \nabla + i  \tx\frac{\aav}{\eps^2} \ri) \psi = 0,									& \mbox{on } \partial \Omega,
	\end{cases}
\eeq
where we have denoted by $ \jv_{\aav} $ the current
\beq
	\label{eq: supercurrent}
	\jv_{\aav}[\psi] : = \tx\frac{i}{2} \lf[ \psi \lf( \nabla - i \tx\frac{\aav}{\eps^2} \ri) \psi^* - \psi^* \lf(\nabla + i \tx\frac{\aav}{\eps^2} \ri) \psi \ri] = \Im \lf( \psi^* \lf(\nabla + i \tx\frac{\aav}{\eps^2} \ri) \psi \ri).
\eeq
{Any minimizing pair is smooth in the interior of $ \Omega $ and continuous at the boundary. More precisely, for any $ \tilde\Omega \subset \Omega $ with $ \partial \tilde\Omega \subset \Omega^{\circ} $ smooth, and for any $ \alpha \in [0,1) $
\beqn
		\glm &\in C^{\infty}(\widetilde\Omega),		\qquad		
	\aavm &\in C^{\infty}(\widetilde\Omega; \R^2);\\
		\glm &\in C^{0,\alpha}(\overline{\Omega}),	\qquad		\aavm &\in C^{0,\alpha}(\overline\Omega; \R^2),		\label{eq: regularity}
\eeqn
as it can be seen by applying standard arguments in elliptic theory (see, e.g., \cite{G}).}

\subsection{Elliptic estimates}
\label{sec: elliptic}

We now state useful estimates valid for any critical point of $ \glf $. The following bounds are direct consequences of \eqref{eq: GL eqs} \cite[Chpts. 10, 11 \& 15]{FH1}:
\beq\label{eq: est infty}
	\lf\|\psi\ri\|_{L^\infty(\Om)}\leq 1.
\eeq
\beq
	\label{eq: vector potential p}
	\lf\| \aavm - \mathbf{F} \ri\|_{L^{p}(\Omega)} \leq C \eps \lf\| \psi \ri\|_{L^2(\Om)}\|\psi\|_{L^4(\Om)}.
\eeq
	
	We also have a quantitative estimate of the magnetic gradient of $ \psi $, which is however limited by the presence of corners at the boundary.
	
	\begin{lem}
		\label{lem: est gradient infty}
		\mbox{}	\\
		Let $ \psi, \aav $ solve \eqref{eq: GL eqs} and let
		\beq
			\Omega_{\eps} : = \lf\{ \rv \in \Omega \: \big| \: \dist(\rv, \Sigma) \geq \eps \ri\},
		\eeq
		then
		\beq
			\label{eq: est grad infty}
			\lf\| \lf(\nabla + i \frac{\aav}{\eps^2} \ri) \psi \ri\|_{L^\infty(\Om_{\eps})} \leq \frac{C}{\eps}.
		\eeq
	\end{lem}
	
	\begin{proof}
		{The result can be deduced from the equations \eqref{eq: GL eqs} and, in particular, the first one, applying in a suitable way, e.g., \cite[Lemma A.1]{BBH1}.}
	\end{proof}
	
	{The counterpart of \eqref{eq: est grad infty} for any minimizer $ \psi $ of the corner problems \eqref{eq: glecorn doms} and \eqref{eq: en n fv corner} reads
	\beq
		\label{eq: est grad infty corner 1}
		\lf\| \lf( \nabla + i \av \ri) \psi \ri\|_{L^{\infty}(\lf\{ \lf|s(\rv)\ri| \geq 1 \ri\})} = \OO(1),
	\eeq
	and combining it with, e.g., \eqref{eq: agmon 2} proven in next \cref{app:Agmon}, we also get			\beq
		\label{eq: est grad infty corner 2}
		\lf\| \nabla \psi \ri\|_{L^{\infty}(\lf\{ \lf|s(\rv)\ri| \geq 1 \ri\})} = \OO(1).
	\eeq}

	\subsection{Agmon estimates}\label{app:Agmon}
		Another typical key tool in the study of the GL theory is the estimate of the decay properties ({\it Agmon estimates}) of any solution $(\psi,\aav)$ of the GL variational equations \eqref{eq: GL eqs} in the surface superconductivity regime, i.e., when the intensity of the applied magnetic field is such that $h_{\mathrm{ex}} > \Hcc $. The result is in fact inherited from the linear problem associated to the GL energy, i.e., a magnetic Schr\"{o}dinger operator, and does not exploit the nonlinearity. The presence of corners does not influence the exponential decay of the order parameter away from the boundary \cite[Sect. 15.3.1]{FH1}. More precisely, for any $b>1$ and for any $ (\psi, \aav) $ solving \eqref{eq: GL eqs} \cite[Thm. 4.4]{BNF},
\beq
	\label{eq: agmon}
		\displaystyle\int_\Om \diff\textbf{r}\;\exp\lf\{\tx\frac{c(b) \: \mathrm{dist} (\textbf{r},\partial\Om)}{\eps}\ri\} \lf\{|\psi|^2+\eps^2 \lf|\lf(\nabla+i \tx\frac{\textbf{A}}{\eps^2}\ri)\psi\ri|^2\ri\} \;  = \OO(\eps),
\eeq
	where $c(b)>0$, for $ b > 1 $, is independent of $ \eps $. When $ b \to 1^+ $, the above bound becomes non-optimal because of the vanishing of $ c(b) $ and one can in fact prove other estimates showing a power law decay of $ \psi $ \cite{FK}. Similarly, in presence of corners, the result might not be optimal for $b>\Theta_0^{-1}$: assuming that there is at least one angle $\beta$ along the boundary such that $\mu(\beta) < \Theta_0$, one can prove \cite[Thm. 1.6]{BNF} a stronger decay w.r.t. the distance from that corner. Here, $ \mu(\beta) $ stands for the ground state energy of the magnetic Schr\"{o}dinger operator in an infinite wedge of opening angle $ \beta $ with unit magnetic field.	
	
	{The translation of \eqref{eq: agmon} in the setting of \cref{sec: strip}, i.e., a GL functional with fixed parameter $ \eps = 1 $ in a finite strip $ R(\ell,L) $ is as follows:}
	
	{\begin{lem}
		\label{lem: agmon strip}
		\mbox{}	\\
		Let $ \psi $ solve \eqref{eq: var eq strip} and satisfy the boundary conditions alternatively in \eqref{eq: strip domd}, \eqref{eq: strip domn} or \eqref{eq: neumann conditions} in $ R(\ell,L) $ with $ \ell,L > 0 $. Then, for any $ b > 1 $, there exists a constant $ c(b) > 0 $, such that
		\beq
			\label{eq: agmon strip}
			\int_{\rect} \diff s \diff t \; e^{c(b) \: t} \lf\{|\psi|^2+ \lf|\lf(\nabla - i  t \ev_s \ri) \psi\ri|^2\ri\}  = \OO(L).
		\eeq
	\end{lem}}

In the paper, we use Agmon estimates also for the corner effective problem. We discuss here such an extension to the setting of the effective problem formulated in \eqref{eq: ecorn} and discussed in \cref{sec: corner energy}.

	\begin{lem}
		\label{lem: agmon 1}
		\mbox{}	\\
		Let $ \corner $ be the region given in \cref{fig: corner}, with $ L, \ell \gg 1$ and $ L \lesssim \ell^{{a}} $, for some $ {a} > 1 $. Let also $ \psi $ be a solution of \eqref{eq: var eq corner}, with $ b >  1 $. Then, there exists a constant $ c(b) > 0 $, such that
		\beq
			\label{eq: agmon 1}
			\int_{\corner} \diff\textbf{r}\; e^{c(b) \: \dist (\textbf{r},\bdo)} \lf\{| \psi|^2+ \lf|\lf(\nabla+i  \av \ri) \psi\ri|^2\ri\}  = \OO(L).
		\eeq
	\end{lem}

The above result is a simple adaptation of \eqref{eq: agmon} to the effective problem in $ \corner $. The only difference is that the magnetic potential $ \av $ is given and not a minimizer of the energy.
Before discussing its proof, however, we first state a technical lemma{, which follows from a standard inequality for the magnetic gradient and the equation solved by $ \psi $.}

	\begin{lem}
		\mbox{}	\\
		Let $ \corner $ be the region given in \cref{fig: corner}, with $ L, \ell \gg 1$ and $ L \lesssim \ell^{{a}} $, for some $ {a} > 1 $. Let also $ \psi $ be a solution of \eqref{eq: var eq corner} and let $ \xi $ be a smooth real function. Then, for any set $ S \subset \corner $ with Lipschitz boundary,
		\bml{
			\label{eq: magnetic bound}
			\int_{S} \diff \rv \lf\{ |\psi|^2 \lf( \nabla \xi \ri)^2 + \tx\frac{1}{b} \lf| \xi \psi \ri|^2 \lf( 1 - |\psi|^2 \ri) \ri\} \geq \int_{S} \diff \rv \: \curl(\av) \: |\xi \psi|^2 \\
			- \int_{\partial S} \diff x \lf\{ \tx\frac{1}{2} \xi^2 \nuv \cdot \nabla |\psi|^2 + \tav \cdot \jv_{\av}[\xi \psi] \ri\},
		}	
		where $ \tav, \nuv $ stand for the tangential and normal unit vectors to $ \partial S $, respectively.
	\end{lem}
	
	\begin{proof}
		We start by integrating the {following} trivial bound (see, e.g., \cite[Lemma 3.2]{CLR}) for any $ u $ weakly differentiable and $ \av \in L^{\infty} $ (we set $ \av : = (a_1, a_2) $)
		\bmln{
			\lf| \lf( \nabla + i \av \ri) u \ri|^2 = \lf| \lf( \partial_1 + i a_1 - i(\partial_2 + i a_2) \ri) u \ri|^2 - \curl \: \jv[u] - \av \cdot \nabla^{\perp} |u|^2	\\
			\geq - \curl \: \jv[u] - \av \cdot \nabla^{\perp} |u|^2,
		}
		which yields, taking $ u = \xi \psi $,
		\beq
			\int_{S} \diff \rv \: \lf| \lf( \nabla + i \av \ri) \xi \psi \ri|^2 \geq  \int_{S} \diff \rv \: \curl(\av) \: | \xi \psi|^2 
			- \int_{\partial S} \diff x \lf\{  \tav \cdot \jv[\xi \psi] + \tav \cdot \av |\xi \psi|^2 \ri\},
		\eeq
		after an integration by parts of the last term and the use of Stokes theorem. Note that the last two terms can be combined to reconstruct the magnetic current $ \jv_{\av} $.
		To complete the proof it suffices to use the equation \eqref{eq: var eq corner} to compute the term on the l.h.s.. The additional boundary term in \eqref{eq: magnetic bound} is produced by the integration by parts of the cross product term $ \xi \nabla \xi \cdot \nabla |\psi|^2 $ to reconstruct the term $ \psi^* \Delta \psi + \mathrm{h.c.} $ of the variational equation.
	\end{proof}

	\begin{proof}[Proof of \cref{lem: agmon 1}]
		As anticipated the result is a simple adaptation of \eqref{eq: agmon} (see \cite[Proof of Thm. 12.2.1]{FH1}). The key ingredient is the inequality \eqref{eq: magnetic bound}, applied to $ S = \corner $, together with the following choice of the function $ \xi $: 
		\beq
			{\xi}(\rv) = {\xi}(t(\rv)) = e^{a t(\rv)} f(t),
		\eeq
		with the function $ f $ such that $ |f^{\prime}| \leq C $ and
		\[
			f =
			\begin{cases}
				 1,		&	\mathrm{for }\,t \in [1,+\infty],\\
				 0,		&	\mathrm{for }\,t \in \lf[0,\frac{1}{2}\ri].
			\end{cases}
		\]

		We first estimate the boundary terms appearing in \eqref{eq: magnetic bound}:
		\beq
			\label{eq: bd terms 1}
			\int_{\partial \corner} \diff x \: \xi^2 \nuv \cdot \nabla |\psi|^2 = \OO(1),
		\eeq
		because $ \xi = 0 $ on $ \bdo $, 
		\beq
			\label{eq: bdi est 1}
			\lf| \xi^2 \nuv \cdot \nabla |\psi|^2 \ri|  \leq C e^{2 a \ell} f_0(\ell) = \OO(\ell^{-\infty}),		\qquad		\mbox{on } \bdi,
		\eeq
		and 
		\beq
			\label{eq: bdbd est 1}
			\int_{\bdbd} \diff x \: \lf| \xi^2 \nuv \cdot \nabla |\psi|^2 \ri| \leq C \int_0^{\ell} \diff t \: e^{2 a t} f_0(t) \leq C,		
		\eeq
		where we have used the boundary conditions on $ \psi $, the exponential decay of $ f_0 $ \eqref{eq: fk decay 1} and the estimate \eqref{eq: est grad infty}, which yields
		\beq
			\label{eq: grad est infty corner}
			\lf| \nabla |\psi| \ri| \leq \lf| \lf( \nabla + i \av \ri) \psi \ri| \leq C,	\qquad		\mbox{for } \dist\lf(\rv, \rv_0\ri) \geq 1,
		\eeq
		{$ \rv_0 $ being the position of the corner.} Similarly,
		 \beq
		 	\label{eq: bd terms 2}
			\int_{\partial \corner} \diff x \: \tav \cdot \jv_\av[\xi \psi] = \OO(1),
		\eeq
		thanks to the vanishing at $0 $ of  $\xi $ and the bounds
		\beq
			\label{eq: bdi est 2}
			\int_{\bdi} \diff x \: \lf| \tav \cdot \jv_\av[\xi \psi] \ri| \leq C \int_{\bdi} \diff x \: e^{2 a \ell} f_0^2(\ell) = \OO(\ell^{-\infty}),
		\eeq
		\beq
			\label{eq: bdbd est 2}
			\int_{\bdbd} \diff x \: \lf| \tav \cdot \jv_\av[\xi \psi] \ri| \leq C \int_0^{\ell} \diff t \: e^{2 a t} f_0(t) \leq C,
		\eeq
		as in \eqref{eq: bdi est 1} and \eqref{eq: bdbd est 1}, respectively.

		The rest of the proof is identical to \cite[Proof of Thm. 12.2.1]{FH1}: the estimates \eqref{eq: bd terms 1} and \eqref{eq: bd terms 2} above together with \eqref{eq: magnetic bound} imply
		\beq
			\label{eq: var eq and curl}
			\lf(1- \tx\frac{1}{b}\ri) \lf\|\xi \psi \ri\|_{L^2(\corner)}^2  \leq \int_{\corner} \diff \rv \: \lf( \nabla \xi \ri) ^2|\psi|^2 + \OO(1).
		\eeq		
		Noticing now that
		\beq
			\lf| \xi^{\prime} \ri|^2 \leq 2 (1+\epsilon) a^2 f^2 e^{2a t} + \lf(1 + \tx\frac{1}{\epsilon}\ri) {f^\prime}^2 e^{2a t}
			\leq 2 (1+\epsilon) a^2 f^2 e^{2a t} + C(\epsilon) e^{2a t},
		\eeq
		we conclude that
		\beq
			\label{eqp: agmon final bound}
			\lf(1- \tx\frac{1}{b} - 2 (1+\epsilon) a^2 \ri) \int_{t(\rv) \geq \frac{1}{2}} \diff \rv \: e^{2 a t(\rv)} \lf| \psi \ri|^2  \leq C \int_{t(\rv) \leq 1} \diff \rv \: |\psi|^2 + \OO(1),
		\eeq
		and since we can always find $ \epsilon > 0  $ and $ a(\epsilon) > 0 $ so that the factor on the l.h.s. of the above expression is positive, we obtain the result for the order parameter. The estimate of the magnetic gradient however follows using \eqref{eq: var eq corner} once more and the bound just proven.
	\end{proof}
	
We complete the discussion of the decaying properties of the order parameter with a refined version of the estimate proven in \cref{lem: agmon 1}: we consider a solution of the differential equation \eqref{eq: var eq corner} and show that, in a subdomain of tangential length of order $ \OO(1) $, the r.h.s. of \eqref{eq: agmon 1} is $ \OO(1) $ as well. In order to state a more precise bound there, we identify two model domains, i.e., a rectangle $ S_{\mathrm{strip}} $ of tangential side length $ \OO(1) $ far from the corner and the region close to it $ S_{\mathrm{corner}} $. More precisely, we set
\beq
	\label{eq: S strip}
	S_{\mathrm{strip}}  : = \lf\{ \rv \in \corner \: \big| \:  {\bar{s}}_1 \leq s(\rv) \leq {\bar{s}}_2 \ri\},		\qquad		{\bar{s}}_2 - {\bar{s}}_1 \leq C,
\eeq
and either $ {\bar{s}}_1 \geq \ell / \tan\lf(\beta/2\ri) $ or $ {\bar{s}}_2 \leq - \ell / \tan\lf(\beta/2\ri) $, which ensures that in $ S_{\mathrm{strip}} $ we can use the coordinates $ (s,t) $ and it corresponds to $ [{\bar{s}}_1, {\bar{s}}_2] \times [0, \ell] $. The other region $  S_{\mathrm{corner}} $ is 
\beq
	\label{eq: S corner}
	 {S_{\mathrm{corner}} : = \lf\{ \rv \in \corner \: \big| \: \bar{s}_1 \leq \dist(\rv, \rv_0) \leq \bar{s}_2 \ri\},	\qquad	 {\bar{s}}_2 - {\bar{s}}_1 \leq C,}
\eeq
and $ 1 \leq \bar{s}_1 $, $ {\bar{s}}_2 \leq C \ell  $, i.e., it is a wedge-like domain where boundary coordinates can not be used globally.

	\begin{lem}
		\label{lem: agmon 2}
		\mbox{}\\
		Let $ S_{\sharp} $ be one of the two domains defined in \eqref{eq: S strip} and \eqref{eq: S corner}. Let also $ \psi $ be a solution of \eqref{eq: var eq corner}, with $ b >  1 $. Then, there exists a constant $ c(b) > 0  $, such that 
		\beq
			\label{eq: agmon 2}
			\int_{S_{\sharp}} \diff\textbf{r}\; e^{c(b) \: \dist (\textbf{r},\bdo)} \lf\{|\psi|^2+ \lf|\lf(\nabla+i  \av \ri) \psi\ri|^2\ri\}  = \OO(1).
		\eeq
	\end{lem}
	
	\begin{proof}
		The proof is identical to the one of \cref{lem: agmon 1}, with the only difference due to the estimate of boundary terms. Exploiting \eqref{eq: grad est infty corner} and the other properties of $ \psi $ and $ f_0 $, it is however easy to show that those terms provide contributions of order $ \OO(1) $, as well as the r.h.s. of \eqref{eqp: agmon final bound}, which leads to the result.
		A short comment is in order for regions close to the corner, where the pointwise bound \eqref{eq: grad est infty corner} might fail: there one can always arrange the domain $ S $  in such a way that the boundary $ \partial S $ is far enough from $ \bdbd $ (still at a distance of order $ 1 $ from the corner) so that \eqref{eq:  grad est infty corner} applies, while on $ \partial S \cap \bdo $, the gradient estimate is not used.
	\end{proof}

We finally provide a simple bound which is a direct consequence of  \eqref{eq: agmon 2}.

	\begin{lem}
		\label{lem: exp estimate}
		\mbox{}\\
		Let $ \psi $ be a solution of \eqref{eq: var eq corner}, with $ b >  1 $. Then, there exists a finite constant $ C  $, such that
		\beq
			\label{eq: decay corner}
			\lf| \psi(\rv) \ri| \leq C e^{- \frac{1}{2} c(b) \dist(\rv, \bdo)},
		\eeq
		where $ c(b) $ is the constant appearing in \eqref{eq: agmon 2}.
	\end{lem}

	\begin{proof}
		The result is proven by contradiction. Suppose that there was a point $ \bar{\rv} \in \corner $, with $ \dist(\bar{\rv}, \bdo) \geq 1 $ and $ \dist(\bar{\rv}, \rv_0) \geq 1 $, {so that}
		\beq
			\lf| \psi({\bar\rv}) \ri| e^{\frac{1}{2} c(b) \dist({\bar\rv}, \bdo)} > C_0,
		\eeq
		for some given $ C_0 > 0 $. Then, thanks to the pointwise bound \eqref{eq: grad est infty corner}, we can always construct a square $ Q $ of unit side length containing $ \bar{\rv} $, such that 
		\beq
			\lf| \psi(\rv) \ri| e^{\frac{1}{2} c(b) \dist(\rv, \bdo)} \geq \tx\frac{1}{2} C_0, 	\qquad	 \mbox{in } Q.
		\eeq
		We are here assuming that $ C_0 $ is large enough, so that
		\bdm
			\inf_{\rv \in Q} \lf( \lf| \psi(\rv) \ri| e^{\frac{1}{2} c(b) \dist(\rv, \bdo)} \ri) \geq C_0 - \sqrt{2} \lf( \lf\| \nabla \lf| \psi \ri| \ri\|_{\infty} + \tx\frac{1}{2} c(b) \ri) \geq \tx\frac{1}{2} C_0.
		\edm
		 Hence, 
		\bdm
			\int_{Q} \diff \rv \: \lf| \psi(\rv) \ri|^2 e^{c(b) \dist(\rv, \bdo)} \geq \tx\frac{1}{4} C^2_0,
		\edm
		which contradicts \eqref{eq: agmon 2}, if $ C_0 $ is large enough, since $ Q $ is fixed.
	\end{proof}

	{Reformulating the above result for the variational problems in the strip considered in \cref{sec: strip} yields the pointwise estimates
	\beq
		\label{eq: point agmon strip}
		\lf| \psi(s,t) \ri| \leq C e^{- \frac{1}{2} c(b) t},	
	\eeq 
	for any $ \psi $ solving \eqref{eq: var eq strip} and where $ c(b) $ is the same constant appearing in \cref{lem: agmon strip}.}

\section{Local Energy Estimates}\label{sec: smooth energy}

In this Section we sum up the salient points of the energy estimate in the smooth part of the boundary layer. Thanks to Agmon estimates (see \cref{app:Agmon}), we can restrict our analysis to the boundary layer \eqref{eq: anne}, i.e.,
\[
	\anne = \lf\{ \rv \in \Omega \: \big| \: \dist\lf(\rv, \partial \Omega\ri) \leq \eps \elle \ri\},
\]
but here we will focus on its smooth component defined in \eqref{eq: acut}: 
\[
	\acut = \lf( \lf[0, s_{1} - \leps \ri] \cup \lf[ s_1 + \leps, s_2 - \leps \ri] \cup \cdots \cup \lf[ s_N + \leps, \tx\frac{|\partial \Omega|}{\eps} \ri] \ri) \times [0, c_1 |\log\eps|],
\]
where $s_j$, $j=1, \ldots, N$ is the tangential coordinate of the $j-$th vertex. By \cref{lem: replacement smooth}, we can take as starting point of our analysis the effective functional introduced in \eqref{eq: gep}:
\[
	\mathcal{G}_{\eps}[\psi,\acut] = \int_{\acut}\diff s\diff t\, (1-\eps k(s)t)\bigg\{|\partial_t\psi|^2 + \tx\frac{1}{(1-\eps k(s)t)^2}|(\partial_s - it)\psi|^2
	-\frac{1}{2b}(2|\psi|^2 -|\psi|^4)\bigg\}
\]
and its ground state energy
\begin{equation}\label{eq: inf G cut}
	G_{\acut}:=\inf_{\psi \in H_{\mathrm{per}}^1(\acut)}\mathcal{G}_{\eps}[\psi,\acut],
\end{equation}
where $ H_{\mathrm{per}}^1(\acut) : = \{ \psi \in H^1(\acut) \: | \: \psi(0,t) = \psi(|\partial \Omega|/\eps, t), \forall t \in [0, \ell] \} $. 
We also denote
\beq
	\label{eq: ismooth}
	\smooth := \bigcup_{{j = 1}}^{{N}} \lf[s_j + \leps, s_{j+1} - \leps \ri],
\eeq
{with the identification $ s_{N+1} = s_1 + |\partial \Omega|/\eps $}. The material presented in this Section is essentially taken from \cite{CR2} (see, in particular \cite[Lemmas 3, 6 and 7]{CR2}), but  an important difference in the lower bound is given by the presence of holes in the boundary layer $ \acut $, where the corner regions have been removed. 
The key tool in the strategy is the decomposition of $ \acut $ into cells:
\beq
	\acut = \bigcup_{n=1}^{M_\eps} \celln,		\qquad
	{\celln} := [{\sigma}_n,{\sigma}_{n+1}]\times[0, {c_1} |\log\eps|],
\eeq
with $ |{\sigma}_{n+1}-{\sigma}_n| \propto 1 $ and $ M_\eps \propto \lf|\smooth\ri|/\eps $. We then approximate the curvature $k(s)$ of the boundary in each cell by its mean value
\beq
	k_n:= \int_{{\sigma}_n}^{{\sigma}_{n+1}}\diff s\: k(s)
\eeq
and set for short $ \alpha_n := \alpha_{k_n} $, $ f_n(t) :=f_{k_n}(t) $ (recall the notation of \cref{sec: 1d curvature}).

	\begin{pro}[Upper bound to $ G_{\acut} $]
		\label{pro: ub smooth}\mbox{}\\
		For any fixed $1<b<\Theta_0^{-1}$, as $\eps \rightarrow 0$, it holds
		\beq
			\label{eq: ub smooth}
			G_{\acut} \leq \disp\frac{|\partial \Omega| \eoneo}{\eps} - 2 \leps N \eoneo - \eps \ecorr \int_{0}^{\frac{|\partial\Om|}{\eps}} \diff s \: k(s) + o(1).
		\eeq
	\end{pro}

	\begin{proof}
	See \cite[Sect. 4.1]{CR2}.
	\end{proof}
	
	We now complement \eqref{eq: ub smooth} with a matching lower bound. As already pointed out, the proposition below is the analogue of \cite[Lemma 6]{CR2} but the effect of the holes in the smooth part of the domain now becomes apparent in the additional boundary terms appearing on the r.h.s. of \eqref{eq: lb smooth}. Those terms are matched in \cite[Proof of Lemma 7, Step 2]{CR2} by the corresponding boundary contributions coming from the cells which are missing in the present setting. 

	\begin{pro}[Lower bound]
		\label{pro: lb smooth}
		\mbox{}\\
		Let $ \psi(s,t) \in H^1(\acut) $ be a function enjoying the same properties as $ \glm(\rv(s,t)) $. Then, for any $1<b<\Theta_0^{-1}$, as $\eps \rightarrow 0$, it holds
		\bml{
			\label{eq: lb smooth}
			\mathcal{G}_{\acut}[\psi]  \geq \disp\frac{|\partial\Om| \eoneo}{\eps} - 2 \leps N \eoneo - \eps \ecorr \int_{0}^{|\partial\Om|} \diff s \: k(s) \\
			 -  \sum_{j {=1}}^{{N}} \int_0^{{c_1}|\log\eps|}\diff t\, \lf.\frac{F_0(t)}{f_0^2(t)} j_t \lf[\psi(s,t)\ri] \ri|_{s=s_j - L_\eps}^{s= s_j +L_\eps}  + o(1).
		}
	\end{pro}
	
	\begin{proof}
		The starting point {is} the very same splitting performed in \cite[Lemma 6]{CR2}, which is analogous to what we did in the proof of \cref{pro: strip}: in each cell $ \celln $, we set 
		\beq
			\label{eq: splitting cell psi}
			\psi(\rv(s,t)) =: u_n(s,t)f_{n}(t)e^{-i\alpha_{n}s},
		\eeq
		where $u_n$ plays the same role as $u$ in the decoupling \eqref{eq: psi splitting strip}. Such a splitting procedure allows to extract from each cell the desired energy, i.e.,
		\beq
			\label{eq: splitting cell}
			E^{\mathrm{1D}}_{k_n} \lf( {\sigma}_{n+1}-{\sigma}_n \ri) + \mathcal{E}_n[u_n],
		\eeq
		where the reduced energies are
		\bml{
			\label{eq: En}
			\E_n[u] : = \int_{{\sigma}_n}^{{\sigma}_{n+1}}\diff s \int_0^{\ell} \diff t \, (1 - \eps k_n t) f_n^2 \lf\{|\partial_t u|^2 + \tx\frac{1}{(1 - \eps k_n t)^2} |\partial_s u|^2 -  2 b_n(t) j_s[u] \ri.	\\ 
			\lf. +\tx\frac{1}{2b} f_n^2 (1-|u|^2)^2\ri\},
		}
		with $ b_n(t) = \tx\frac{1}{(1 - \eps k_n t)^2} (t + \alpha_n - \frac{1}{2} \eps k_n t^2) $.
		By \cite[Lemma 2.1]{CR3}, the first terms of \eqref{eq: splitting cell} above sum up to
		\beq\label{eq: smooth energy}
			\disp\frac{|\partial\Om| \eoneo}{\eps} - 2 \leps N \eoneo - \eps \ecorr \int_{\smooth} \diff s \: k(s) + o(1).
		\eeq

		If $ 1 < b < \theo^{-1} $, the reduced functionals $ \E_n[u_n] $ can be proven to be positive \cite[Lemma 7]{CR2} and  can thus be dropped from the lower estimate, again up to small errors. Here, however, the major difference with \cite{CR2} occurs: the positivity of $ \E_n[u_n] $ is proven in \cite[Lemma 7]{CR2} via an integration by parts and exploits the pointwise positivity of the cost function $ K_k $ (see \eqref{eq: cost function} and \eqref{eq: cost positive}), but the estimate of the boundary terms emerging from the integration has to be adjusted. Such terms have the form
		\bdm
			- \int_0^{{c_1}|\log\eps|}\diff t\, \frac{F_0(t)}{f_0^2(t)} j_t\lf[\psi(s,t) \ri] \bigg\vert_{s={\sigma}_{n}}^{{s =} {\sigma}_{n+1}}.
		\edm
		The sum of all the terms is shown in \cite[Lemma 7]{CR2} to be small, but this requires (see  \cite[Step 2 and eq. (5.33)]{CR2}) to pair the term coming from one cell at $ {\sigma}_n $ with the one generated in the adjacent cell again at $ {\sigma}_n $. In our setting, due to the absence of corner regions in $ \acut $, some boundary terms are missing. Such terms are precisely given by
		\beq\label{eq: bd terms splitting}
			\sum_{j\in \Sigma}\int_0^{{c_1}|\log\eps|}\diff t\, \frac{F_0(t)}{f_0^2(t)} j_t[\psi(s,t)]\bigg\vert_{s=s_j-L_\eps}^{s=s_j+ L_\eps},
		\eeq
		and have to be added and subtracted to apply \cite[Lemma 7]{CR2}, leading to \eqref{eq: lb smooth}.
	\end{proof}
	
		Note that in both the upper and lower bounds \eqref{eq: ub smooth} and \eqref{eq: lb smooth}, we can easily replace the integral over $ \smooth $ with the integral over the whole boundary, since
		\beq\label{eq: integral curvature}
			\eps \ecorr \int_{\smooth} \diff s \: k(s) = \ecorr \int_{0}^{|\partial \Omega|} \diff \ss \: \curv + \OO(\eps |\log\eps|),
		\eeq
		by the boundedness of the curvature.

		We conclude the Section with an important corollary of the above lower bound, which will be used to prove a uniform estimate of  $|\psi^{\mathrm{GL}}|$ in the smooth part of the layer.
		
		\begin{lem}[Lower bound on the reduced energies]
			\label{lem: sum En}
			\mbox{}	\\
			Let $ u_n $ be defined in \eqref{eq: splitting cell psi} and $ \E_n $ be given by \eqref{eq: En}. Then, if $ 1 < b < \theo^{-1} $, as $ \eps \to 0 $,
			\begin{multline}\label{eq: est sum En}
				\sum_{n=1}^{M_\eps}\mathcal{E}_n[u_n] \geq |\log\eps|^{-4}\sum_{n=1}^{M_\eps}\int_{\celln}\diff s\diff t\, (1-\eps k_nt)f_n^2\lf[|\partial_t u_n|^2 + \tx\frac{1}{(1-\eps k_n t)^2}|\partial_s u_n|^2\ri]
				\\
				+\frac{1}{2b}\sum_{n=1}^{M_\eps}\int_{\celln}\diff s\diff t\, (1-\eps k_nt)f_n^4(1-|u_n|^2)^2 + o(1).
			\end{multline}
		\end{lem}
		
		\begin{proof}
			See \cite[Proof of Lemma 7]{CR2}.
		\end{proof}

\end{document}